\documentclass[12pt]{article}

\usepackage[hmargin=1in,a4paper, height=21.5cm]{geometry}
\usepackage{hyperref}
\usepackage{verbatim}
\usepackage{amsmath,amsthm,amssymb,amsfonts,amscd,amstext}
\usepackage{graphicx}
\usepackage[table]{xcolor}
\usepackage{multirow} 
\usepackage{times} 
\usepackage{mathtools}
\usepackage{url}
\usepackage{esint}

\usepackage{sidecap}
\sidecaptionvpos{figure}{c} 

\usepackage{xr}
\usepackage{setspace}
\usepackage[ruled, vlined]{algorithm2e}

\usepackage{tikz}
\usetikzlibrary{decorations.markings}
\newcommand{\dhorline}[3][0]{%
    \tikz[baseline]{\path[decoration={markings,
      mark=between positions 0 and 1 step 4*#3
      with {\node[fill, circle, minimum width=#3, inner sep=0pt, anchor=south west] {};}},postaction={decorate}]  (0,#1) -- ++(#2,0);}}
\newcommand{\dottedhfill}{\noindent\makebox{\dhorline{0.93\linewidth}{0.8pt}}}

\newenvironment{alglist}
  {\begin{list}
    {}
    {\setlength{\labelwidth}{1em}
     \setlength{\labelsep}{0em}
     \setlength{\itemsep}{2pt}
     \setlength{\leftmargin}{1cm}
     \setlength{\rightmargin}{1.2cm}
     \setlength{\itemindent}{0em} 
     
    }
  }
{\end{list}}

\usepackage{rotating}

\usepackage[font=scriptsize]{caption}

\global\long\def\BC{Barber and Cand\`es}
\global\long\def\KR{Katsevich and Ramdas}
\global\long\def\subTDC{_{\text{TDC}}}
\global\long\def\subSSS{_{\text{SSS+}}}
\global\long\def\subAS{_{\text{AS}}}
\global\long\def\supKR{^{\text{KR}}}
\global\long\def\supUB{^{\text{UB}}}
\global\long\def\supSB{^{\text{SB}}}
\global\long\def\supUSB{^{\text{SB/UB}}}
\global\long\def\supUSKR{^{\text{SB/UB/KR}}}
\global\long\def\dmax{\text{$d_{\max}$}}
\global\long\def\til#1{\tilde{#1}}%
\global\long\def\alp{\alpha}
\global\long\def\gam{\gamma}

\global\long\def\lam{\lambda}
\global\long\def\vrp{\varphi}%
\global\long\def\sig{\sigma}%
\global\long\def\N{\mathbb{N}}%
\global\long\def\supsec{Supplementary Section}%
\global\long\def\supfig{Supplementary Figure}%

\theoremstyle{plain}
\newtheorem{Theorem}{Theorem}

\newtheorem{Definition}{Definition}

\newtheorem{Lemma}{Lemma}

\newtheorem{Corollary}[Lemma]{Corollary}

\theoremstyle{remark}
\newtheorem{Remark}{Remark}

\begin{document}

\title{Bounding the FDP in competition-based control of the FDR}

\author{Arya Ebadi$^1$, Dong Luo$^1$, Jack Freestone$^1$, William Stafford Noble$^2$, Uri Keich$^1$\\
$^1$School of Mathematics and Statistics F07, University of Sydney\\
$^2$Departments of Genome Sciences and of Computer Science and Engineering\\
University of Washington\\
}


\maketitle

\begin{abstract}
	Competition-based approach to controlling the false discovery rate (FDR) recently rose to prominence when, generalizing
	it to sequential hypothesis testing, \BC\ used it
	as part of their knockoff-filter. Control of the FDR implies that the, arguably more important, false discovery
	proportion is only controlled in an average sense. We present TDC-SB and TDC-UB that provide upper prediction bounds on the FDP
	in the list of discoveries generated when controlling the FDR using competition. Using simulated and real data we show that,
	overall, our new procedures offer significantly tighter upper bounds than ones obtained using the recently published
	approach of \KR, even when the latter is further improved using the interpolation concept of Goeman et al.
\end{abstract}


\section{Introduction}

In the multiple testing problem we consider $m$ null hypotheses $H_1,\dots,H_m$. When $m$ is large we
typically try to maximize the number of discoveries (rejections) while controlling the false discovery rate (FDR).
Introduced by \cite{benjamini:controlling}, the FDR is the expectation of the false discovery proportion (FDP):
$Q=V/(R\vee1)$, where $R$ is the number of rejected hypotheses, of which $V$ are true nulls/falsely rejected, and $x\vee y=\max\{x,y\}$.

Canonical FDR-controlling procedures (e.g., \cite{benjamini:controlling,storey:direct}) reject the hypotheses associated with the $\tau$ most significant
p-values, where $\tau$ is determined from the p-values, $p_1,\dots p_m$, each computed assuming the corresponding $H_i$ holds
(is a true null).

By definition, controlling the FDR only guarantees that we control the FDP in an averaged sense: the expectation is taken
over the true null hypotheses conditional on the false nulls. However, in practice, the scientist typically only has a single
sample and they care more about the actual FDP in their list of discoveries than about its theoretical average.

To address this need, two types of methods for controlling the FDP have been developed. The first, often referred to as
false discovery exceedance (FDX) control, aims at probabilistically controlling the FDP at a specific level, i.e., for a
desired levels of FDP, $\alpha$, and confidence $1-\gamma$, the procedure guarantees that $P(Q>\alpha)\le\gamma$ \cite{guo:generalized}.
The second type offers simultaneous probabilistic bounds for all FDP levels: assuming the hypotheses are ordered by decreasing
significance ($p_1\le p_2\le\dots \le p_m$), the procedure computes $\bar Q_k$, an upper prediction band (also called a confidence
envelope) on $Q_k$, the FDP among the top $k$ hypotheses, so that $P(\exists k:Q_k>\bar Q_k)\le\gamma$ \cite{genovese:exceedance,katsevich:simultaneous}.

Controlling the FDR at level $\alpha$ is very similar to controlling the FDP at the same $\alpha$ with confidence of
$1-\gamma=1/2$ (mean vs. median control). Clearly, increasing the confidence level to, say 0.95, will generally decrease the
number of discoveries. It is therefore understandable that scientists would generally prefer using FDR control than the
stricter FDP control. In such cases, the aforementioned upper prediction bands still offer scientists a way to probabilistically
bound the level of the FDP in their FDR-controlled list of discoveries. Indeed, if the FDR-controlling procedure rejects the
top $\tau$ hypotheses, then $\bar Q_\tau$ provides an upper prediction bound on $Q_\tau$ with $1-\gamma$ confidence.

Our goal in this paper is to provide such bounds on the FDP while using the competition-based approach to controlling the FDR.
Known as target-decoy competition (TDC), the latter was commonly practiced in mass spectrometry analysis since first introduced
by \cite{elias:target}. It later gained significant popularity in the statistics and machine learning community following \BC'
introduction of the knockoff filter for selecting variables in linear regression while controlling the FDR \cite{barber:controlling}.

The strength of that approach is that instead of relying on the p-values $p_i$ that we canonically associate with the
observed/target scores $Z_i$, it only requires as little as a single competing decoy/knockoff null score $\tilde Z_i$ for each
$Z_i$. The decoys are constructed so that for a true null $H_i$, $\tilde Z_i$ and $Z_i$ are exchangeable (e.g., independently drawn from the
same null distribution), independently of all other hypotheses.

While the decoy scores can be used to compute so-called ``one-bit p-values'' (1/2 or 1, depending on whether $Z_i>\tilde Z_i$),
TDC does not directly utilize those p-values as such. Instead, it competes each target score with its corresponding decoy to
define a winning score $W_i=\max\{Z_i,\tilde Z_i\}$, and a label $L_i=\pm1$ indicating whether the higher/winning score was the
target ($L_i=1$) or the decoy ($L_i=-1$). TDC next sorts the scores $W_i$ in decreasing order (higher scores are assumed to be
more significant), and reports the target wins in the top $k\subTDC$ winning scores. The cutoff $k\subTDC$ is determined
using the number of decoy wins to gauge the number of false discoveries (true null target wins), and hence to estimate and
control the FDR. The rationale here is that for a true null $H_i$ ($i\in N$), $L_i$ is equally likely to be $\pm1$, hence the observed
number of decoy wins can be used to conservatively gauge the unobserved number of true null target wins \cite{he:theoretical,barber:controlling}.

Integral to \BC’ knockoff filter, is their Selective SequentialStep+ (SSS+) procedure that generalizes TDC to control the FDR
in the context of sequential hypothesis testing. In that case the hypotheses are preordered, and associated with each is a
p-value $p_i$, such that the true null p-values are iid that stochastically dominate the uniform $(0,1)$ distribution,
independently of the false nulls. Like TDC, SSS+ reports the ``target wins'' (defined as $p_i\le c$, where $c\in(0,1)$
is a parameter) in the top $k\subSSS$ hypotheses. The cutoff, $k\subSSS$, is determined similarly to TDC by using the (scaled) number of ``decoy wins'' ($p_i>c$)
to gauge the number of true null target wins.

SSS+ was further generalized by Lei and Fithian’s Adaptive SeqStep (AS), which introduced a separate parameter $\lambda\in[c,1)$ to define
the decoy wins \cite{lei:power}. Specifically, AS defines $L_i=1$ (``target win'') if $p_i\le c$, and $L_i=-1$ (``decoy win'') if $p_i>
c$. Ties (when $p_i \in (c, \lambda]$) can be randomly broken or discarded, and note that the target-decoy terminology is borrowed here from TDC. Let
$D_k=\sum_{i=1}^k1_{\{L_i=-1\}}$ (the number of decoy wins in the top $k$ hypotheses), and $T_k=\sum_{i=1}^k1_{\{L_i=1\}}$ (the
corresponding number of target wins). AS reports the \emph{target wins} among top $k\subAS$ hypotheses, where the cutoff is
determined by
\begin{equation}
	\label{eq:AS}
	k\subAS = \max\left\{k\in\left\{1,\dots,m\right\} : \frac{D_k+1}{T_k}\frac{c}{1-\lambda} \le\alpha\right\} \vee 0. 
\end{equation}
Lei and Fithian proved that AS controls the FDR in the finite
sample setting of the sequential hypotheses testing, that is, with $V_k$ denoting the number of true null target wins in the
top $k$ hypotheses, and $Q_k=V_k/(T_k\vee1)$, $E(Q_k)\le\alpha$. Moreover, both SSS+ and TDC can be viewed as special cases of
AS: choosing $c=\lambda$ AS reduces to SSS+, and with the scores ordered by $W_i$, $c=\lambda=1/2$, and the p-values
being the aforementioned one-bit p-values, it reduces to TDC.

We recently presented FDP-SD, a procedure that probabilistically controls the FDP in this setup, at a prescribed level $\alpha$
and confidence $1-\gamma$. Our complementary goal here is to
develop upper prediction bands that can be used to
provide an upper prediction bound on the list of discoveries returned by the FDR-controlling AS/SSS+/TDC.

\KR\ recently proposed one such band in Theorem 2 of \cite{katsevich:simultaneous} that we refer to as the KR band. Specifically, they provide an
upper prediction band for $V_k$:
\begin{equation}
	\label{eq:V-KR}
	\bar V_k\supKR := \frac{-\log(\gamma)}{\log\left(1 + \frac{1-\gamma^{B}}{B}\right)}\left(1+B\cdot D_k\right) , 
\end{equation} 
where $B := c/(1-\lambda)$ (note that our $\gamma$ is their $\alpha$). Such a band immediately provides simultaneous bounds
on $Q_k$, the FDP among the target wins in the top $k$ hypotheses:
\begin{equation}
	\label{eq:Q-KR}
	\bar Q_k\supKR = \frac{\bar V_k\supKR}{T_k\vee 1}\wedge 1 .
\end{equation} 

Goeman et al. recently developed a general machinery to improve p-value based bands and specifically demonstrated it on a
p-value based analog of the above KR band \cite{goeman:only}. Their approach involves several steps of which the first,
interpolation, can be trivially implemented to improve the above KR band, by ensuring the implied guaranteed number of
discoveries does not decrease. However, we could not see a practical way to implement the remaining steps in our competition
setup. Instead, we propose two alternative bands, the ``uniform band'' (UB) and the ``standardized band'' (SB). We use simulated and
real data from multiple domains to demonstrate that our bands typically yield
tighter bounds on the FDP than the ones provided by the improved, interpolated KR band.

\section{Upper prediction bands for false discoveries}
\label{sec:UPBs}

We begin with
constructing upper prediction bands that allow us to bound the FDP in any list of top target wins.
Initially deterministic, our bands will later be stochastic.
\begin{Definition}
	A sequence of real numbers $\{\xi_i\}_{i\in I}$ is a $1-\gamma$ upper prediction band for the random process $Z_i$ with $i\in I\subset\N$
	if $P(\exists i\in I\,:\,Z_i > \xi_i)\le \gamma$.
\end{Definition}

Recall that 
$Q_k = V_k/(T_k\vee1)$ is the FDP among the $T_k$ target discoveries in the top $k$ scores
(see Supplementary Table \ref{suptable:definitions} for notations).
$T_k$ is known, so to bound $Q_k$ it suffices to bound $V_k$ by analyzing $D_k$, and
as $D_k$ only varies with a decoy win, we consider $V_k$ only for $k$'s corresponding to decoy wins. That is, we define
$i_d=\min\{i\,:\,D_i=d\}\wedge m$ (index of the $d$th decoy win, or $m$ if $D_m\le d$), and we consider the process
$N_d\coloneqq V_{i_d}$ (number of true null target wins before the $d$th decoy win).

Let $\Delta\coloneqq\{1,2,\dots,d_{\max}\}$, where $d_{\max}\le m$.
The process $\{N_d\,:\,d\in\Delta\}$ itself is not particularly amenable to theoretical analysis because
(a) some of the decoy wins can be attributed to false null hypotheses, and (b) the number of true null hypotheses is finite (why this
is a problem will become clear below).

Therefore, we first create a new process $\{\til N_d\,:\,d\in\Delta\}$ which is defined the same way as $N_d$, except that we modify the data
by (a) forcing each false null to correspond to a target win, and (b) adding infinitely many independent (virtual) true null hypotheses.
More precisely, we define a sequence of pairs $\{(\til W_i,\til L_i)\}_{i=1}^\infty$ as $\til W_i=W_i$ for $i=1,\dots,m$, $\til L_i=L_i$ for $i\in N$,
$\til L_i=1$ for $i\in\{1,\dots,m\}\setminus N$, and for $j=i+m$ we define $\til W_{m+i}\coloneqq W_m-i$ and $\til L_{m+i}\coloneqq {(-1)}^{B_i}$
where $B_i$ are iid Bernoulli($R$) RVs, where
\begin{equation}
	\label{eq:R}
R = \frac{1 - \lambda} {c + 1 - \lambda} = \frac1{1 + B} .
\end{equation}
\begin{Remark}
	\label{rem:wlog_equal}
Throughout this analysis we assume that for a true null $H_i$, $P(L_i = 1) = c$ and $P(L_i = -1) = 1-\lam$
independently of all the winning scores and all other labels. These equalities are stronger than the weak inequalities
$P(L_i = 1) \le c$ and $P(L_i = -1) \ge 1-\lam$ AS requires. We will argue in Remark \ref{rem:wlog_equal2} below that there is no loss of generality in this.
\end{Remark}
\begin{Lemma}
	\label{lemma:dominate}
	$N_d\le \til N_d$ for all $d\in\Delta$.
\end{Lemma}
\begin{proof}
Recall that $i_d$
is the index of the $d$th decoy win in our original sequence $(W_i,L_i)$, and let
$\til i_d=\inf\{i\,:\,\til D_i=d\}$ 
be the analogous index relative to the modified sequence $(\til W_i,\til L_i)$. 
Both modifications
increase $\til i_d$ relative to $i_d$ therefore $i_d\le \til i_d$. It follows that
\[
N_d = V_{i_d} \le \til V_{i_d} \le \til V_{\til i_d} = \til N_d .
\]
\end{proof}

\begin{Corollary}
	\label{cor:transitive_upper prediction band}
	If $\{\xi_d\}_{d\in\Delta}$ is a $1-\gamma$ upper prediction band for $\{\til N_d\,:\,d\in\Delta\}$ then it is also an
	upper prediction band for $\{N_d\,:\,d\in\Delta\}$.
\end{Corollary}

It follows from the last corollary that we can construct an upper prediction band for $\{N_d\,:\,d\in\Delta\}$ by constructing one
for a process with the same distribution as that of $\til N_d$.
We next introduce such an equivalent process that is described slightly more succinctly: consider a sequence $B_i$ of iid Bernoulli($R$) RVs
and define $X_i\coloneqq\sum_{j=1}^i B_j$, $Y_i\coloneqq i-X_i$, and $i_d=\inf\{i\,:\,Y_i=d\}$.
Clearly, the distribution of $U_d\coloneqq X_{i_d}$ is the same as that of $\til N_d$.


Looking for an upper prediction band for the process $\{U_d\,:\,d\in\Delta\}$, we note that the marginal distribution varies with the time/index $d$:
$U_d$ has a negative binomial $NB(d,R)$ distribution so $E(U_d)=Bd$ and $V(U_d)=B(1+B)d$.
Clearly, a band that seeks to effectively bound the process simultaneously for all times, needs to take that variation into account.
One way to do that is to standardize the process, i.e., consider instead
\begin{equation}
\label{eq:std_proc}
    \hat U_d := \frac{U_d - E(U_d)}{\sqrt{V(U_d)}} = \frac{U_d - Bd}{\sqrt{B(1+B)d}} ,
\end{equation}
which is analogous to what Ge and Li used when defining a similar upper prediction band where p-values are available~\cite{ge:control}.
Specifically, with $z_\Delta^{1-\gamma}$ denoting the $1-\gamma$ quantile of $\max_{d\in\Delta}\hat U_d$, we have
\begin{align*}
  P(\exists d\in\Delta\,:\,U_d > z_\Delta^{1-\gamma}\sqrt{B(1+B)d}+Bd)  &= P(\exists d\in\Delta\,:\,\hat U_d > z_\Delta^{1-\gamma}) \\
	&= P(\max_{d\in\Delta}\hat U_d > z_\Delta^{1-\gamma}) \le \gamma .
\end{align*}
It follows that our  ``standardized band'' (SB) defined by
\begin{equation}
\label{eq:std_band}
\xi_d\supSB = \xi_d\supSB(\Delta,\gamma) \coloneqq z_\Delta^{1-\gamma}\sqrt{B(1+B)d}+Bd \qquad d\in\Delta,
\end{equation}
is a $1-\gamma$ upper prediction band for $\{U_d\,:\,d\in\Delta\}$.

Standardization offers one approach to accounting for the variation in the marginal distribution of $U_d$.
Alternatively, we can apply a probability-based normalization to the process. Specifically, we consider
the process $\til U_d\coloneqq G_d(U_d)$, where with
$F_{NB(d,R)}$ denoting the CDF of a $NB(d,R)$ RV
\[
G_d(k)\coloneqq P(NB(d,R)\ge k) = 1 - F_{NB(d,R)}(k-1) .
\] 

To facilitate the construction of an upper prediction band, any normalizing transformation $\vrp(U_d)$ should allow us to efficiently
(a) compute the probability that $\vrp(U_d)$ exceeds (in the case of $\hat U_d$) 
or falls below (in the case of $\til U_d$) a \emph{fixed} bound $u$, and (b) compute an ``inverse'', i.e., find
$U_d$ given $\vrp(U_d)$.
We can address (a) by using Monte Carlo simulations for both $\hat U_d$ and $\til U_d$ (more on that below). 
As for (b), the inverse of $\hat U_d$ is obvious ($U_d=\sqrt{B(1+B)d}\hat U_d+Bd$), and the following lemma shows
how to compute it for $\til U_d$.

\begin{Lemma}
	\label{lemma:inverse_tilU}
	Let $\beta_d^{1-u}$ denote the $1-u$ quantile of the $NB(d,R)$ distribution. Then $\til U_d\le u$ if and only if $U_d>\beta_d^{1-u}$.
\end{Lemma}
\begin{proof}
$U_d$ is an integer valued RV and by its definition, the quantile
\begin{equation}
	\label{eq:UB_quantile}
\beta_d^{1-u} = \min\{i\,:\,F_{NB(d,R)}(i) \ge 1-u\} ,
\end{equation}
is also an integer,   so $U_d>\beta_d^{1-u}$ if and only if $U_d - 1 \ge \beta_d^{1-u}$. Hence,
\[
U_d>\beta_d^{1-u} \iff F_{NB(d,R)}(U_d - 1) \ge 1-u \iff u \ge 1 - F_{NB(d,R)}(U_d - 1) .
\]
Once again using that $U_d\in\N$ we get
\[
U_d>\beta_d^{1-u} \iff G_d(U_d)=P(NB(d,R) \ge U_d) \le u  .
\]
\end{proof}

\begin{Corollary}
	\label{cor:quantiles_upper prediction band}
	\begin{equation}
	P\left(\min_{d\in\Delta}\til U_d \le u \right) = P\left(\exists d\in\Delta\,:\,\til U_d \le u \right) 
		= P\left(\exists d\in\Delta\,:\,U_d > \beta_d^{1-u} \right) .
	\end{equation}
\end{Corollary}
\begin{proof}
The first equality is obvious while the second follows from Lemma \ref{lemma:inverse_tilU}.
\end{proof}

What remains is to select $u$ so the left hand side above is as small as desired. Let
\begin{equation}
	\label{eq:c_star}
u_\gamma = u_\gamma(\Delta) \coloneqq \max_{u\in R_\Delta} P\left(\min_{d\in\Delta} \til U_d \le u \right)\le \gamma ,
\end{equation}
where $R_\Delta=\{P(NB(d,R) \ge k)\,:\,k\in\N,d\in\Delta\}$ is the range of values $\{\til U_d\}_{d\in\Delta}$ can attain (it
is easy to see that maximum is attained).
Note that $u_\gamma$ varies not just with $\gamma$ but also with the set $\Delta=\{1,2,\dots,d_{\max}\}$.
Nevertheless, $u_\gamma(\Delta)$ can be precomputed for most ``typical values.''
In \supsec~\ref{supsec:c_star_simulations} we show how we efficiently use Monte Carlo simulations to find numerical approximations of $u_\gamma(\Delta)$
for typical values of $\gamma$ and $d_{\max}$.
We can now define our ``uniform (confidence) band'' (UB):
\begin{Theorem}
	\label{thm:ub}
	Let $u=u_\gamma(\Delta)$ be as in \eqref{eq:c_star} and let $\beta_d^{1-u}$ denote the $1-u$ quantile of the $NB(d,R)$ distribution as in \eqref{eq:UB_quantile}.
	The sequence $\{\xi_d\supUB\coloneqq\beta_d^{1-u}\}_{d\in\Delta}$ is a $1-\gamma$ upper prediction band for $\{N_d\,:\,d\in\Delta\}$.
\end{Theorem}
\begin{proof}
By Corollary \ref{cor:transitive_upper prediction band} and the fact that the distributions of $\{\til N_d\,:\,d\in\Delta\}$ and $\{U_d\,:\,d\in\Delta\}$
are the same it suffices to show that $\{\beta_d^{1-u}\}_{d\in\Delta}$ is a $1-\gamma$ upper prediction band for $\{ U_d\,:\,d\in\Delta\}$.
The latter follows immediately from Corollary \ref{cor:quantiles_upper prediction band} and \eqref{eq:c_star}.
\end{proof}
Note that any $u$ for which $P\left(\min_{d\in\Delta} \til U_d \le u \right)\le \gamma$ can be used similarly to $u_\gamma(\Delta)$
to create a band as above. The reason we take the maximum in \eqref{eq:c_star} is to maximize the power by getting the tightest band we can.

If $\{\xi_{d}=\xi_{d}^{d_{\max}}\}_{d\in[d_{\max}]}$, where $[k]=\{1,2,\dots,k\}$, is
a $1-\gamma$ upper prediction band for $\{N_{d}\,:\,d\in[d_{\max}]\}$, then we can readily convert it
into an upper prediction band for the number of false target discoveries $\{V_i\,:\,i\in[m]\}$:
\begin{equation}
	\label{eq:band_V}
\bar V_i = \begin{cases} \xi_{D_i} & L_i=-1 \text{ and } D_i\le d_{\max} \\
						\xi_{D_i+1} & L_i=1 \text{ and } D_i+1\le d_{\max} \\
						T_i & \text{otherwise}
	\end{cases} .
\end{equation}
This follows immediately from (i) $V_i\le T_i$ trivially holds, and
(ii) $V_i=N_{D_i}$ if $L_i=-1$, and $V_i\le N_{D_i+1}$ if $D_i+1\le d_{\max}$.

As in \eqref{eq:Q-KR}, a band for $\{V_i\,:\,i\in[m]\}$ is trivially turned into a band for the FDP:
\begin{equation}
	\label{eq:band_Q}
\bar Q_i =  \frac{\bar V_i}{T_i\vee 1}\wedge 1 .
\end{equation}

As pointed out by \KR, upper prediction bands can be used to bound the FDP in any list of top target wins, even when that list is
generated using post hoc analysis, e.g., one that considers multiple rejection thresholds~\cite{katsevich:simultaneous}.
In particular, to control the FDP at level $\alpha$ and confidence $\gamma$ we report the target wins in the top $k_0$ scores,
where
\begin{equation}
\label{eq:FDP_control}
k_0=\max\left\{ i\le m\,:\,L_i=1\,,\,\bar Q_i\le\alpha\right\} \vee 0 ,
\end{equation}
where we included the condition $L_i=1$ because we only report target wins.
Similarly, with $k\subAS$ as in \eqref{eq:AS}, $\bar Q_{k\subAS}$ provides an upper prediction bound on the list of discoveries reported by AS
(we define $\bar Q_0=0$).

\begin{Remark}
	\label{rem:wlog_equal2}
Going back to Remark \ref{rem:wlog_equal}, note that our analysis assumes that for a true null $H_i$, $P(L_i=-1|L_i\ne0)=R$,
where $R$ is defined in \eqref{eq:R}. In the general case $P(L_i=-1|L_i\ne0)\ge R$, which would only make our analysis conservative so our
upper prediction bands are still valid.
\end{Remark}

\section{Setting \dmax}

Comparing \eqref{eq:V-KR} and \eqref{eq:band_V} it is clear that in contrast with the KR band, the effectiveness of our bands depend on \dmax. 
In particular, the penalty for setting it too small is substantial as $\bar Q\supUSB_i=1$ if $D_i>\dmax$. 
At the same time, it is easy to see that our bands are monotone increasing with \dmax\ so our goal is to set it as low as practically possible.

Figure \ref{fig:var_dmax} demonstrates the effect of \dmax\ on both bands $\xi\supUSB$. Specifically, while both increase with \dmax, 
the effect is nowhere as significant as when \dmax\ is exceeded. Thus, for both controlling the FDP, as well as for bounding the FDP 
when controlling the FDR, we choose \dmax\ so as to guarantee it is never exceeded.

\begin{figure}
\centering
\begin{tabular}{ll}
\hspace{-8ex}
\includegraphics[width=3in]{{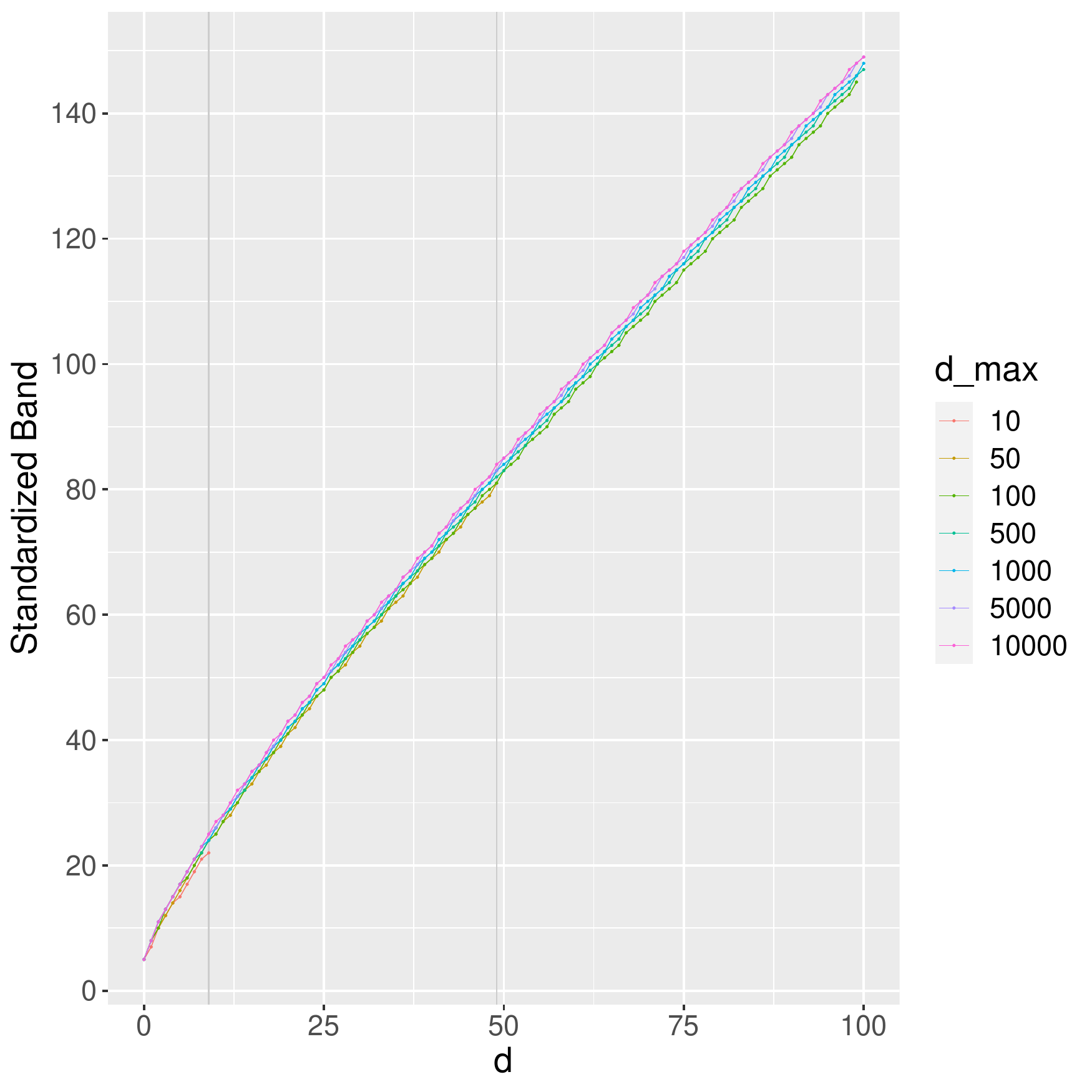}} &
\includegraphics[width=3in]{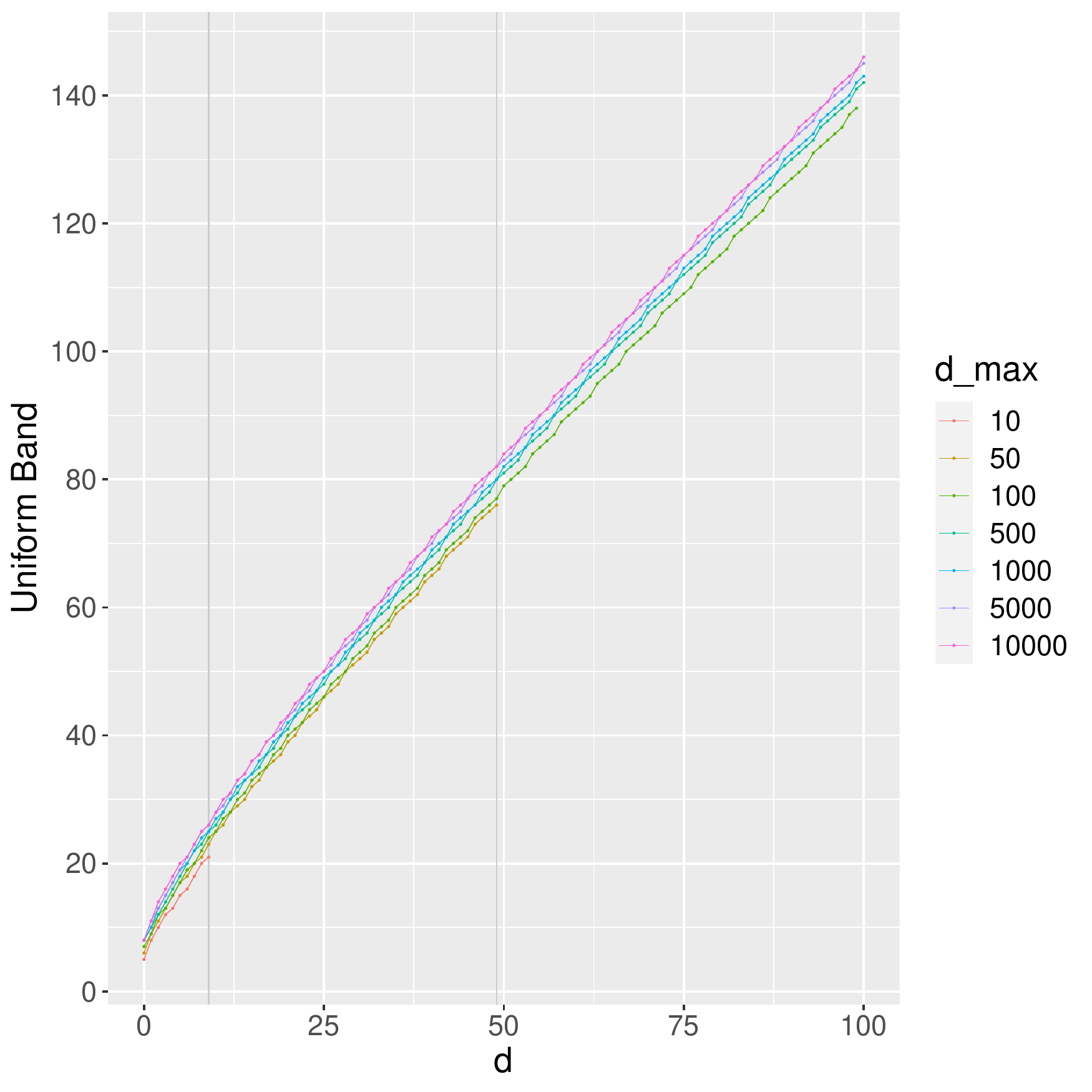}
\end{tabular}
\caption{\textbf{The effect of varying \dmax\ on the uniform and standardized bands.}\label{fig:var_dmax} For $d \in \{1, \ldots, 100\}$
we computed the value of $\xi\supSB$ \eqref{eq:std_band}, the standardized band (left), and $\xi\supUB$ (Theorem \ref{thm:ub}), the uniform band (right)
using $\dmax \in \{10, 50, 100, 500, 1000, 5000, 10000\}$.
Note that in practice the bands are undefined for any $d$ such that $d + 1 \geq \dmax$ (indicated by the grey vertical lines for $\dmax \in \{10, 50\}$).
}
\end{figure}

Starting with controlling the FDP, let $\{\xi_d^{d_0}\}_{d\in[d_{0}]}$
be an upper prediction band for $N_d$ with \dmax\ set to $d_{0}\in[m]$. Define $\xi_{0}^0\coloneqq0$ and let
\begin{equation}
d_{\infty}=\max\{d_{0}\in\{0,1,\dots,m\}\,:\,\xi_{d_{0}}^{d_0}/(m-d_{0}+1)\leq\alpha\}.\label{eq:d_max_FDP_control}
\end{equation}
\begin{Lemma}
	Let $\bar Q_k^{d_0}$ denote an FDP band defined by \eqref{eq:band_V} and \eqref{eq:band_Q} with $\dmax=d_0$.
	If there exist $d_{0}\in[m]$ and $k\in[m]$ such that $L_k=1$ and $\bar Q_k^{d_0}\le\alpha$ then $D_{k}+1\le d_{\infty}$.
\end{Lemma}
\begin{proof}
Because $L_k=1$, $\xi_{D_{k}+1}^{d_{0}}/T_{k}=\bar Q_k^{d_0}\le\alpha$ and $D_{k}+1\le d_{0}$ (otherwise $\bar Q_k^{d_0}=1$),
and because for a fixed $d\le d_{0}$, $\xi_{d}^{d_{0}}$ increases
with $d_{0}$ we have
\[
\frac{\xi_{D_{k}+1}^{D_{k}+1}}{m-(D_{k}+1)+1} \le \frac{\xi_{D_{k}+1}^{D_{k}+1}}{k-D_{k}} = 
	\frac{\xi_{D_{k}+1}^{D_{k}+1}}{T_{k}} \le \frac{\xi_{D_{k}+1}^{d_{0}}}{T_{k}} \le \alpha.
\]
It follows from the definition of $d_{\infty}$ (\ref{eq:d_max_FDP_control}) that $D_{k}+1\le d_{\infty}$.
\end{proof}
\begin{Corollary}
	\label{cor:dmax_FDP_control}
	Setting $\dmax=d_{\infty}$ when controlling the FDP using $k_0$ in \eqref{eq:FDP_control} guarantees it is large enough: $D_{k_0}+1\le \dmax$.
\end{Corollary}

We next show how to set \dmax\ for bounding the FDP while controlling the FDR using AS (recall that TDC and SSS+ are special cases of AS). Let
\begin{equation}
	\label{eq:dmax_TDC}
d_c \coloneqq \lfloor\alpha(m+1)/(\alpha + B)\rfloor.
\end{equation}
\begin{Lemma}
	With $k\subAS$ defined in \eqref{eq:AS} let $D\subAS\coloneqq D_{k\subAS}$, and $T\subAS\coloneqq T_{k\subAS}$.
	If $k\subAS>0$ then $D\subAS + 1\le d_c$.
\end{Lemma}
\begin{proof}
	If $k\subAS>0$ then by \eqref{eq:AS} $B(D\subAS+1)/T\subAS\le\alpha$ and hence 
\[
B(D\subAS+1)\le\alpha T\subAS=\alpha(k\subAS-D\subAS)\le\alpha(m-D\subAS).
\]
It follows that $D\subAS+1\le\alpha(m+1)/(\alpha + B)$, and as $D\subAS\in\N$,
$D\subAS+1\le d_c$ as claimed.
\end{proof}
\begin{Corollary}
	\label{cor:dmax_TDC_FDP}
	Setting $\dmax=d_c$ as in \eqref{eq:dmax_TDC} guarantees it is large enough: $D\subAS+1\le \dmax$. The same holds for TDC with $c=\lam=1/2$ and $B=1$,
	as well as for SSS+ with $c=\lam$ and $B=c/(1-c)$.
\end{Corollary}

\section{Interpolation}
All three considered bands for $V_k$ can be improved through what Goeman et al.~called ``interpolation''.
Specifically, this improvement uses the fact that the implied number of guaranteed discoveries (essentially $T_k-\bar V_k$)
should not decrease. More precisely, given a band $\bar V_i$ we define
\begin{equation*}
    \bar{G}_k \coloneqq \max_{i=1,\ldots,k} \lceil T_i - \bar V_i \rceil \vee 0 .
\end{equation*}
We then update our FDP band $\bar Q$ of \eqref{eq:band_Q} to
\begin{equation}
	\label{eq:band_iQ}
\bar Q_k =  \frac{T_k - \bar G_k}{T_k\vee 1}\wedge 1 .
\end{equation}

The latter interpolation step is included in our implementation of all three considered bands: our uniform (Theorem \ref{thm:ub}) and
standardized \eqref{eq:std_band} bands, and the KR band \eqref{eq:V-KR}. When controlling the FDP, we use those interpolated bands
with the threshold $k_0$ of \eqref{eq:FDP_control}, where our bands rely on \dmax\ as defined in Corollary \ref{cor:dmax_FDP_control}.
For bounding the FDP when controlling the FDR we provide in \supsec~\ref{supsec:algs} algorithmic descriptions of our
TDC-UB and TDC-SB that rely on \dmax\ as specified in Corollary \ref{cor:dmax_TDC_FDP}, and of TDC-KR that does not require defining $d_{\max}$.

In Section \ref{sec:inter_impact} we show that the KR band can benefit substantially
from interpolation, while the improvements seem to be marginal for the other two bands.

\section{Comparative analysis using Real and Simulated Data}
\label{sec:applications}

To evaluate the procedures presented here we looked at their performance on simulated and real data where competition-based FDR
control is already an established practice. We focus on the commonly used case of a single decoy/knockoff, where $c=\lambda=1/2$,
but we also examine the use of multiple decoys.
However, before engaging in the analysis of specific datasets that inevitably involves some random elements, it is 
instructive to deterministically compare the original, non-interpolated bands for $V_k$.

\subsection{Deterministic comparison of the (non-interpolated) bands for $V_k$}
\label{sec:det_comp}
Because we only report target wins we should focus on comparing the bands for $k$s that correspond to target wins ($L_k=1$).
Assuming further that \dmax\ is set sufficiently large, with $d=D_k$, the corresponding values of $\bar V_k\supUSKR$ are given by $\xi_{d+1}\supSB$
\eqref{eq:std_band}, $\xi_{d+1}\supUB$ (Theorem \ref{thm:ub}), and $\xi_d\supKR$ which we define as the right hand side of \eqref{eq:V-KR}
with $D_k$ replaced with $d$.

Expressed this way, as depending only on $d$, the bands can be compared deterministically. Figure \ref{fig:det_comp}
shows that for $B=1$, with the exception of very small values of $d$ (0 and 1), both our bands provide tighter bounds than the KR band,
and increasingly so as $d$ increases.
Overall the UB band offers tighter bounds than the SB band although this is reversed for very small values of $d$, and at any rate the differences are rather small.
\supfig~\ref{supfig:det_comp2} looks at the same comparison for two more values of $B$: $B=1/3$ (as explained below, this corresponds to using the ``max method''
with 3 decoys), and $B=1/7$ (max method with 7 decoys). As $B$ decreases, the differences between the methods become smaller but the trend is similar to what
we see with $B=1$: for very small values of $d$ the KR band is better, but then it does significantly worse than both SB and UB. Similarly, UB is somewhat better
for most values of $d$ but for a really small $d$ SB is slightly better.

\begin{figure}
\centering
\begin{tabular}{ll}
\includegraphics[width=3in]{{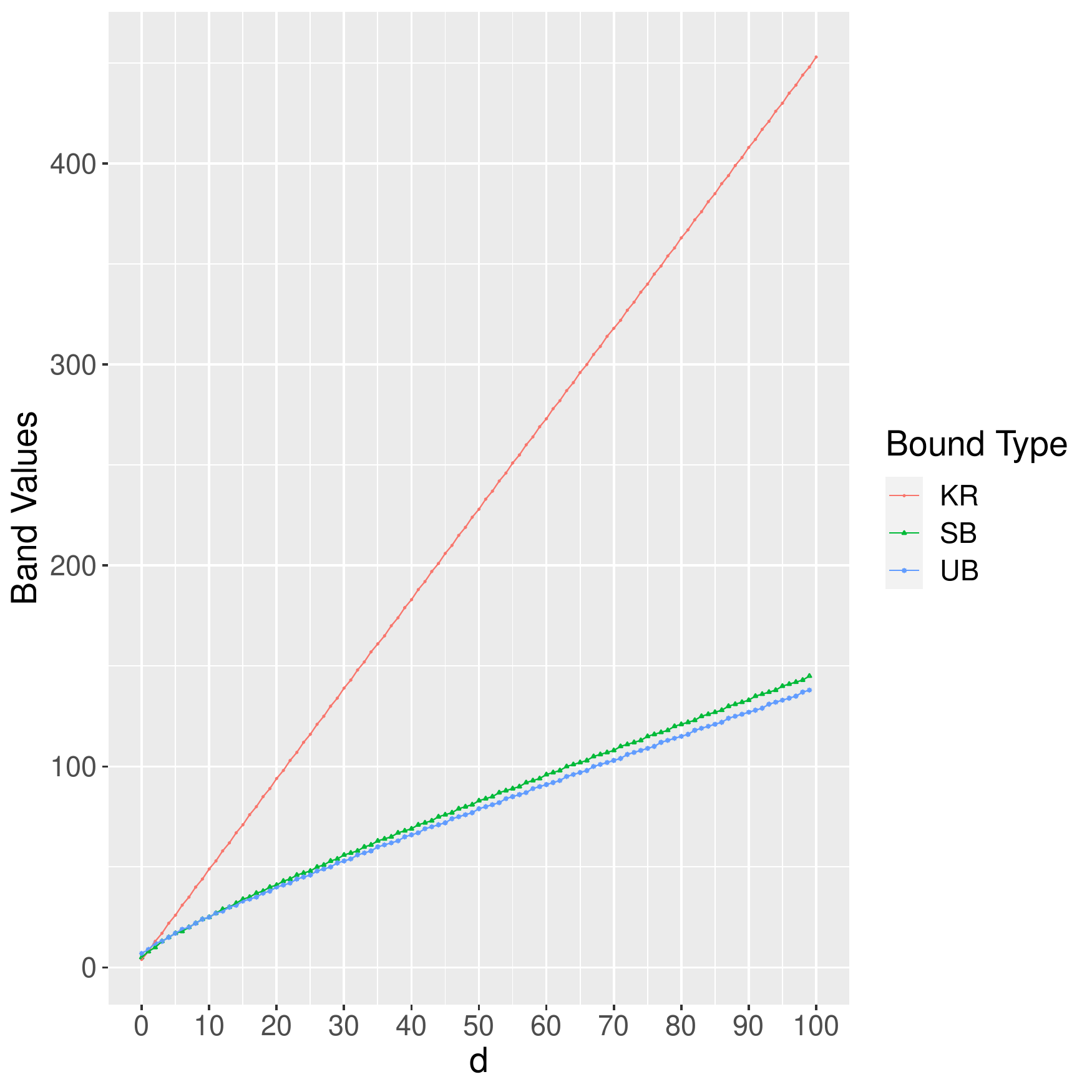}} &
\includegraphics[width=3in]{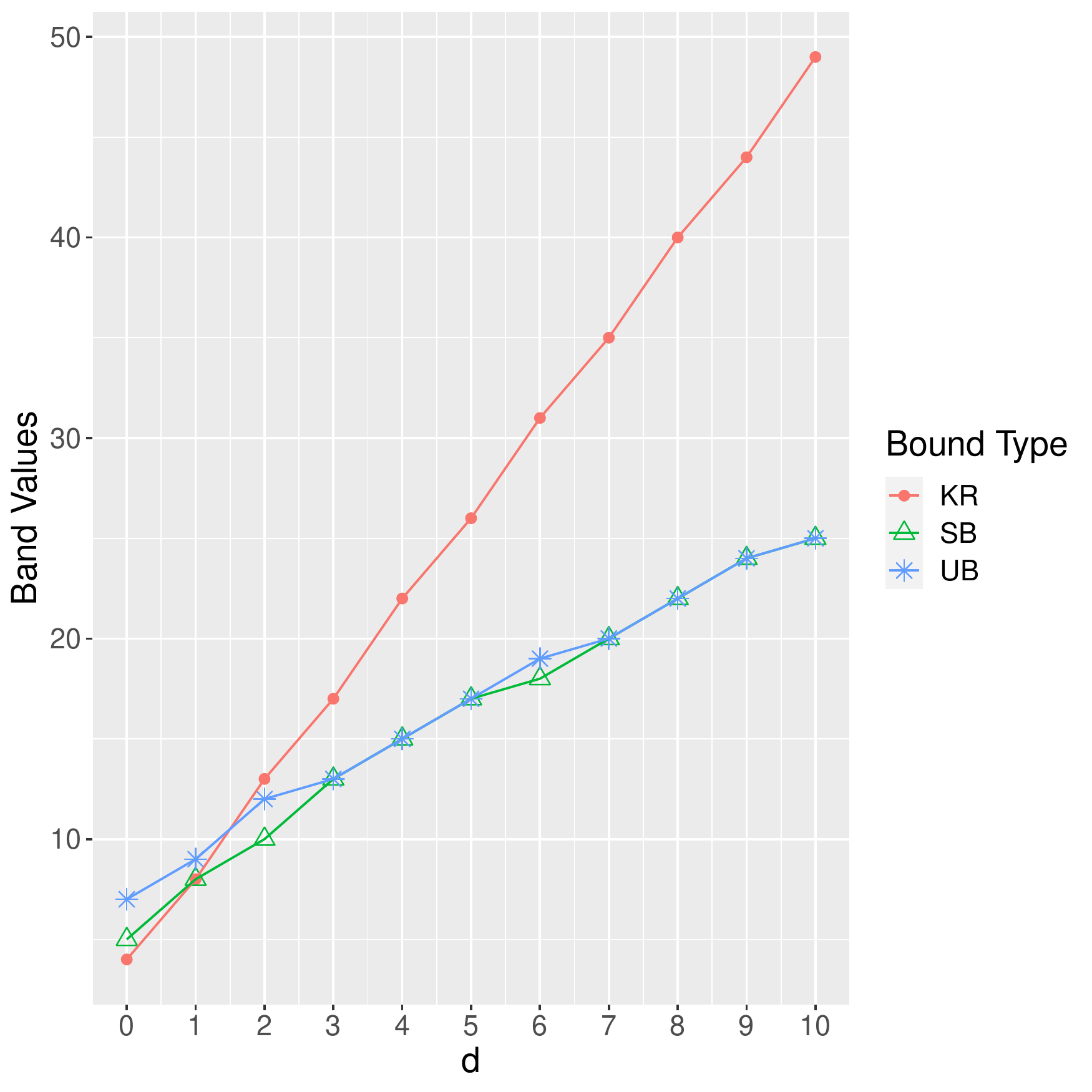}
\end{tabular}
\caption{\textbf{Deterministically comparing the $\bar V\supSB$, $\bar V\supUB$, and $\bar V\supKR$ bands.}\label{fig:det_comp} For $d \in \{1, \ldots, 100\}$ we computed the value of the upper prediction bands for the number of false discoveries as described in the text with $\dmax\coloneqq100$. The right figure is a zoomed-in version of the left figure for small values of $d$. \supfig~\ref{supfig:det_comp2} offers the same comparison for $B=1/3,1/7$.
}
\end{figure}

\subsection{Normal mixture model}
\label{sec:mixture}
We use datasets generated by the same mixture of normals model as in \cite{emery:multipleRECOMB,Luo:competition} to analyze the performance of the various
competition-based bands across a wide variety of controlled setups. 
Briefly, we drew decoy scores ($\til Z_i$) as well as true null scores ($Z_i$, $i\in N$) from a hypothesis-specific $N(\mu_i, \sigma_i)$ distribution,
and false null target scores ($Z_i$, $i\notin N$) from a shifted $N(\mu_i + \rho_i, \sigma_i^2)$.

Most of our analysis was done using simulated calibrated scores, where the null distribution does not vary
with the hypothesis, i.e., $\mu_i = \mu$ and $\sigma_i = \sigma$ for all $i$ (in this paper we used $\mu=0$, $\sigma=1$, as well as fixing $\rho_i\equiv3$,
except when noted otherwise). Where it is explicitly stated, we also used datasets generated simulating uncalibrated scores
where each $\mu_i$ is sampled from a $N(0, 1)$ distribution, $\sigma_i = 1 + \xi_i$ where $\xi_i$ is sampled from $\exp(1)$ (the
exponential distribution with rate $1$), and $\rho_i$ is sampled from a $1 + \exp(\nu)$ distribution, where $\nu$ is a hyperparameter that
determines the degree of separation between the false and true null target scores (we used $\nu=0.075$).
That said, because we generally saw little difference compared with using calibrated scores, we mostly used the latter.

We used the mixture model to randomly draw 20k sets of paired target/decoy scores for each considered parameter combination including:
varying the number of hypotheses $m\in\{500, 2\text{k}, 10\text{k}\}$, the proportion of true nulls $\pi_0\in\{0.2, 0.5, 0.8\}$,
the calibrated-scores separation parameter $\rho\in\{2.5, 3.0, 3.5\}$.

\subsection{The impact of interpolation}
\label{sec:inter_impact}
The penultimate section deterministically compared the bands on $V_k$ for which interpolation is irrelevant.
As our interest is in bounding the FDP, in this section we look explicitly at the effect of interpolation on those bounds.
Specifically, we used the above mixture of normals model to generate calibrated-score datasets using different data-parameter combination
by varying $m$, $\pi_0$, and $\rho$ as described above.

We then applied TDC (AS) to each simulated competition set, while varying the FDR threshold $\alpha$ from 1\% to 10\% and noted
the rejection threshold, $k\subTDC$ ($k\subAS$). For each considered band we found the difference between the upper prediction bound
on the FDP in TDC's list of target wins as given by the original band, and the corresponding value of the interpolated band.
We then plotted the median of those differences across the 20k drawn datasets for each of the parameter combinations
used here.

The results for one such combination ($m = 2000$, $\pi_0 = 0.5$, $\rho = 3$) are shown in Figure~\ref{fig:interpolation},
and all parameter combinations are presented in \supfig~\ref{supfig:interpolation}.
While the interpolated bands always offer an improvement over the non-interpolated versions, in the case of TDC-UB and TDC-SB the
difference is marginal: less than 0.01 across all parameter combinations. In contrast, the interpolation seems to drastically improve TDC-KR,
particularly for larger FDR thresholds.

While the impact of interpolation is significantly larger on the KR band, the analysis of Section \ref{sec:det_comp} above shows
it also has the lowest starting point. Indeed, the subsequent analyses below show that even with the significant interpolation-induced gains
the KR band is typically inferior to our bands (all considered bands are interpolated).

\begin{SCfigure}
\includegraphics[width=2.5in]{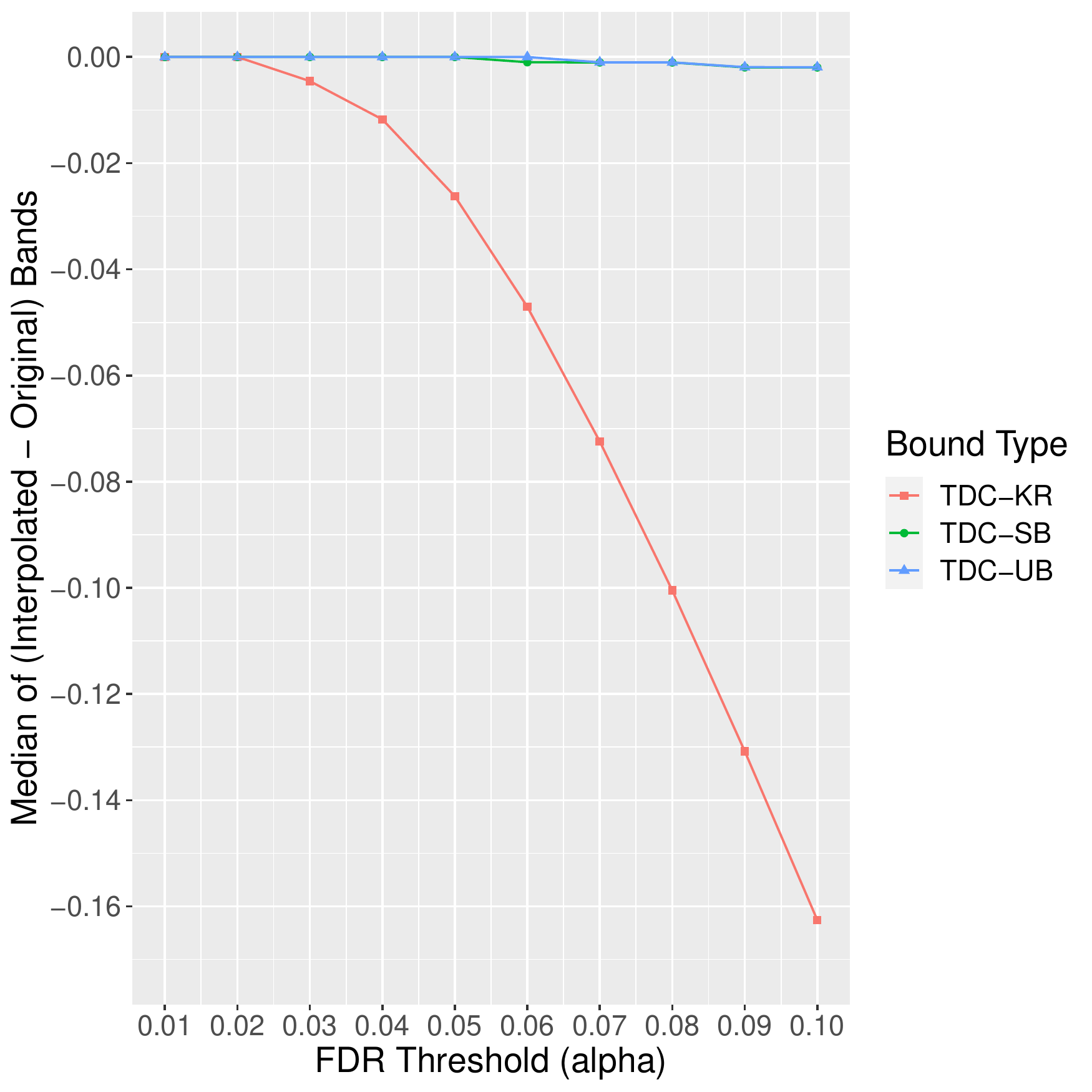}
\caption{\textbf{Comparing the interpolated and non-interpolated bands.} 
The median (over 20k datasets) of the difference between the non-interpolated and the interpolated bound on TDC's FDP
($m = 2000$, $\pi_0 = 0.5$, $\rho = 3$). Interpolation marginally improves both TDC-UB and TDC-SB here (they visually coincide), whereas 
it significantly boosts TDC-KR. More parameter combinations are examined in \supfig~\ref{supfig:interpolation}.
\label{fig:interpolation}}
\end{SCfigure}

\subsection{Examining the bounds on TDC's FDP}
\label{sec:FDP_bnd_resuls}

In this section we examine how well TDC-SB, TDC-UB, and TDC-KR (all interpolated) bound the FDP in TDC's list of discoveries.
We first look at the results in data generated by our mixture model, then look at variable selection in simulated linear
regression models, as well as genome-wide association studies (GWAS).

\subsubsection{In the mixture model setting}
\label{sec:FDP_bnd_normal}

We first randomly drew 20k datasets for each of the 18 parameter combinations of calibrated/uncalibrated scores with $m\in\{500, 2\text{k}, 10\text{k}\}$,
and $\pi_0\in\{0.2, 0.5, 0.8\}$. The hyperparameters used for the calibrated scores were $\mu = 0$, $\sigma = 1$ and $\rho = 3$, and
for the uncalibrated scores we used $\nu = 0.075$.

We then applied TDC with $\alpha\in\{0.01, 0.05, 0.1\}$ to each dataset, followed by two applications of each of TDC-SB/UB/KR,
once for each value of the confidence parameter $\gam\in\{0.01, 0.05\}$.
The upper prediction bounds were recorded and their median was calculated for each of the 108 different parameter combinations ($18\times3\times2$).
The boxplots of those medians for each of our three procedures are displayed in the left panel of Figure \ref{fig:sim_exp}.

In only 8 of those 108 cases TDC-KR's median bound was smaller than those of TDC-UB and TDC-SB, and typically it was substantially larger:
for $\gamma = 0.01$, the median of the 108 points were: 0.087 (TDC-UB), 0.094 (TDC-SB) and 0.243 (TDC-KR),
and for $\gamma = 0.05$: 0.079 (TDC-UB), 0.083 (TDC-SB) and 0.189 (TDC-KR).

The right panel of the figure looks at the median of the difference between the FDP bound returned by TDC-SB and TDC-UB.
Notably, TDC-UB generally offers tighter bounds, but the difference is not substantial.

\begin{figure}
\centering
\begin{tabular}{ll}
\hspace{-8ex}
\includegraphics[width=3in]{{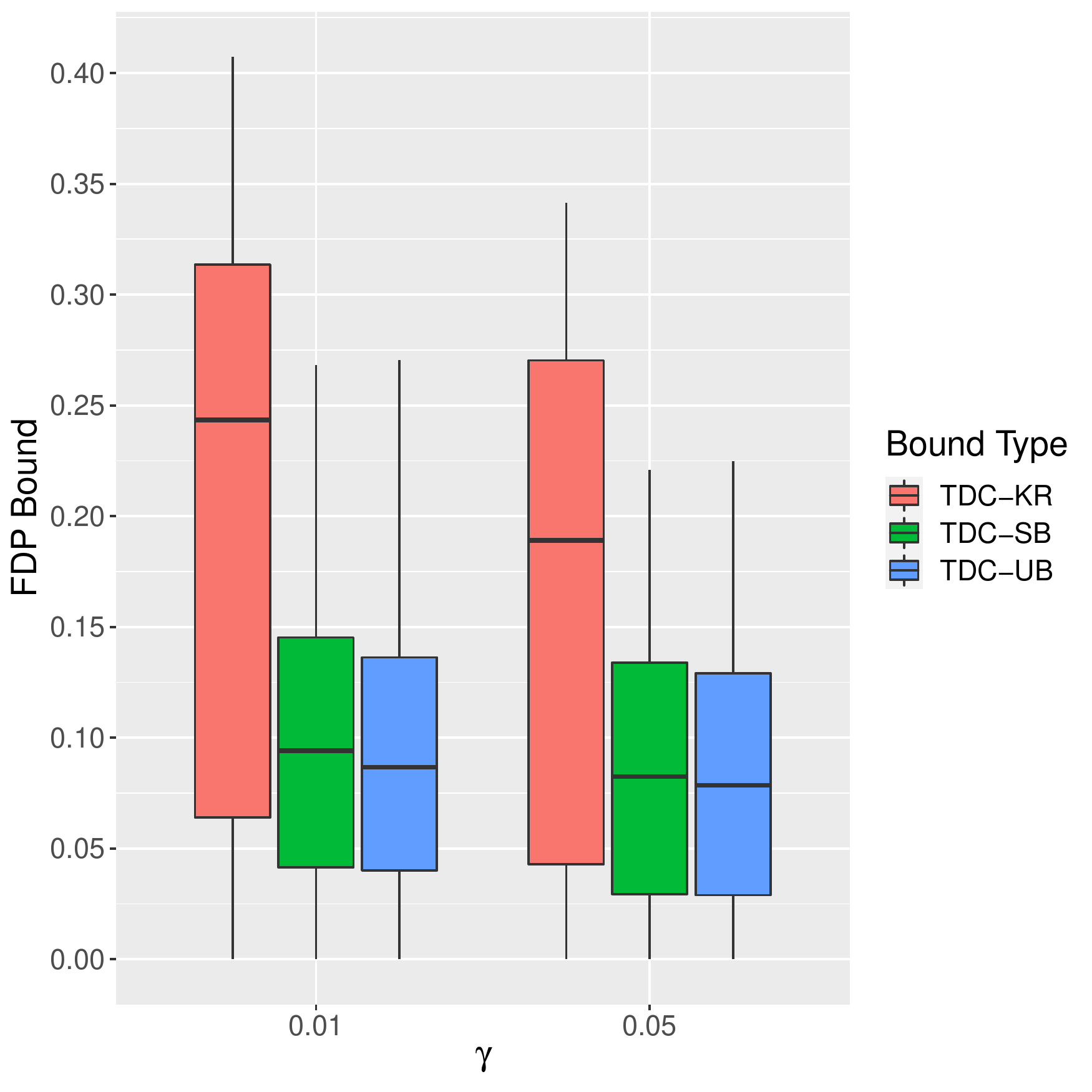}} &
\includegraphics[width=3in]{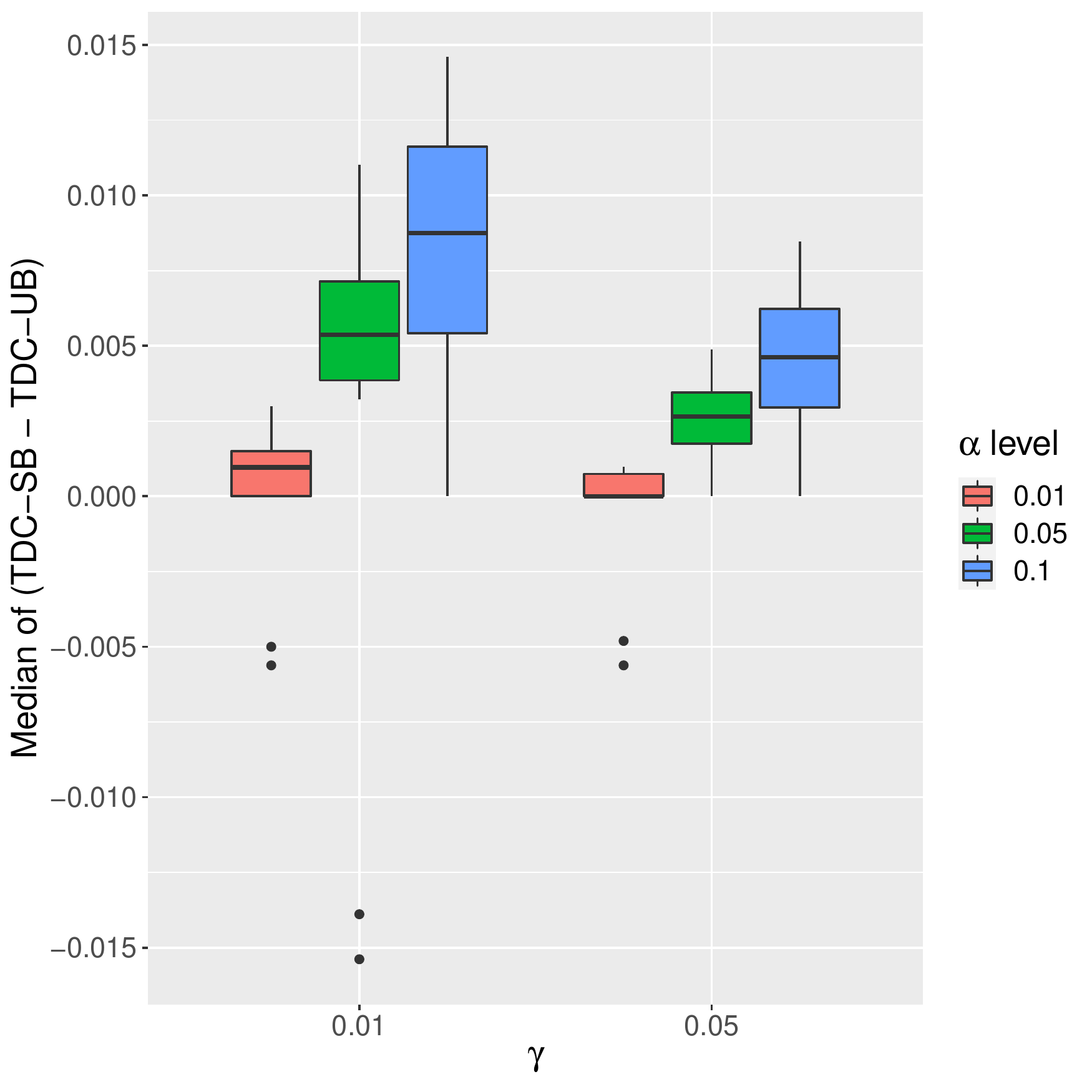}
\end{tabular}
\caption{\textbf{Comparing TDC-KR, TDC-SB and TDC-UB using the mixture model.}\label{fig:sim_exp}
	\textbf{Left:} For each of the 108 combinations of calibrated/uncalibrated scores with $m\in\{500, 2\text{k}, 10\text{k}\}$,
	$\pi_0\in\{0.2, 0.5, 0.8\}$, $\alpha\in\{0.01, 0.05, 0.1\}$, and  $\gam\in\{0.01, 0.05\}$,
	we computed the median of the FDP bound returned by TDC-KR, TDC-SB and TDC-UB.
	Each median was taken over 20K samples and the boxplot shown consists of those 108 medians.
	\textbf{Right:} using the same randomly generated data we noted the median of the difference between the FDP bound returned by TDC-SB and TDC-UB.
}
\end{figure}

To gain further insight we varied a single parameter of the mixture model at a time (with $\alp=0.05$ and $\gam=0.05$ throughout).
First, we varied $m$ keeping the ``signal-to-noise ratio'' parameters, $\pi_0$ and $\rho$ fixed.
\supfig s~\ref{supfig:vary_m}-\ref{supfig:vary_m_pi0_0.8} show that increasing $m$, which increases the number of discoveries (bottom rows),
have a mixed effect on the bounds: while the variability of all three bounds decreases, the median bound decreases for TDC-SB/UB but it
increases for TDC-KR (middle rows).

Surprisingly, when increasing the number of discoveries by decreasing $\pi_0$ (\supfig~\ref{supfig:vary_pi0}), or by increasing $\rho$
(\supfig~\ref{supfig:vary_rho}) the median bound of TDC-KR decreases (as do TDC-SB/UB's). Given the KR band is originally determined by
\eqref{eq:V-KR}, the only explanation is that this is due to the added interpolation step.

\subsubsection{In linear regression / GWAS}
\label{sec:FDP_bnd_lin}

As an example of how the FDP bounding procedures compare in the context of linear regression, we looked at the first example of Tutorial 1
of ``Controlled variable Selection with Model-X Knockoffs'' (
\href{http://web.stanford.edu/group/candes/knockoffs/software/knockoffs/tutorial-1-r.html}{``Variable Selection with Knockoffs''}) \cite{candes:panning}.
Specifically, we repeated the following sequence of operations 1000 times:
we randomly drew a normally-distributed $1000\times1000$ design matrix and generated a response vector using only 60 of the 1000
variables while keeping all other parameters the same as in the online example (amplitude=4.5, $\rho$=0.25, $\Sigma$ is a Toeplitz matrix whose
$d$th diagonal is $\rho^{d-1}$).
We then computed the model-X knockoff scores (taking a negative score as a decoy win and a positive score as a target win)
and applied TDC with $\alp\in\{0.1,0.2\}$ followed by TDC-SB/UB/KR with confidence levels of $1-\gam\in\{0.90, 0.95\}$.

Our GWAS example is taken from \cite{katsevich:simultaneous}, which in turn is based on data made publicly available by
\cite{sesia:multiresolution}. The goal of this analysis was to identify genomic loci (the features) that are significant factors in the expression of
each of the eight traits that were analyzed (the dependent variables). The raw data was taken from the UK Biobank \cite{bycroft:biobank}
and transformed to a regression problem by \cite{sesia:multiresolution}, who then created knockoff statistics.
We downloaded the scores using the functions \verb+download_KZ_data+ and \verb+read_KZ_data+ defined in \KR'
\href{https://raw.githubusercontent.com/ekatsevi/simultaneous-fdp/master/UKBB_utils.R}{UKBB\_utils.R}. Consistent with the latter, we
applied TDC with $\alp=0.1$, and the FDP bounding procedures using $\gamma = 0.05$.

Figure \ref{fig:lin_reg} shows that the trends we saw in our simulated mixture data also hold in both of the examples analyzed here:
there is very little that separates TDC-SB and TDC-UB and both overall provide substantially tighter upper bounds than TDC-KR's.

\begin{figure}
\centering
\begin{tabular}{ll}
\hspace{-8ex}
\includegraphics[width=3.5in]{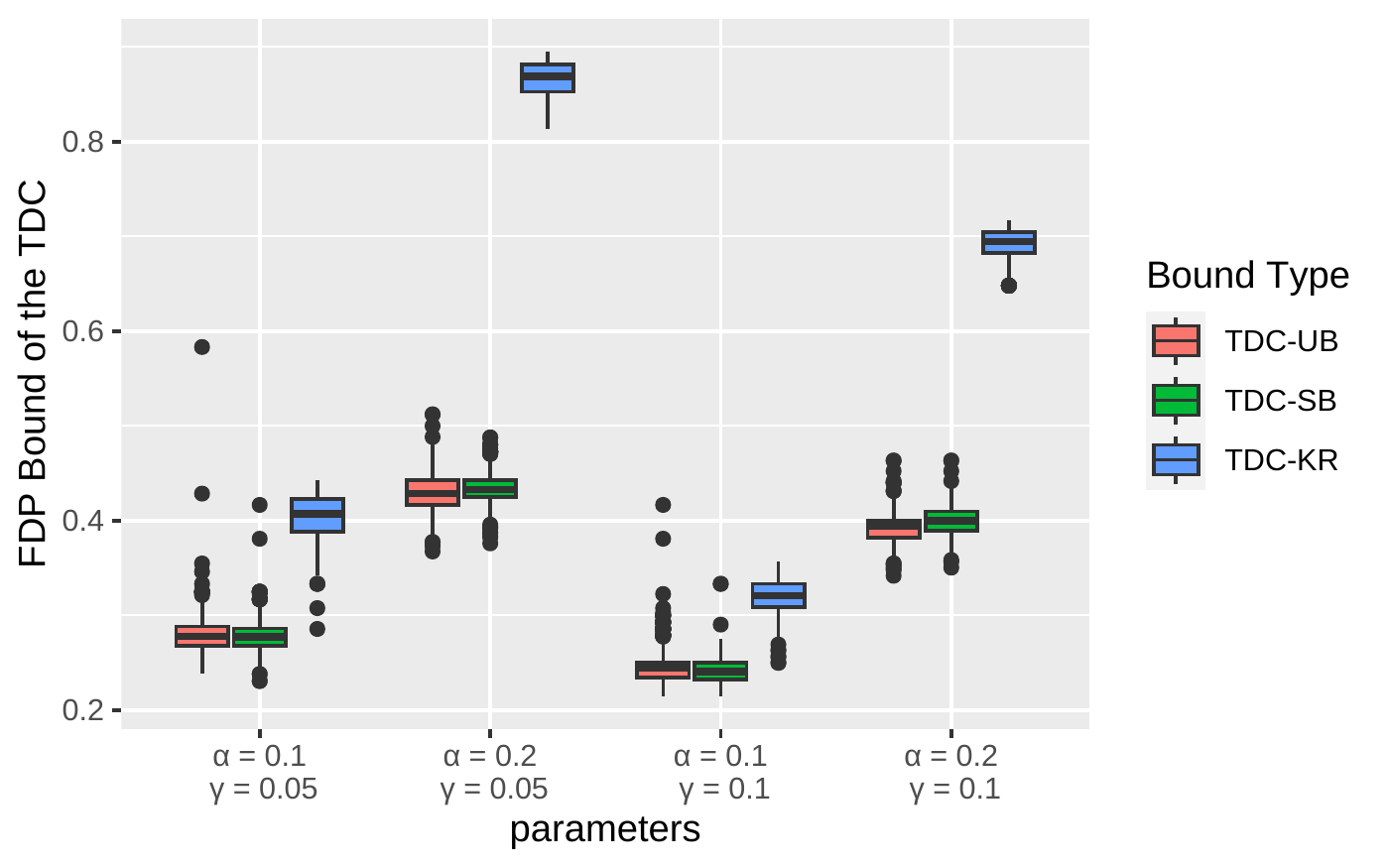} &
\includegraphics[width=3.5in]{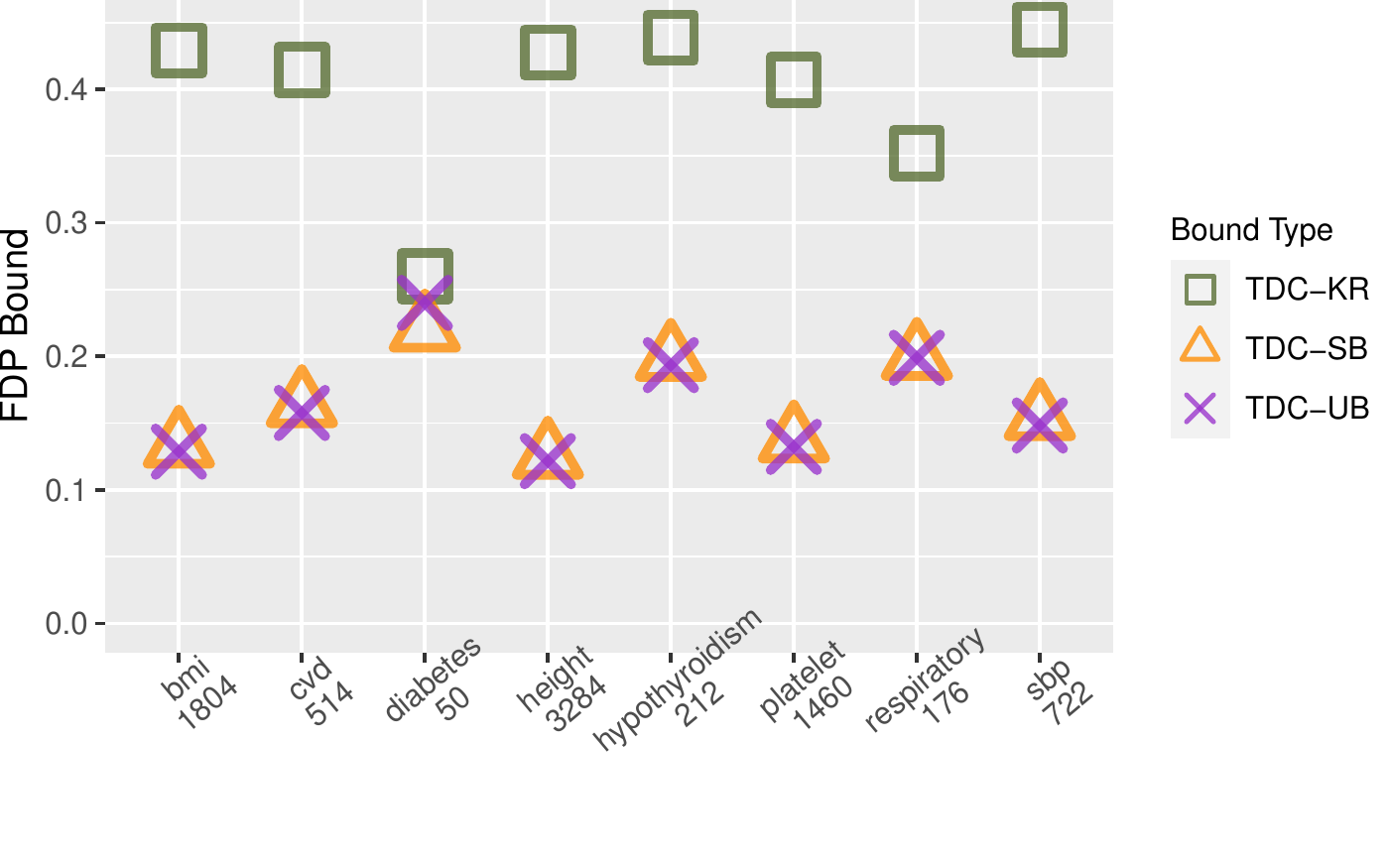}
\end{tabular}
\caption{\textbf{Comparing TDC-KR, TDC-SB and TDC-UB when selecting features in (i) simulated linear regression (ii) GWAS.}
\textbf{Left:} Each boxplot is made of the upper prediction bounds on the FDP in the selected list of variables as provided by TDC-KR, TDC-SB and TDC-UB.
Specifically, using the model-X knockoff scores, we selected variables in the linear regression problem described in the text,
while controlling the FDR using TDC/SSS+ with $\alp\in\{0.1,0.2\}$, and followed it by applying TDC-KR/SB/UB with confidence levels of $1-\gam\in\{0.90, 0.95\}$.
For each of the 8 investigated traits we compare the upper prediction bounds provided TDC-KR/SB/UB (all with $\gam=0.05$)
on TDC's FDP. TDC was run with $\alp=0.1$ and its number of discoveries is reported below the trait.  TDC-UB's upper prediction bound was averaged over 1K runs.
\textbf{Right:} For each of the 8 investigated traits we compare the upper prediction bounds provided by TDC-KR/SB/UB (all with $\gam=0.05$)
on TDC's FDP. TDC was run with $\alp=0.1$ and its number of discoveries is reported below the trait.  TDC-UB's upper prediction bound was averaged over 1K runs.
\label{fig:lin_reg}}
\end{figure}

\subsection{Controlling the FDP}
\label{sec:FDP_control}

FDP-SD is our recently published method for controlling the FDP in the competition setup. We showed it is generally more powerful 
compared with using the original (non-interpolated) KR band \cite{Luo:competition}.
Here we briefly revisit this problem, comparing FDP-SD with controlling the FDP with all three interpolated bands.

Specifically, we drew 20k datasets for each of the same parameter combinations of our normal mixture model as above.
We then applied FDP-SD, as well as the cutoff \eqref{eq:FDP_control} with each of the bands (SB, UB, and KR) to yield four FDP-controlled
lists of discoveries with a fixed confidence $1-\gamma = 0.95$ for each $\alpha \in \{0.01, 0.02, \ldots, 0.1\}$.

Figure \ref{fig:fdp-control} show the median power of each method over the 20k datasets for one parameter combination ($m = 2000$, $\pi_0 = 0.5$, $\rho = 3$),
and \supfig~\ref{supfig:fdp-control} shows the results for all combinations. The results are consistent: FDP-SD generally delivers the most
power, using the KR band typically delivers the fewest discoveries, and there is very little difference between using SB and UB.

\begin{SCfigure}
\includegraphics[width=2.5in]{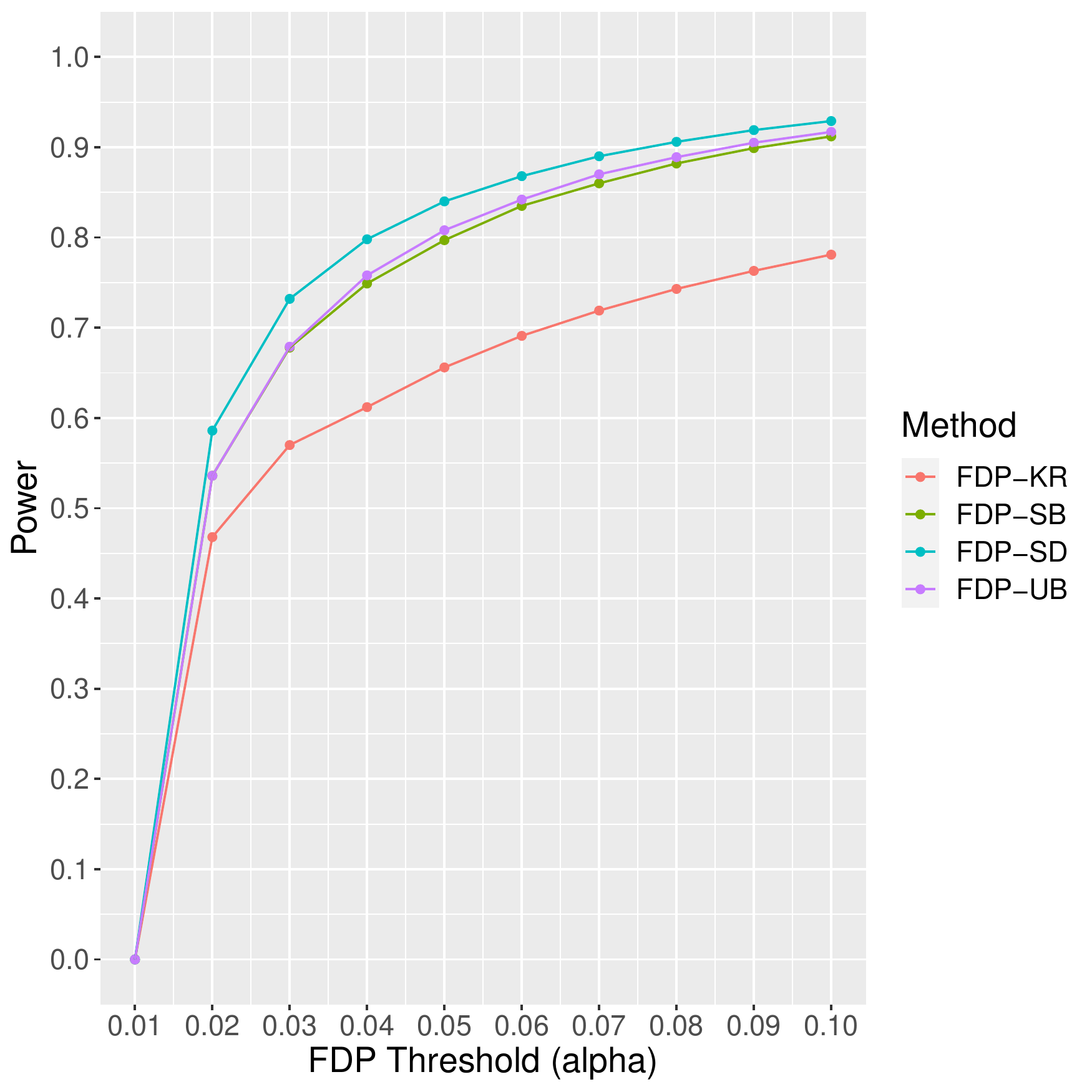}
\caption{\textbf{Median power of FDP controlling procedures.} 
Plotted are the median power (over 20k datasets) of FDP controlling procedures with  $\alpha \in \{0.01, 0.02, \ldots, 0.1\}$
and a fixed confidence level of $1-\gamma = 0.95$. The datasets were generated using our normal mixture model with $m = 2000$, $\pi_0 = 0.5$, $\rho = 3$.
More parameter combinations are examined in \supfig~\ref{supfig:fdp-control}.
\label{fig:fdp-control}}
\end{SCfigure}

\subsection{Bounding the FDP in peptide detection}
\label{sec:peptide}

We next report some results using real tandem mass spectrometry (MS/MS) data --- a technology that currently provides the most efficient
means of studying proteins in a high-throughput fashion.
In a ``shotgun proteomics'' MS/MS experiment, the proteins that are extracted from a complex biological sample 
are not measured directly. For technical reasons, the proteins are first digested into shorter
chains of amino acids called ``peptides.''  The peptides are then run through the mass spectrometer, in which distinct peptide sequences generate
corresponding spectra. A typical 30-minute MS/MS experiment will generate approximately 18,000 such spectra.
Thus, one of the first goals of the downstream analysis is to determine which peptides were present in the sample.
The list of discoveries is canonically generated by controlling the FDR using TDC as briefly explained below.

The peptide detection is done relative to an appropriate reference (``target'') peptide database: only target peptides can be detected.
In order to control the FDR, a ``decoy'' peptide database is generated by reversing, or randomly shuffling each target peptide.
Pioneered by SEQUEST \cite{eng:approach}, a search engine then uses an elaborate score function to assign to each spectrum its
optimally matching peptide in the concatenated target-decoy database --- forming the (optimal) peptide-spectrum match, or PSM \cite{nesvizhskii:survey}.

The problem with those PSMs is that in practice, many expected fragment ions will fail to be observed for any given spectrum,
and the spectrum is also likely to contain a variety of additional, unexplained peaks \cite{noble:computational}.
Hence, sometimes the PSM is correct --- the peptide assigned to the spectrum was
present in the mass spectrometer when the spectrum was generated --- and sometimes the PSM is incorrect.
This uncertainty carries over to the peptide level, hence the need to control the FDR.

We next assign each target database peptide two scores: a target score $Z_i$, which is the maximal score of all PSMs associated with it,
and a decoy score, $\til Z_{i}$, which is the maximal score of all PSMs associated with the target's paired decoy peptide
(assume for simplicity that a PSM score is $>0$ so we assign a score of 0 if no PSM is associated with the target/decoy peptide).
Finally, we apply TDC, reporting all top scoring target winning peptides ($Z_i>\til Z_i$) with the score cutoff determined by \eqref{eq:AS}.
The underlying assumption justifying the use of TDC is that for a true null hypothesis (the peptide is not in the sample) the winning score is equally 
likely to be the target or the decoy (independently of all the winning scores and of all other peptides) \cite{lin:improving}.

We examined the bounds provided by TDC-SB/UB/KR on the FDP in TDC's list of detected peptides in 10 essentially randomly selected MS/MS datasets.
To increase the confidence in our analysis we repeated the process 20 times for each dataset using that many randomly drawn decoy databases.
Then, for each of the 10 datasets and each $\alpha \in \{0.01, 0.02, \ldots, 0.1\}$ we averaged each method's FDP bound over its 20 applications
to this dataset, each with its own decoy database and a fixed confidence level of $1-\gamma = 0.95$. Additional details of the process are specified in
\supsec~\ref{supsec:PXIDs}.

Figure \ref{fig:peptides}A summarizes those sets of 10 averages, one for each FDR threshold $\alp$ and FDP-bounding method
in the form of boxplots. Clearly, the picture is consistent with the other applications we looked at: TDC-UB generally offers
marginally tighter bounds than TDC-SB, and both bounds are significantly tighter than TDC-KR's.

\begin{figure}
\centering
\begin{tabular}{cc}
\hspace{-8ex}
\includegraphics[width=3.5in]{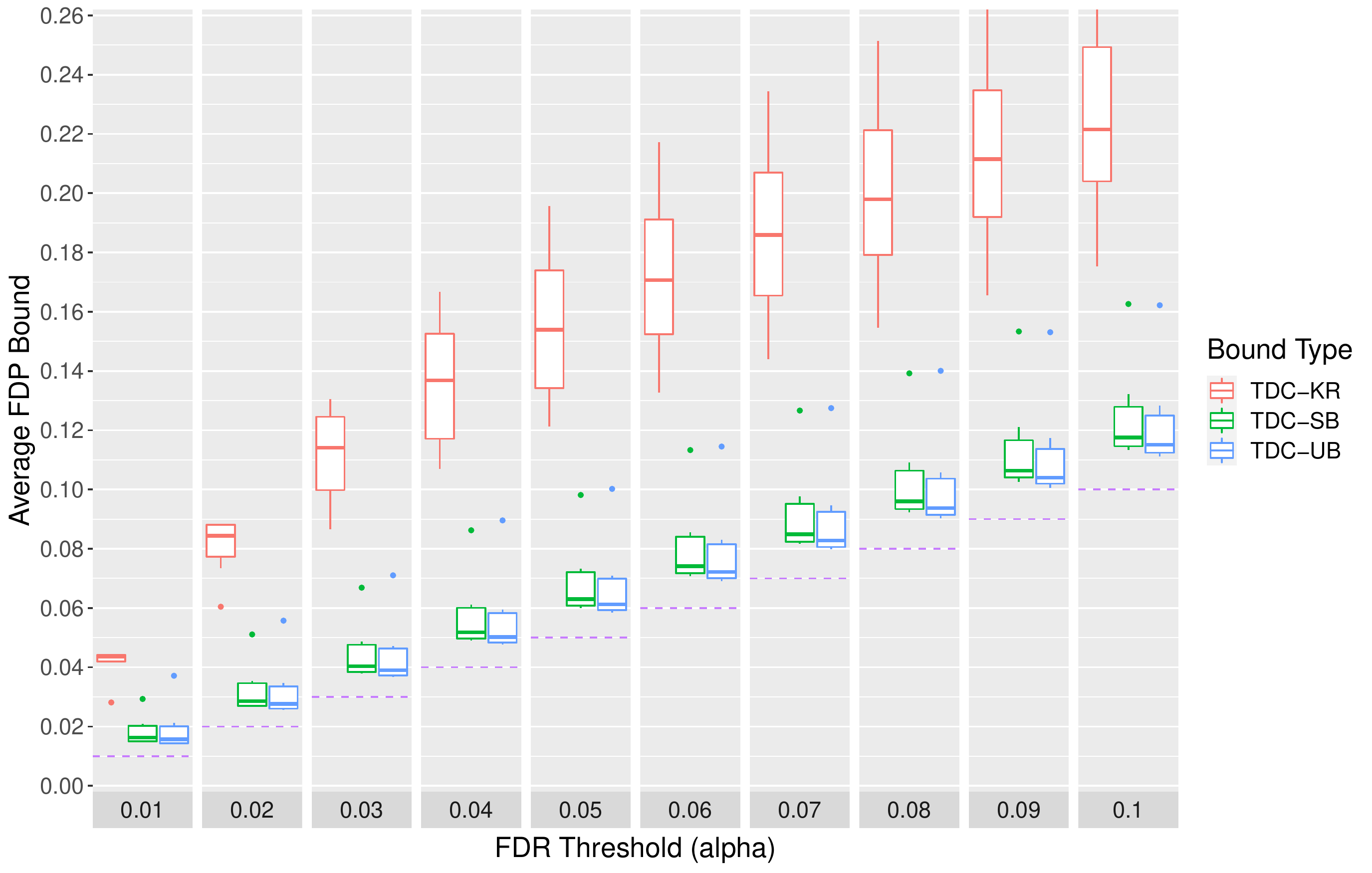} &
\includegraphics[width=3.5in]{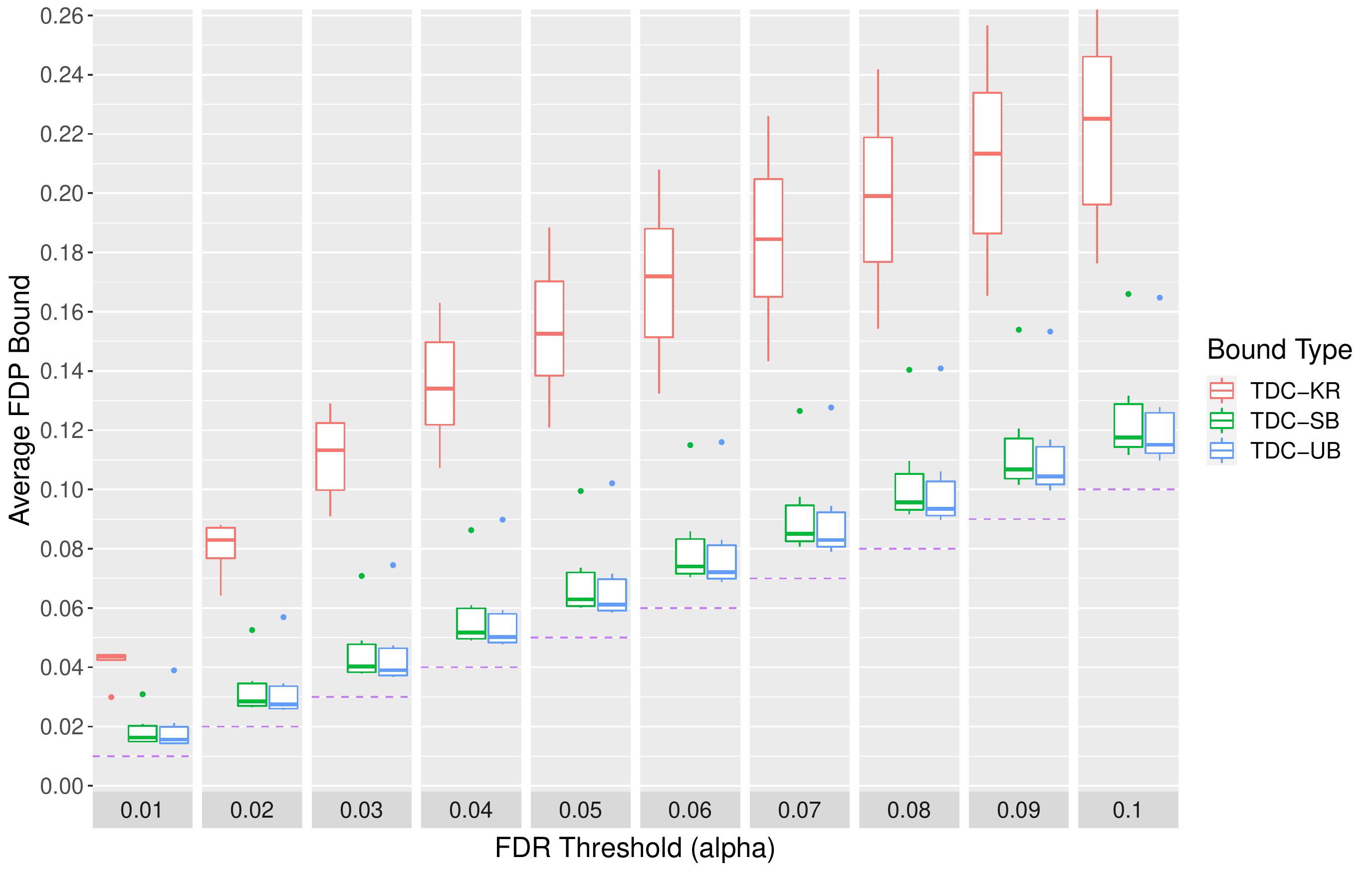} \\
\hspace{-8ex}
A & B
\end{tabular}
\caption{\textbf{Applications of TDC-KR, TDC-SB and TDC-UB in peptide detection.}
Each boxplot is made of 10 averages, one for each of our 10 datasets, of the upper prediction bounds on the FDP in TDC/AS's
reported list of peptides as provided by TDC-KR, TDC-SB and TDC-UB.
Each average is taken over 20 applications, each with a different set of randomly drawn decoy databases (single for A, three for B), the prescribed FDR level $\alp$,
and a fixed confidence level of $1-\gamma = 0.95$. The purple dashed line corresponds to $y = \alpha$.
\textbf{A:} using a single decoy (standard TDC).
\textbf{B:} using the mirror method with 3 decoys. \supfig\ref{supfig:upband-boxplot} includes more settings.
\label{fig:peptides}}
\end{figure}

We further used this peptide detection setup to examine the results of TDC-SB/UB/KR when controlling the FDR 
using multiple decoys \cite{emery:multipleRECOMB}. In this setup each target peptide is associated with $d$
decoys generated by random shuffles of the target peptide. Each spectrum is searched against the concatenated
database of all target and decoy peptides for its best matching PSM. Similarly to using a single decoy, we
associate with $H_i$ (the $i$th target peptide is not in the sample) a target score $Z_i$ (maximal score of all PSMs associated with it),
and $d$ decoy scores $Z^i_j$ $j=1,\dots d$ (maximal score of all PSMs associated with each corresponding decoy).

The p-value associated with $H_i$ is the rank p-value of the target $Z_i$ in the combined list of $d+1$
target and decoy scores (generalizing the $d=1$ case). Given the chosen $c$ and $\lam$ parameters, the rank
p-value determines whether $H_i$ corresponds to a target ($p_i\le c$, $L_i=1$) or a decoy win ($p_i>\lam$, $L_i=-1$) as in AS.
Emery et al.'s mirandom procedure can then assign winning scores $W_i$ so that if the decoys are independently generated
(or satisfy a weaker exchangeability condition), then we can sort the hypotheses in decreasing order of $W_i$
and the ($W_i$-sorted) true null rank p-values are iid that stochastically dominate the uniform $(0,1)$ distribution,
independently of the false nulls \cite[\supsec~6.13 of that paper]{emery:multipleRECOMB}.
In other words, AS' conditions of the sequential hypothesis testing hold.

Here we only considered the mirror method ($c=\lam=1/2$ for an odd $d$) and the max method ($c=\lam=1/(d+1)$).
Specifically, we applied the above process to the same 10 sets of real spectra with $d=3$ and $d=7$ decoys
controlling the FDR at varying level $\alpha \in \{0.01, 0.02, \ldots, 0.1\}$ using both the mirror and the max methods.
We then applied to the reported lists of discovered peptides the FDP-bounding procedures with confidence $1-\gam=0.95$.
Again, this process was repeated 20 times, randomly drawing $d$ out of 20 decoy databases in each,
and the FDP bounds were averaged over those 20 repetitions.

Figure \ref{fig:peptides}B presents boxplots each made of 10 of those averages. Interestingly, there is very
little difference from using the mirror method with a single decoy (A), or seven decoys (\supfig~\ref{supfig:upband-boxplot}).
This is further born out in \supfig~\ref{supfig:mult_method}, which examines the
same problem using our mixture model: using the mirror method with 1, 3, or 7 decoys the bounds provided by
TDC-UB are almost identical. Nevertheless, the advantage of using multiple decoys becomes obvious when comparing 
the power: using the mirror method to control the FDR with $d=3,7$ decoys increases the number of discoveries (\supfig~\ref{supfig:mult_power}).

The max method behaves somewhat differently: while TDC-KR still delivers substantially larger bounds than
TDC-SB/UB do, the bounds seem to decrease as $d$ increases. However, the power of FDR control using the
max method no longer seems to be monotone (same \supfig s).

\section{Discussion}

TDC-SB/UB were developed to improve the bounds on the FDP when controlling the FDR using the competition-based
TDC, or its generalizations to sequential hypothesis testing: AS and SSS+.
The bounds, derived from our novel SB/UB upper prediction bands on the FDP, are generally substantially tighter than bounds
derived analogously from the recently published KR band, even after improving the latter through interpolation.

When seeking tighter control over the FDP, the user should still use our recently published FDP-SD, which
overall delivers more power compared with controlling the FDP using any of the above bands.
Instead, TDC-SB/UB are designed for the typical case where the user would be reluctant to pay the price associated 
with controlling the FDP, allowing them to gauge how large can the FDP in their FDR-controlled list
of discoveries be.

While we focused on bounding the FDP in the FDR-controlled discovery list, our bands more generally provide simultaneous
FDP bounds for competition-based analysis of the multiple testing problem. As such, they are useful for post-hoc analyses,
including where we apply additional domain-knowledge to prioritize a subset of the hypotheses (\cite{katsevich:simultaneous,goeman:only}).

Notably, even TDC-UB, which generally delivers the tightest bounds, is often overly conservative suggesting 
further improvements can be made. One such improvement might be gained by trying to optimize the set of decoy-win indices $\Delta$ on which we bound
$N_d$.

Assuming the necessary quantiles are precomputed (\supsec~\ref{supsec:c_star_simulations}) all the procedures presented here require
sorted data, but beyond that their complexity is linear; hence, they all share the runtime complexity of $O(m\log m)$.
An R implementation of our bands (with precomputed quantiles for most practical problems) is available at 
\url{https://github.com/uni-Arya/bandsfdp}.

\bibliography{refs}
\bibliographystyle{plain}

\clearpage

\section{Supplementary Material}

\begin{table}[h!]
\centering
\scriptsize
\begin{tabular}{cp{5in}}
\hline 
Variable & Definition \\
\hline
$m$ & the number of hypotheses (e.g., PSMs, features, peptides)\\
$\alp$ & the FDR/FDP threshold\\
$\gam$ & for a $1-\gam$ confidence level\\
$N$ & the set of indices of the true null hypotheses (unobserved, could be a random set)\\
$Z_i$ & the target/observed score (the higher the score the less likely $H_i$)\\
$\til Z_{i}$ & decoy/knockoff score (generated by the user)\\
$c$ & parameter of AS (determines the target win threshold in TDC terminology)\\
$\lam$ & parameter of AS (determines the decoy win threshold in TDC terminology)\\
$B$ & $c/(1-\lambda)$ (the ratio of probabilities of target to decoy wins for true nulls)\\
$R$ & $(1-\lambda)/(c+1-\lambda)$ (for a true null, the probability of a decoy win given that it was not discarded)\\
$L_i$ & with values in $\{-1,1\}$ the target/decoy win labels (assigned, ties are randomly broken or the corresponding hypotheses are dropped)\\
$W_i$ & the winning score (assigned, WLOG assumed in decreasing order)\\
$D_i$ & the number of decoy wins in the top $i$ scores\\
$T_i$ & the number of target wins in the top $i$ scores\\
$V_i$ & the number of true null target wins in the top $i$ scores\\
$\bar V_i$ & upper prediction band for $V_i$\\
$Q_i$ & the FDP among the target wins in the top $i$ scores\\
$\bar Q_i$ & upper prediction band for $Q_i$\\
$N_d$ & the number of true null target wins before the $d$th decoy win\\
$\xi_d$ & upper prediction band for $N_d$\\
$\bar G_i$ & lower prediction band for the number of guaranteed discoveries (used for interpolation)\\
$U_d$ & a process that stochastically dominates $N_d$ with negative binomial marginal distributions\\
$\hat U_d$ & the standardized version of $U_d$: $\hat U_d\coloneqq (U_d-d)/\sqrt{2d}$\\
$\tilde{U}_d$ & the uniform (confidence) version of $U_d$: $\til U_d\coloneqq G_d(U_d)$, where $G_d(k)\coloneqq P(NB(d,1/2)\ge k)$\\ 
$\Delta$ & $\Delta\coloneqq\{1,2,\dots,d_{\max}\}$\\
$z_\Delta^{1-\gamma}$ & a $(1-\gamma)$-quantile of $\max_{d\in\Delta} \hat U_d$ \\
$u_\gamma(\Delta)$ & a $\gamma$-quantile of $\min_{d\in\Delta}\tilde{U}_d$\\
$k_0$ & rejection threshold for FDP control using upper prediction bands\\
$k_{\text{AS}}$ & rejection threshold for FDR control using AS\\
\hline 
\end{tabular}
\caption{\textbf{Commonly used notations and their definitions.}
  \label{suptable:definitions}}
\end{table}

\begin{table}[h!]
\centering
\scriptsize
\begin{tabular}{lp{5in}}
\hline 
Abbreviation & Definition \\
\hline
MS/MS & Tandem Mass Spectrometry\\
PSM & Peptide-Spectrum Match (the match between a spectrum and its best matching database peptide)\\
RV & Random Variable\\
NB & Negative Binomial (distribution)\\
FDP & False Discovery Proportion (the proportion of discoveries which are false)\\
FDX & False Discovery Exceedance (an alternative term for FDP-control used in the literature)\\
FDR & False Discovery Rate (the expected value of the FDP taken over the true nulls conditional on the false nulls)\\
TDC & Target Decoy Competition (canonical approach to FDR control in mass spectrometry)\\
AS & Adaptive SeqStep (FDR control in sequential hypothesis testing, generalizing TDC and SSS+)\\
SSS+ & Selective SequentialStep+ (FDR control in sequential hypothesis testing, a special case of AS)\\
FDP-SD & FDP-Stepdown (recommended approach to FDP control using competition)\\
FDP-UB & FDP-Uniform Band (an alternative approach to FDP control based on the uniform band)\\
FDP-SB & FDP-Standardized Band (same as FDP-UB but based on the standardized band)\\
FDP-KR & FDP-\KR\ Band (same as FDP-UB but based on the \KR\ band)\\
TDC-UB & TDC-Uniform (Upper Prediction) Bound (a bound on TDC's FDP derived from the uniform band)\\
TDC-SB & TDC-Standardized (Upper Prediction) Bound (same as TDC-UB but using the standardized band)\\
TDC-KR & TDC-\KR\ (Upper Prediction) Bound (same as TDC-UB but using the KR-band)\\
GWAS & Genome-Wide Association Studies (here referring to a specific analysis of 8 traits using Biobank data)\\
\hline 
\end{tabular}
\caption{\textbf{Commonly used abbreviations/names and their definitions.}
  \label{suptable:abbreviations}}
\end{table}

\clearpage

\subsection{Procedures in algorithmic format}
\label{supsec:algs}

\begin{algorithm}
\textbf{\caption{AS (also referred to as TDC)}\label{algorithm:TDC}}
\SetKwFor{For}{For}{let:}{endfor}
\SetKwIF{If}{ElseIf}{Else}{If}{then:}{else if:}{else:}{endif} 
\KwIn{
    \begin{alglist}
        \item an FDR threshold $\alpha$\;
        \item the number of hypotheses $m$\;
        \item competition parameters $c$ and $\lambda$;
        \item a list of labels $L_i \in \{1, -1\}$ where $1$ indicates a target win and $-1$ a decoy win (sorted so that the corresponding scores $W_i$ are in decreasing order: $W_1\ge W_2\ge\dots\ge W_m$)\;
    \end{alglist}
}
\KwOut{
    \begin{alglist}
        \vspace{-0.35em}
        \item an index $k_{\text{AS}}$ specifying that target wins in the top $k_{\text{AS}}$ hypotheses are discoveries\;
    \end{alglist}
}
\vspace{-0.5em}
\dottedhfill \\
\setstretch{1.4}
\For{$i=1$ \KwTo $m$} {
    $D_i$ be the number of $-1$'s in $\{L_1,\ldots,L_i\}$\;
    $T_i$ be the number of $1$'s in $\{L_1,\ldots,L_i\}$\;
}
$B := c/(1-\lambda)$\;
$M := \{k \in \{1, \ldots, m\} : B(D_k+1)/T_k \leq \alpha\}$\;
\eIf{$M = \emptyset$}{
    \Return $k_{\text{AS}} := 0$\;
}{
    \Return $k_{\text{AS}} := \max(M)$\;
}
\end{algorithm}

\begin{algorithm}
\textbf{\caption{TDC-UB}\label{algorithm:TDC-UBm}}
\SetKwFor{For}{For}{let:}{endfor}
\SetKwIF{If}{ElseIf}{Else}{If}{then:}{else if:}{else:}{endif} 
\KwIn{
    \begin{alglist}
        \item an FDR threshold $\alpha$\;
        \item a confidence parameter $\gamma$\;
        \item the number of hypotheses $m$\;
        \item the threshold $\tau := k_\text{AS}$ returned by AS\;
        \item competition parameters $c$ and $\lambda$;
        \item a list of labels $L_i \in \{1, -1, 0\}$ where $1$ indicates a target win, $-1$ a decoy win and $0$ an uncounted hypothesis (sorted so that the corresponding scores $W_i$ are decreasing: $W_1\ge W_2\ge\dots\ge W_m$)\;
    \end{alglist}
}
\KwOut{
    \begin{alglist}
        \vspace{-0.35em}
        \item a $1-\gamma$ upper prediction bound $\bar{Q}_\tau$ on the FDP in the list of discoveries returned by AS\;
    \end{alglist}
}
\vspace{-0.5em}
\dottedhfill \\
\setstretch{1.4}
\If{$\tau = 0$}{
    \Return $\bar{Q}_\tau = 0$\;
}
$B := c/(1-\lambda)$\;
$R := 1/(1+B)$\;
$d_{\text{max}} := \left\lfloor  \frac{\alpha(m + 1)}{\left(\alpha + B\right)} \right\rfloor$\;
Compute $u := u_\gamma(\Delta_{d_{\max}})$ (in practice we typically draw $u\in\{\rho_d,\sig_d$\}, where $\rho_d$ and $\sigma_d$ are pre-computed using MC simulations; see Section \ref{supsec:c_star_simulations})\;
\For{$i=1$ \KwTo $m$} {
    $D_i$ be the number of $-1$'s in $\{L_1,\ldots,L_i\}$\;
    $T_i$ be the number of $1$'s in $\{L_1,\ldots,L_i\}$\;
    $\bar{V}_i$ be defined as in \eqref{eq:band_V} using $\xi_d := \min\{i\,:\,F_{\text{NB}(d, R)}(i) \ge 1-u\}$, the $1-u$ quantile of the negative binomial $\text{NB}(d, R)$\;
    $\bar{G}_i = (\max_{j=1,\ldots,i}\lceil T_j - \bar{V}_j \rceil) \vee 0$\;
}
\eIf{$T_{\tau} = 0$}{
    \Return $\bar{Q}_\tau = 0$\;
}{
    \Return $\bar{Q}_\tau = (T_\tau - \bar{G}_\tau) / (T_\tau \vee 1)$;
}
\end{algorithm}

\begin{algorithm}
\textbf{\caption{TDC-SB}\label{algorithm:TDC-SBm}}
\SetKwFor{For}{For}{let:}{endfor}
\SetKwIF{If}{ElseIf}{Else}{If}{then:}{else if:}{else:}{endif} 
\KwIn{
    \begin{alglist}
        \item an FDR threshold $\alpha$\;
        \item a confidence parameter $\gamma$\;
        \item the number of hypotheses $m$\;
        \item the threshold $\tau := k_\text{AS}$ returned by AS\;
        \item competition parameters $c$ and $\lambda$;
        \item a list of labels $L_i \in \{1, -1, 0\}$ where $1$ indicates a target win, $-1$ a decoy win and $0$ an uncounted hypothesis (sorted so that the corresponding scores $W_i$ are decreasing: $W_1\ge W_2\ge\dots\ge W_m$)\;
    \end{alglist}
}
\KwOut{
    \begin{alglist}
        \vspace{-0.35em}
        \item a $1-\gamma$ upper prediction bound $\bar{Q}_\tau$ on the FDP in the list of discoveries returned by AS.
    \end{alglist}
}
\vspace{-0.5em}
\dottedhfill \\
\setstretch{1.4}
\If{$\tau = 0$}{
    \Return $\bar{Q}_\tau = 0$\;
}
$B := c/(1-\lambda)$\;
$d_{\text{max}} := \left\lfloor  \frac{\alpha(m + 1)}{\left(\alpha + B\right)} \right\rfloor$\;
Compute $z := z(\gamma)$, an approximated $1-\gamma$ quantile of $\max_{d\in\Delta_{d_{\max}}}\hat U_d$, where $\hat U_d\coloneqq (U_d-Bd)/\sqrt{B(1+B)d}$ (in practice, the quantile would be pre-computed using MC simulations rather than being computed on demand)\;
\For{$i=1$ \KwTo $m$} {
    $D_i$ be the number of $-1$'s in $\{L_1,\ldots,L_i\}$\;
    $T_i$ be the number of $1$'s in $\{L_1,\ldots,L_i\}$\;
    $\bar{V}_i$ be defined as in \eqref{eq:band_V} using $\xi_d := \lfloor z\sqrt{B(1+B)d}+Bd \rfloor$\;
    $\bar{G}_i = (\max_{j=1,\ldots,i}\lceil T_j - \bar{V}_j \rceil) \vee 0$\;
}
\eIf{$T_{\tau} = 0$}{
    \Return $\bar{Q}_\tau = 0$\;
}{
    \Return $\bar{Q}_\tau = (T_\tau - \bar{G}_\tau) / (T_\tau \vee 1)$;
}
\end{algorithm}

\vspace*{\fill}

\begin{algorithm}
\textbf{\caption{TDC-KRB}\label{algorithm:TDC-KRBm}}
\SetKwFor{For}{For}{let:}{endfor}
\SetKwIF{If}{ElseIf}{Else}{If}{then:}{else if:}{else:}{endif} 
\KwIn{
    \begin{alglist}
        \item an FDR threshold $\alpha$\;
        \item a confidence parameter $\gamma$\;
        \item the threshold $\tau := k_\text{AS}$ returned by AS\;
        \item competition parameters $c$ and $\lambda$;
        \item a list of labels $L_i \in \{1, -1, 0\}$ where $1$ indicates a target win, $-1$ a decoy win and $0$ an uncounted hypothesis (sorted so that the corresponding scores $W_i$ are decreasing: $W_1\ge W_2\ge\dots\ge W_m$)\;
    \end{alglist}
}
\KwOut{
    \begin{alglist}
        \vspace{-0.35em}
        \item a $1-\gamma$ upper prediction bound $\bar{Q}_{\tau}$ on the FDP in the list of discoveries returned by AS.
    \end{alglist}
}
\vspace{-0.5em}
\dottedhfill \\
\setstretch{1.4}
\If{$\tau = 0$}{
    \Return $\bar{Q}_\tau = 0$\;
}
$B := c/(1-\lambda)$\;
$C := -\log(\gamma) / \log\left(1 + \frac{1-\gamma^{B}}{B}\right)$\;
\For{$i=1$ \KwTo $m$} {
    $D_i$ be the number of $-1$'s in $\{L_1,\ldots,L_i\}$\;
    $T_i$ be the number of $1$'s in $\{L_1,\ldots,L_i\}$\;
    $\bar{V}_i = \lfloor C(1+BD_i)\rfloor$\;
    $\bar{G}_i = (\max_{j=1,\ldots,i}\lceil T_j - \bar{V}_j \rceil) \vee 0$\;
}
\eIf{$T_{\tau} = 0$}{
    \Return $\bar{Q}_\tau = 0$\;
}{
    \Return $\bar{Q}_\tau = (T_\tau - \bar{G}_\tau) / (T_\tau \vee 1)$;
}
\end{algorithm}

\clearpage

\subsection{A Monte-Carlo Approximation of the confidence parameter}
\label{supsec:c_star_simulations}

We next present an efficient algorithm to approximate the $u_\gamma(\Delta)$, as defined in \eqref{eq:c_star}.

For simplicity, we consider only the single-decoy case $R = 1/2$, but note that the following results generalise to multiple decoys by using $R = (1-\lambda)/(c + 1 - \lambda)$. Consider a sequence $B_i$ of iid Bernoulli(1/2) RVs
and define $X_i:=\sum_{j=1}^i B_j$, $Y_i := i-X_i$, and $i_d=\inf\{i\,:\,Y_i=d\}$.
Note that $P(i_d<\infty)=1$ so we can also define $U_d := X_{i_d}$ and it is clear that 
$$U_d\sim NB(d,1/2).$$
For $d\in\N$ define the function $G_d:\N\mapsto[0,1]$ as
$$ G_d(i) := 1 - F_{NB(d,1/2)}(i-1),$$
where $F_{NB(d,1/2)}$ denotes the CDF of a $NB(d,1/2)$ RV, and let
$$
\tilde{U}_d := G_d(U_d) .
$$
Denote $\Delta = \Delta(d_{\max}) := \{1,\ldots,d_{\max}\}$.
We are interested in numerically approximating $u_\gam$ defined as
\[
u_\gamma = u_\gamma(\Delta) := \max_{u\in R_\Delta} P\left(\min_{d\in\Delta} \tilde{U}_d \le u \right)\le \gamma ,
\]
where $R_\Delta=\{P(NB(d,1/2) \ge k)\,:\,k\in\mathbb{N},d\in\Delta\}$ is the range of values $\{\tilde{ U}_d\}_{d\in\Delta}$ can attain.
With
$$
\mathcal{M}_d := \min_{k:k\leq d}\tilde{U}_k ,
$$
\[
u_\gamma(\Delta) = \max_{u\in R_\Delta} P\left(\mathcal{M}_{d_{\max}} \le u \right)\le \gamma ,
\]
showing that $u_\gamma(\Delta)$ is essentially a $\gamma$-quantile of $\mathcal{M}_{d_{\max}}$.

The following procedure uses Monte-Carlo simulations to simultaneously approximate these quantiles
for all values $d=d_{\max}\le d_0$ (in practice we used $d_0=50000$) and for commonly used confidence levels $1-\gam$
(we used 50\%, 80\%, 90\%, 95\%, 97.5\% and 99\%).

For $j$ in $1,\ldots,N$, where $N$ is the total number of Monte-Carlo simulations (we used $N=2 \cdot 10^6$), 
we draw a sequence of iid $\operatorname{Bernoulli}(1/2)$ RVs
$B_1^j,B_2^j,\ldots,B_n^j$ until we have $d_0$ failures, i.e., until the first $n=n(j)$ for which
$n - \sum_{k=1}^n B_k^j = d_0$.
For $d=1,2,\dots,d_0$ we define
$$
i_d^j = \min\{n\in\mathbb{N}:n-\sum_{k=1}^n B_k^j = d\}
$$
as a realization of the RV $i_d$, and we compute a ``path" $U_1^j, U_2^j,\ldots, U_{d_0}^j$ as
\[
U_d^j = \sum_{i=1}^{i_d^j}B_i^j \qquad d=1,2,\dots d_0.
\]
Then, for each such sampled path we inductively compute the cumulative minimum $\mathcal{M}_d^j$ using
$$
\mathcal{M}_d^j = \min\big\{\tilde{U}_d,\mathcal{M}_{d-1}\big\} = \min\big\{G_d(U_d^j),\mathcal{M}_{d-1}^j\big\}  \qquad d=1,2,\dots d_0,
$$
where $\mathcal{M}_0^j\coloneqq1$.

Let $[N]\coloneqq\{1,2,\dots,N\}$, and for $d=d_{\max}\in\{1,2,\dots,d_0\}$ let $S_d\coloneqq\left\{\mathcal{M}_d^j:j\in[N]\right\}$ be the set of all
observed values of $\mathcal{M}_d^j$ across our $N$ MC samples. If $N$ is large enough then practically (i.e., with probability $\approx1$)
there are values $\rho_d,\sig_d\in S_d$ such that
\begin{itemize}
    \item $\rho_d < \sig_d$,
    \item $(\rho_d , \sig_d) \cap S_d = \emptyset$, and
    \item $\left|\{j\in[N]:\mathcal{M}_d^j\leq \rho_d\}\right| \le \gamma N < \left|\{j\in[N]:\mathcal{M}_d^j\leq \sig_d \}\right|$ .
\end{itemize}

Ignoring the discrete effect, we may conservatively take $\tilde{u}_\gamma(\Delta_d)= \rho_d$ as an estimate for $u_\gamma(\Delta_d)$.
Alternatively, we can again introduce randomization to get a small power boost. Specifically, let
$r_d = |\{j\in[N]:\mathcal{M}_d^{j}\leq \rho_d\}| / N$ and  $s_d = |\{j\in[N]:\mathcal{M}_d^{j}\leq \sig_d\}| / N$.
With $w^j_d = (\gam-r_d)/(s_d-r_d)$, we have $w_d r_d+(1-w_d)s_d=\gam$. 
Then, given a sample, we flip a coin to determine whether to use $\tilde{u}_\gamma(\Delta(d=d_{\max}))= \rho_d$ (with probability $w_d$)
or $\tilde{u}_\gamma(\Delta(d=d_{\max}))= \sig_d$ (with probability $1-w_d$). All the results in the paper were obtained using this
randomized version.

\subsection{Peptide detection}
\label{supsec:PXIDs}

We downloaded 10 MS/MS spectrum files from the Proteomics Identifications Database, PRIDE \cite{martens:pride}.
Each spectrum file was obtained by iteratively and randomly selecting PRIDE projects that were submitted no later than 2018,
and then randomly selecting an \texttt{mgf} file from each of these projects. If an \texttt{mgf} file was found,
param-medic was then used to check whether the selected file was high-resolution and that no variable modifications were detected for simplicity \cite{may:detecting}.
If not, the next project was selected until 10 such spectrum files were acquired.
For each spectrum file, a protein \texttt{FASTA} file was obtained from the associated project in PRIDE with the exception of human data,
which we used the UniProt database UP000005640 (downloaded 9/11/2021). Table \ref{suptable:PRIDE} reports the list of spectrum files used.

For each of the 10 MS/MS spectrum files, we used Tide-index to digest the corresponding \texttt{FASTA} files
and to generate 20 randomly shuffled decoy databases using the default settings.
For each spectrum file, we used Tide-search to conduct separate searches of the target database, and each of the 20 decoy databases with the options
\texttt{--auto-precursor-window warn --auto-mz-bin-width warn}.
Only the top XCorr scoring PSM for each scan in the output search files was considered. All other options were set to their default values.
Tide was implemented in Crux v4.1.decd99ff \cite{park:rapid,mcilwain:crux,diament:faster}.

In the averaging process we randomly selected 20 sets of $d$ decoy search files, out of the 20 available files, making sure each set of $d$ files
is unique (so for $d=1$ each search file was selected exactly once). Given the target search file, and a selected set of $d$ decoy search files
we kept for each spectrum only its best matching PSM across the $d+1$ search files. We then assigned each target or decoy peptide a score, which
is that of the maximal (of the remaining) PSMs associated with that peptide (or $-\infty$ if there was no such PSM).
We then generated a list of discovered peptides by applying the max and the mirror methods for controlling the FDR on the resulting
set of target and $d$ decoy scores per target peptide. Next, the three FDP bounding procedures TDC-SB/UB/KRB were applied at confidence
level $1-\gam=0.95$, and finally we averaged the computed bounds over the 20 selected sets of $d$ decoys.

\begin{center}
\begin{table}
\begin{tabular}{|p{3cm}|p{13cm}|}
\hline
\multicolumn{1}{|c|}{Project ID} & \multicolumn{1}{|c|}{Spectrum file}\\
\hline
PXD008920 & QEP1\_ZADAM\_deg1\_2\_1\_170615.mzid\_QEP1\_ZADAM\_deg1\_2\_1\_
170615.mgf\\
PXD008996 & A431\_01uM\_DON\_3\_2.mgf  \\
PXD010504 & QX05437.mgf  \\
PXD014277 & Q06965\_MS18-017\_37\_J9(55).mgf  \\
PXD016274 & 190322-SI-0149-F5-01.mgf  \\
PXD019354 & Progenesis-Trophoplast-top5-140416.mgf  \\
PXD024284 & p2830\_NaClvsHank\_mascot.mgf  \\
PXD025130 & Q26431\_MS20-025\_Virus-purif\_1.mgf  \\
PXD029319 & MP\_16072020\_LvN\_S\_layer\_gps\_5\_DDA01.mgf  \\
PXD030118 & Q26756\_MS20-013\_C15.mgf  \\
\hline
\end{tabular}
\caption{\textbf{The PRIDE data} The list of 10 spectrum files used and their associated project IDs.}
\label{suptable:PRIDE}
\end{table}
\end{center}


\vspace{2em}
\subsection{Supplementary Figures}

\begin{figure}[h]
\centering
\begin{tabular}{ll}
\includegraphics[width=3in]{{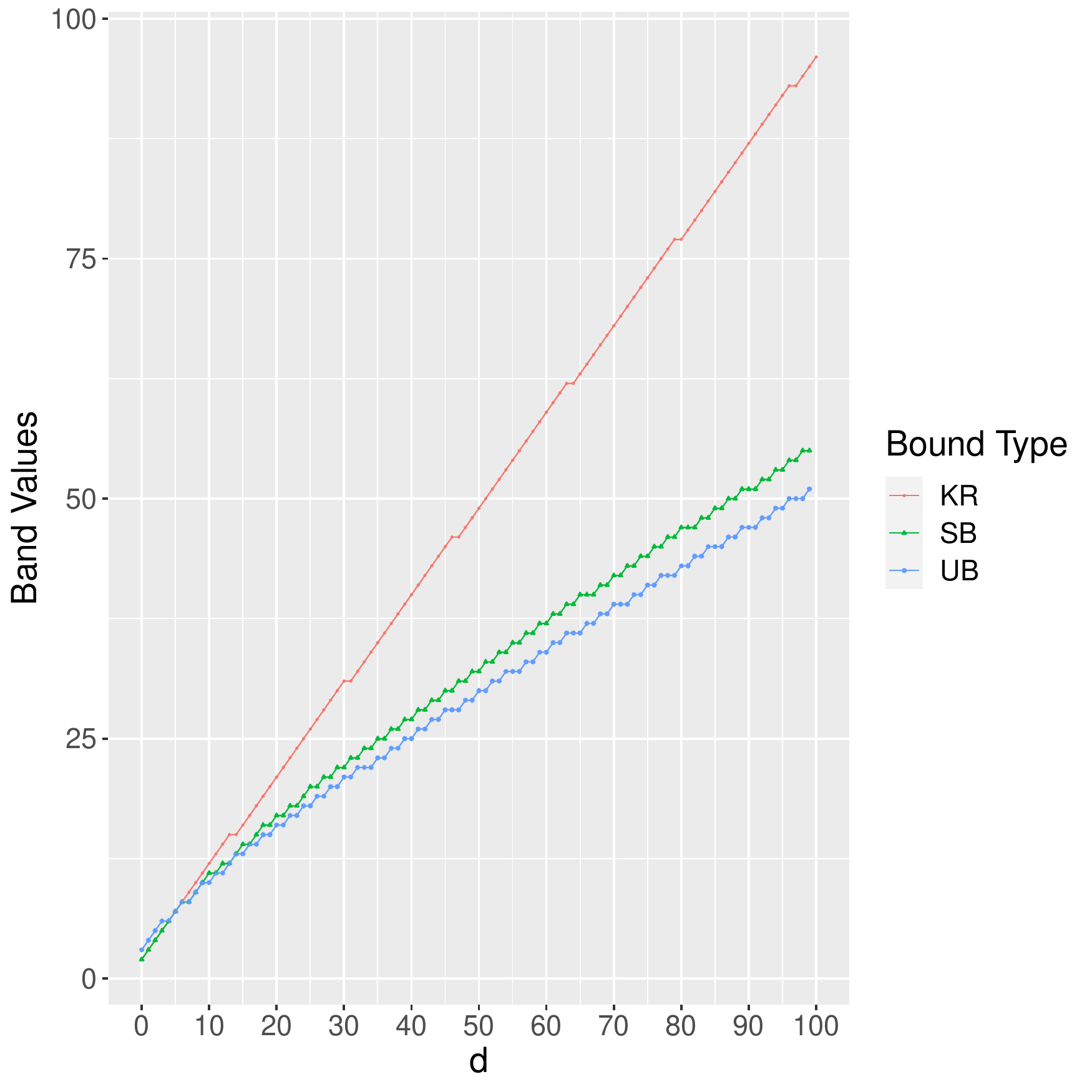}} &
\includegraphics[width=3in]{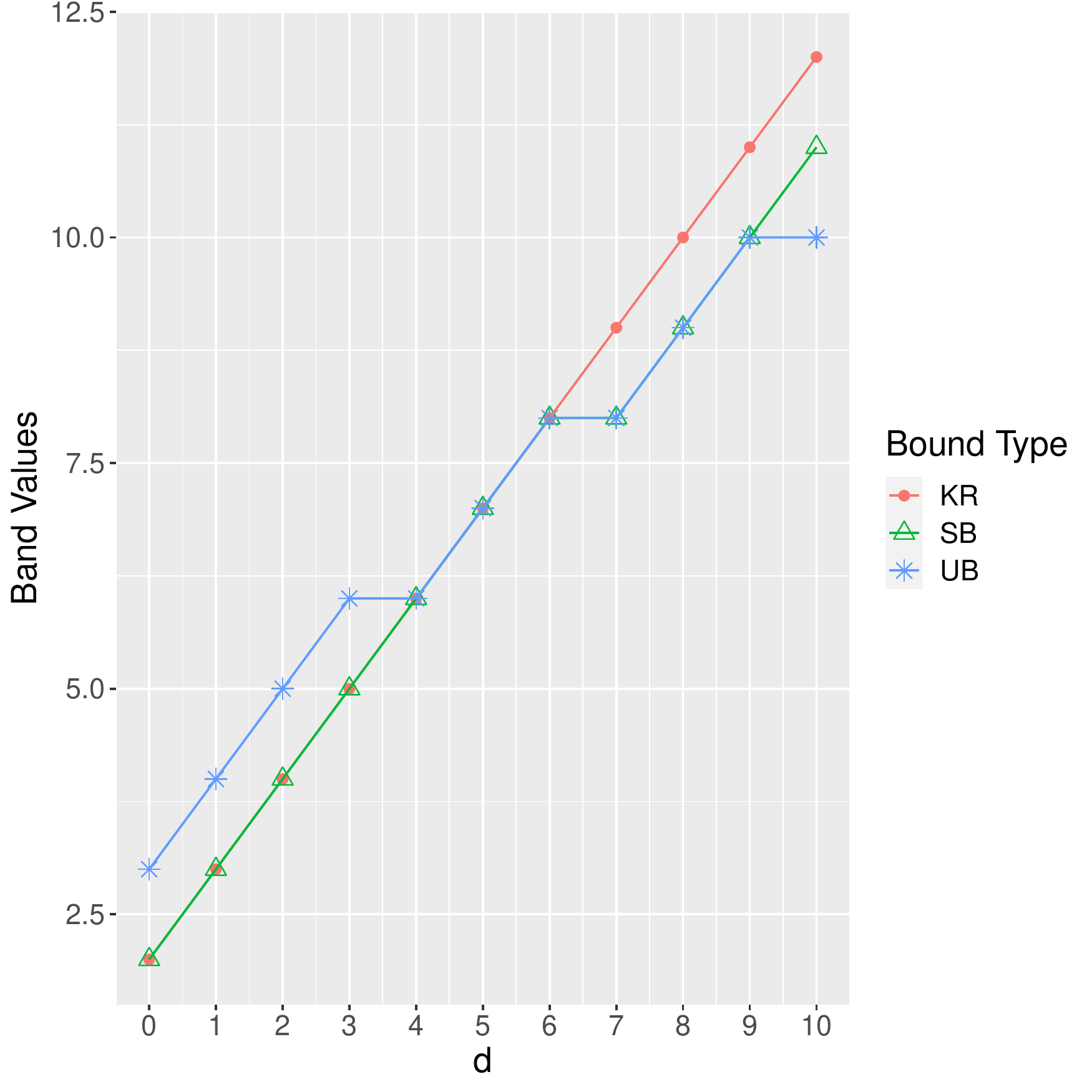} \\
\includegraphics[width=3in]{{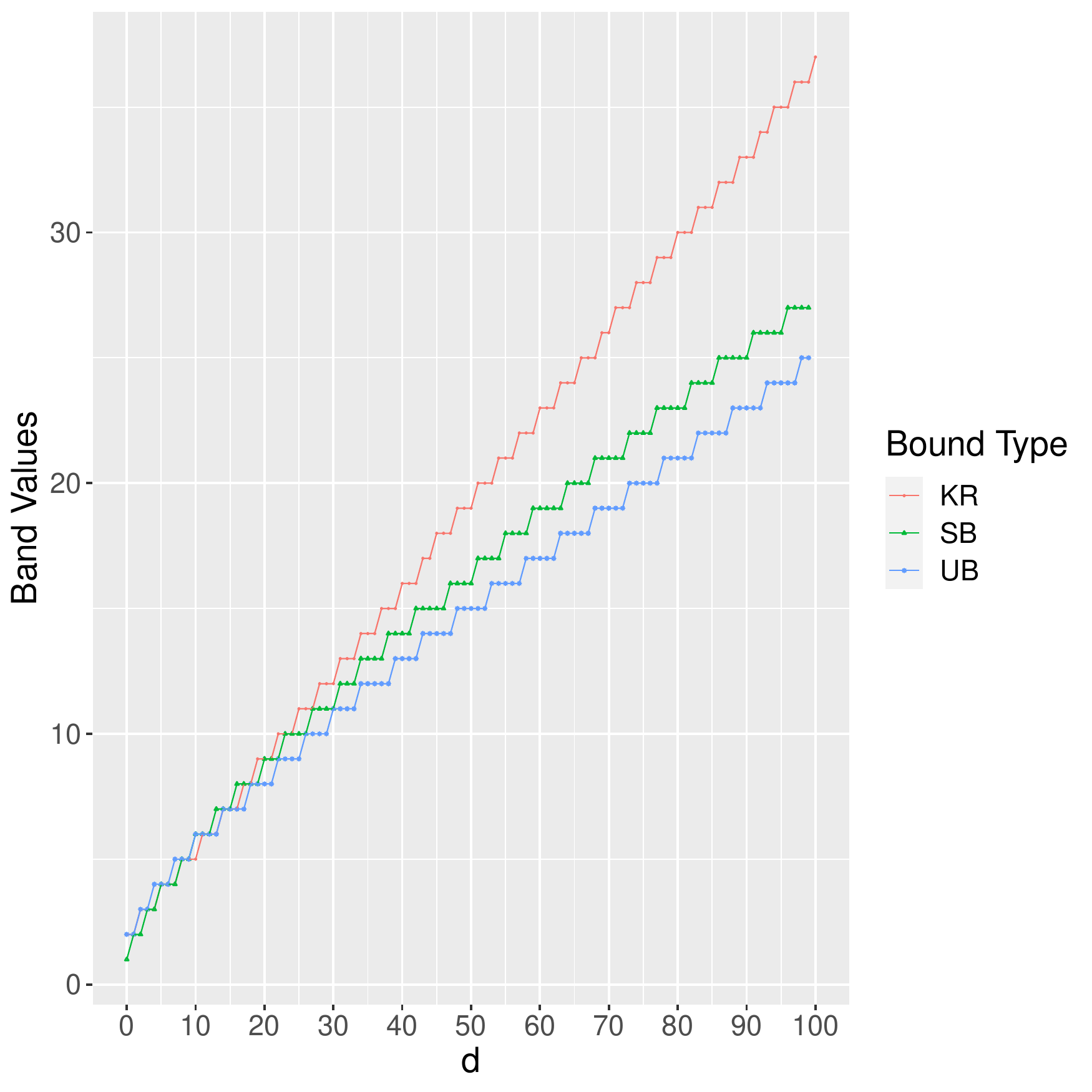}} &
\includegraphics[width=3in]{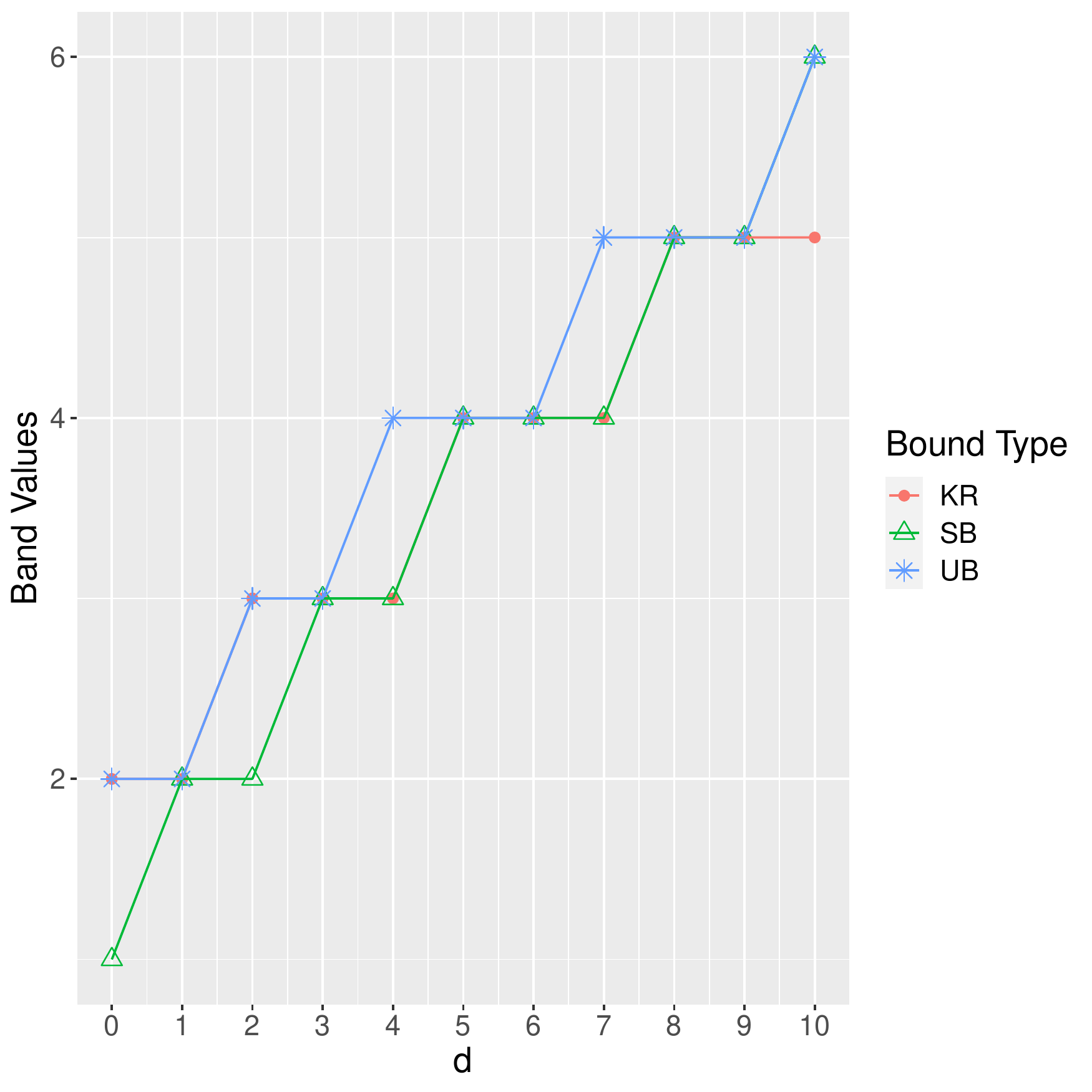}
\end{tabular}
\caption{\textbf{Deterministically comparing the $\bar V\supSB$, $\bar V\supUB$, and $\bar V\supKR$ bands with $\mathbf{B = 1/3,1/7}$.}
\label{supfig:det_comp2} For $d \in \{1, \ldots, 100\}$ we computed the value of the upper prediction bands for the number of false discoveries as described in the text with $\dmax\coloneqq100$. We set $B=1/3$ in the top row (max method with 3 decoys), and $B=1/7$ in the bottom row (max method with 7 decoys). The right figures are zoomed-in version of the left figures for small values of $d$.}
\end{figure}

\clearpage

\begin{figure}
\centering %
\begin{tabular}{ccc}
\hspace{-10ex}
$\color{blue}m = 500$, $\pi_0 = 0.5$, $\rho = 3$  & $\color{blue}m = 2000$, $\pi_0 = 0.5$, $\rho = 3$ & $\color{blue}m = 10000$, $\pi_0 = 0.5$, $\rho = 3$ \tabularnewline
\hspace{-10ex}
\includegraphics[width=2.5in]{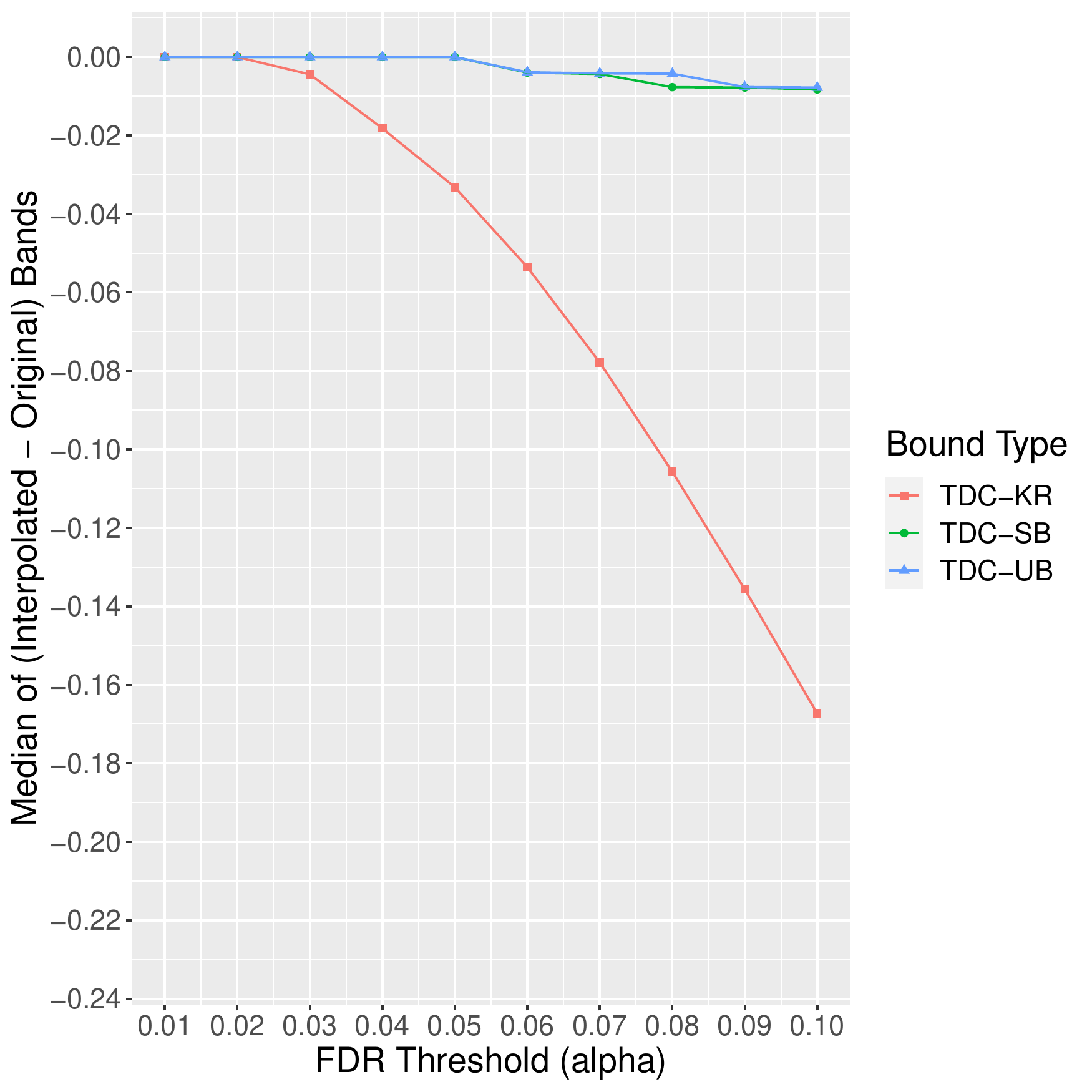}  & \includegraphics[width=2.5in]{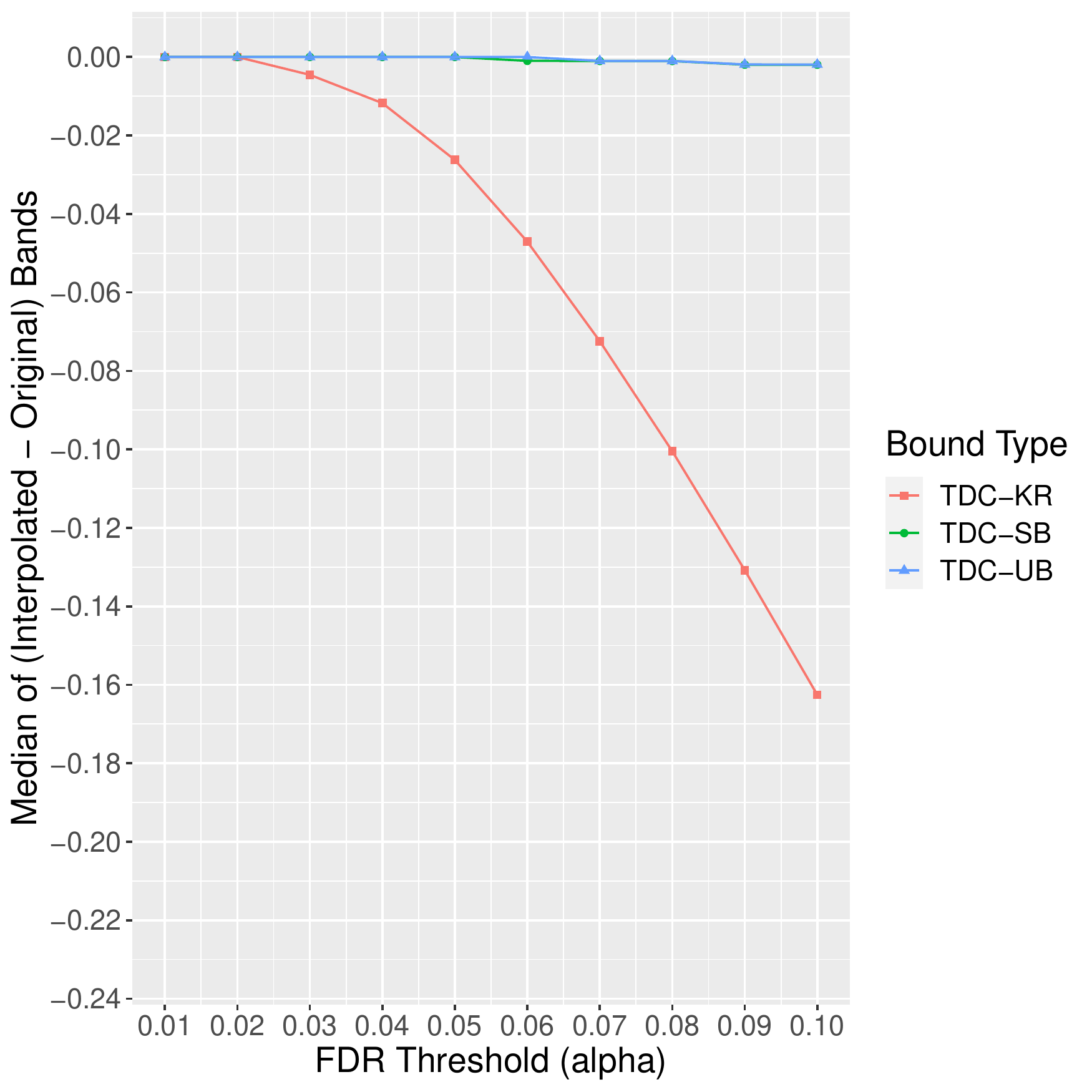}  & \includegraphics[width=2.5in]{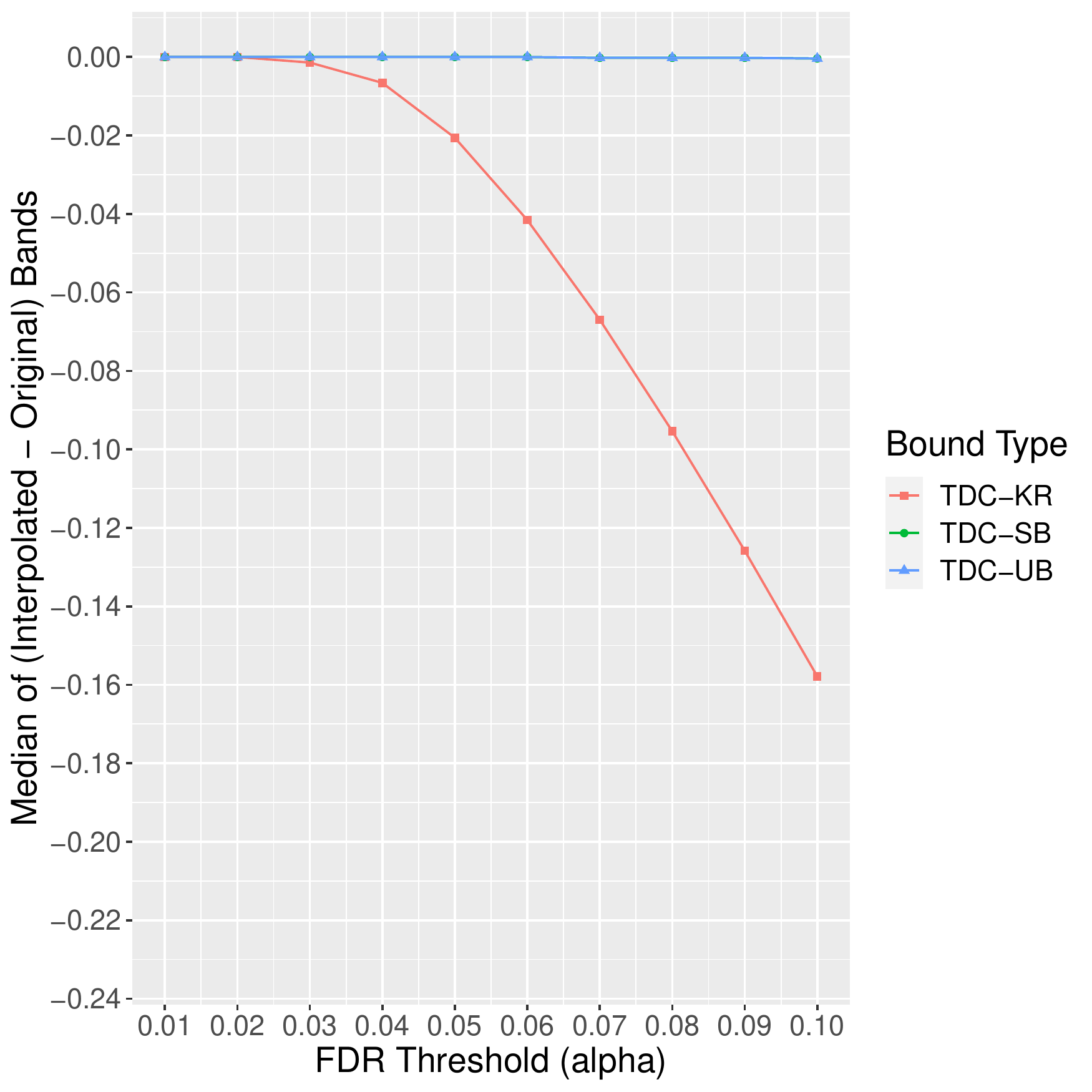} \tabularnewline
\hspace{-10ex}
$m = 2000$, $\color{blue}\pi_0 = 0.8$, $\rho = 3$  & $m = 2000$, $\color{blue}\pi_0 = 0.5$, $\rho = 3$ & $m = 2000$, $\color{blue}\pi_0 = 0.2$, $\rho = 3$ \tabularnewline
\hspace{-10ex}
\includegraphics[width=2.5in]{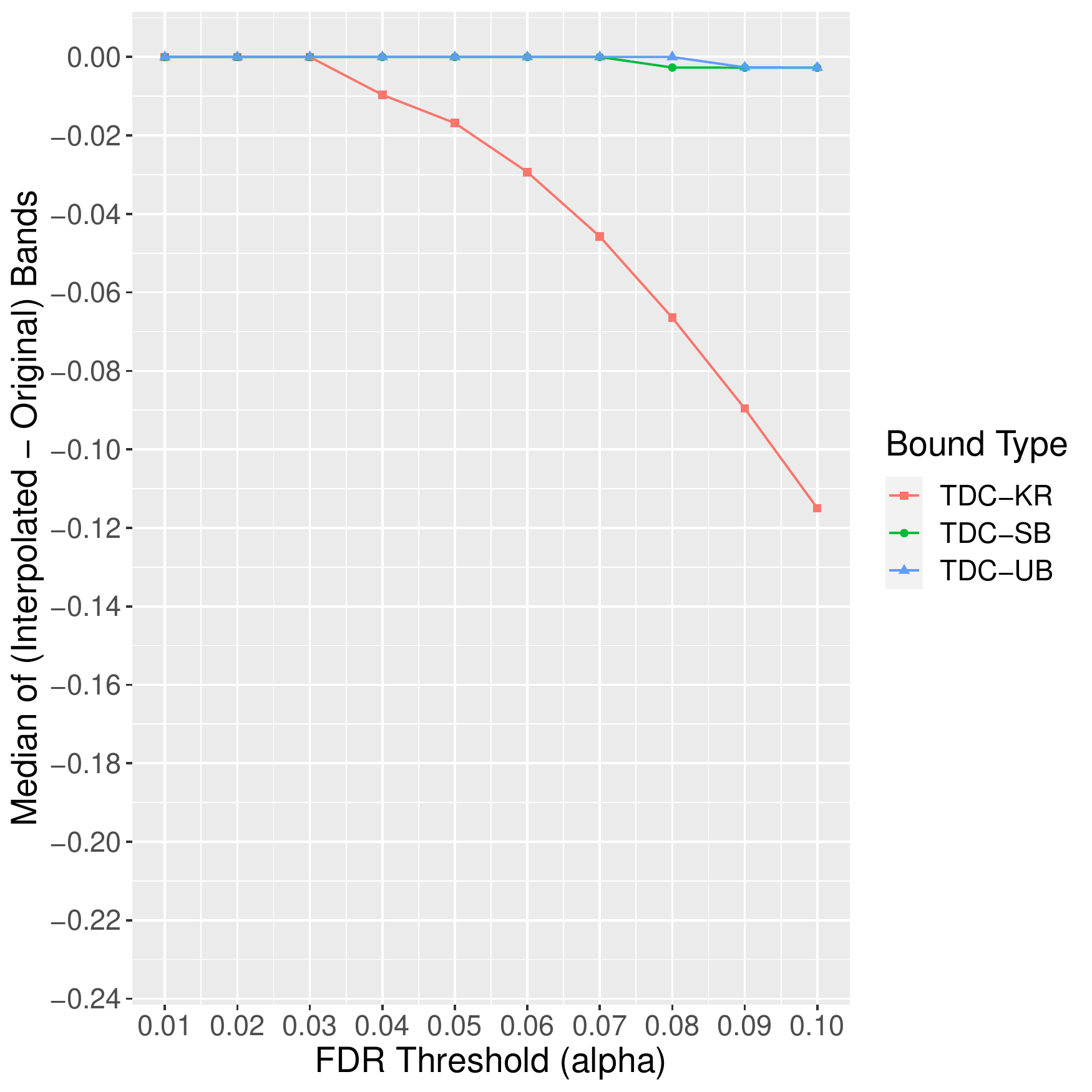}  & \includegraphics[width=2.5in]{{figures_Arya/interpolation_figure_m2000_pi0.5_calTRUE_sep3_alpha0.05_conf0.05}.pdf}  & \includegraphics[width=2.5in]{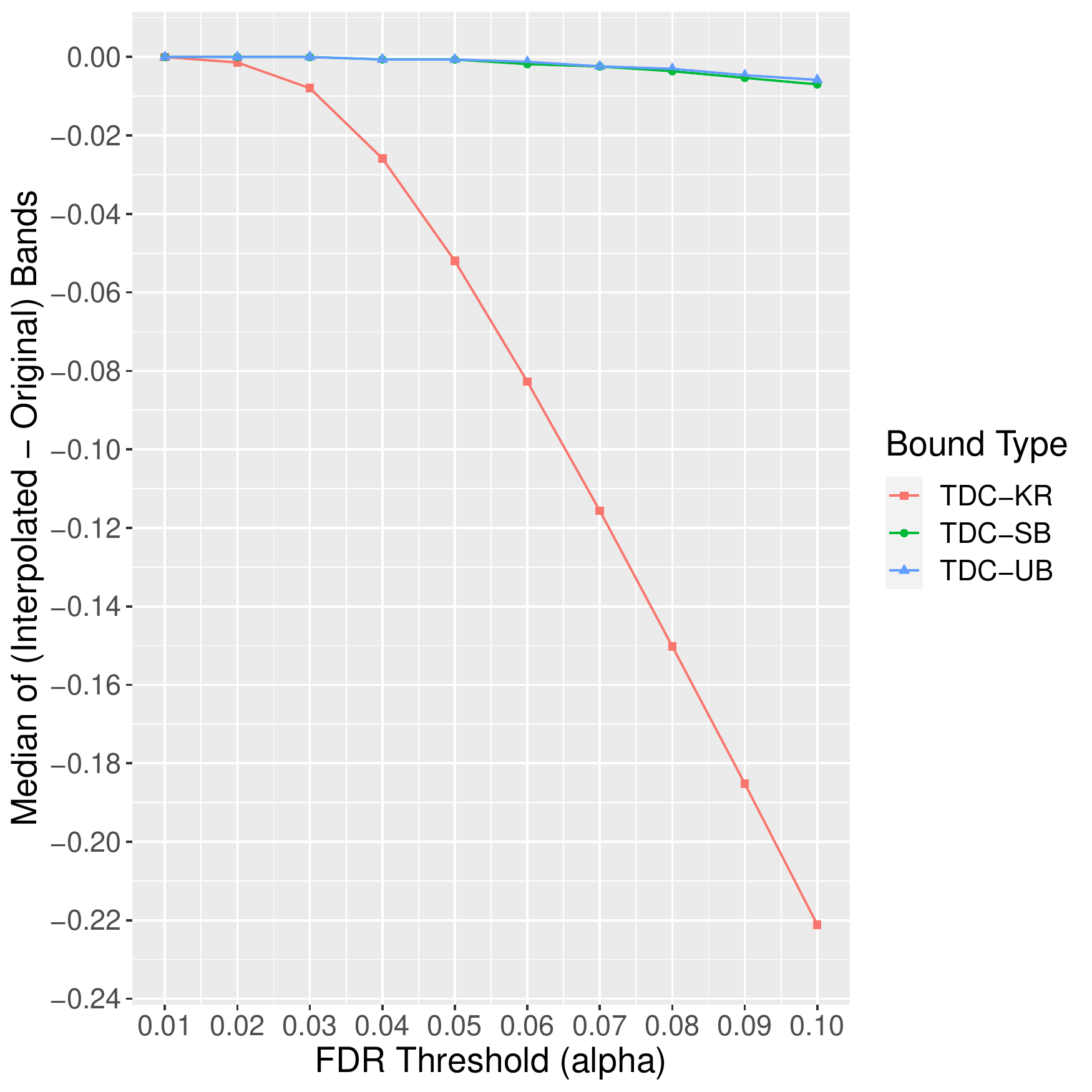} \tabularnewline
\hspace{-10ex}
$m = 2000$, $\pi_0 = 0.5$, $\color{blue}{\rho = 2.5}$  & $m = 2000$, $\pi_0 = 0.5$, $\color{blue}{\rho = 3}$ & $m = 2000$, $\pi_0 = 0.5$, $\color{blue}{\rho = 3.5}$ \tabularnewline
\hspace{-10ex}
\includegraphics[width=2.5in]{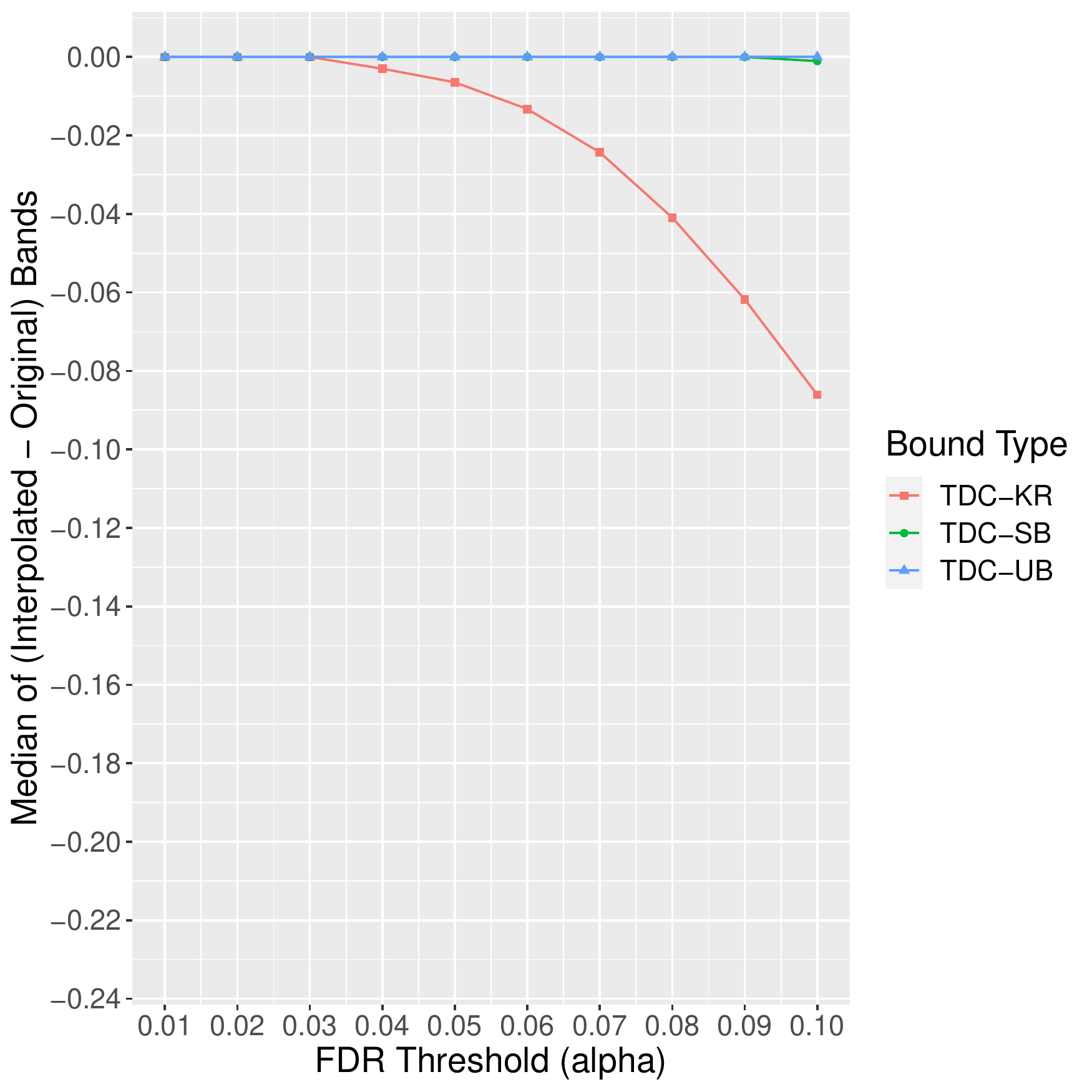}  & \includegraphics[width=2.5in]{{figures_Arya/interpolation_figure_m2000_pi0.5_calTRUE_sep3_alpha0.05_conf0.05}.pdf}  & \includegraphics[width=2.5in]{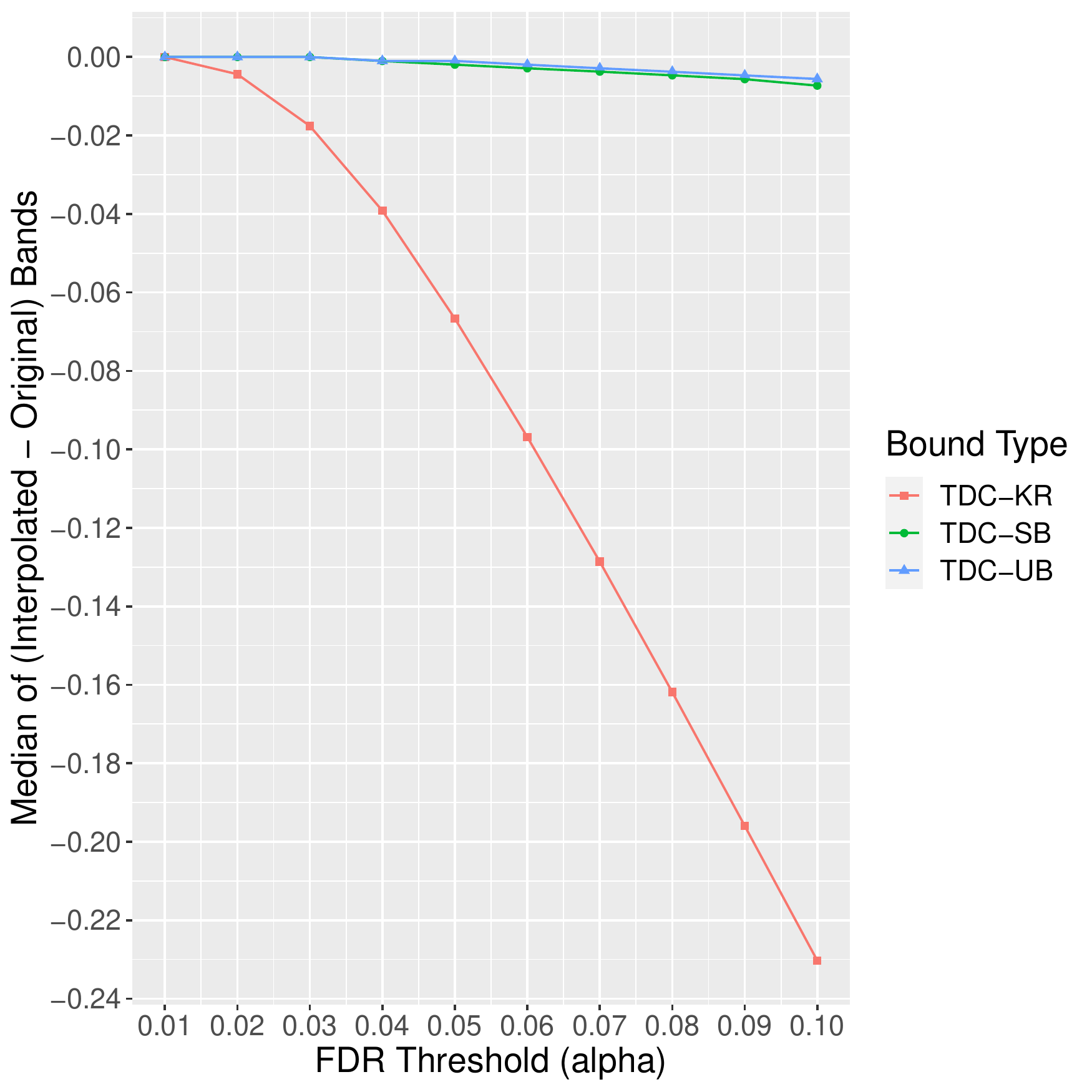} \tabularnewline
\end{tabular}
\caption{\textbf{Comparing the interpolated and non-interpolated bands.} 
Using several different parameter combinations (each specified on top of its panel) we looked at the median (over 20k datasets) of the difference between
the interpolated and non-interpolated bound on TDC's FDP.
\label{supfig:interpolation}}
\end{figure}

\clearpage

\begin{figure}
\centering %
\begin{tabular}{ccc}
\hspace{-10ex}
$m = 500$  & $m = 2000$ & $m = 10000$ \tabularnewline
\hspace{-10ex}
\includegraphics[width=2.5in]{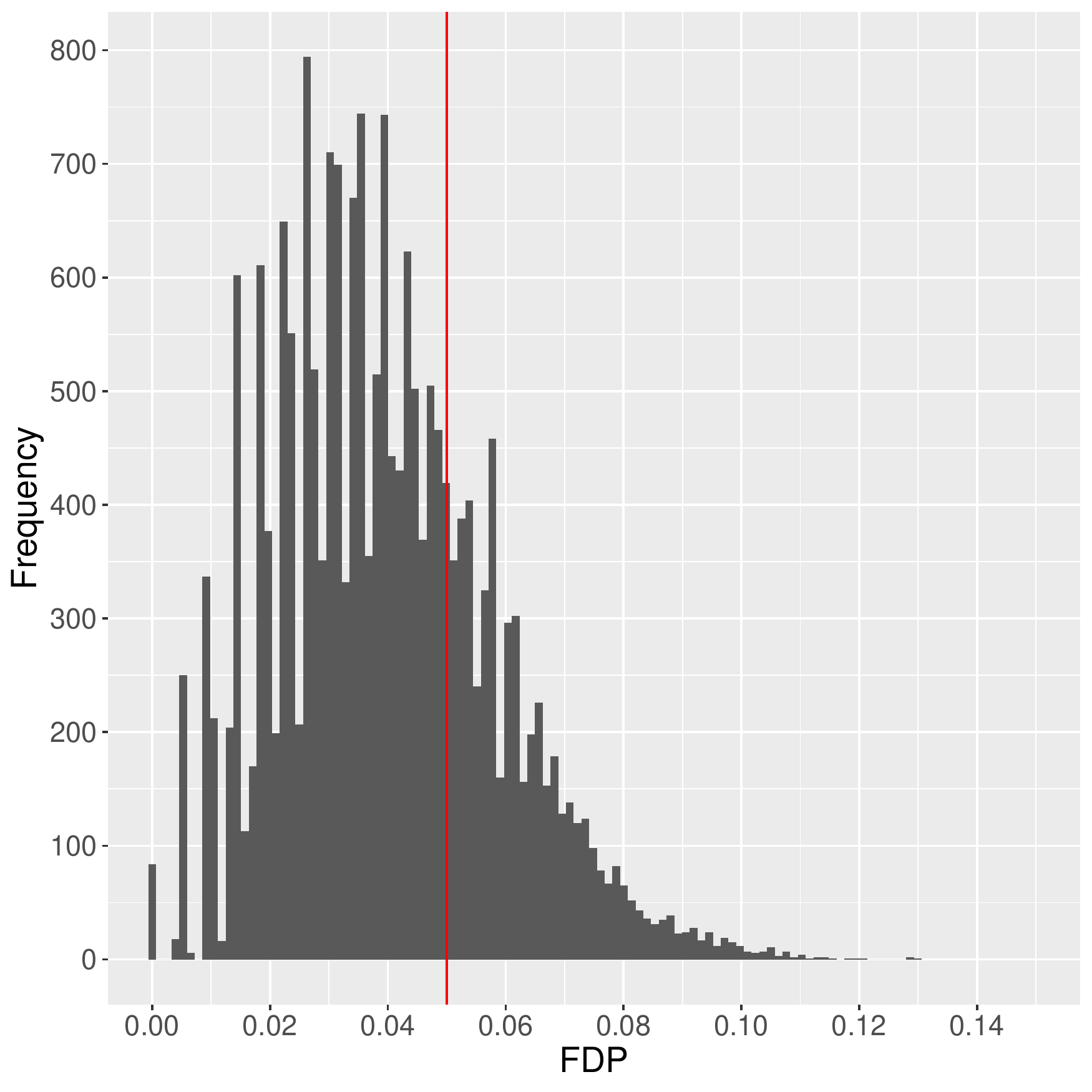}  & \includegraphics[width=2.5in]{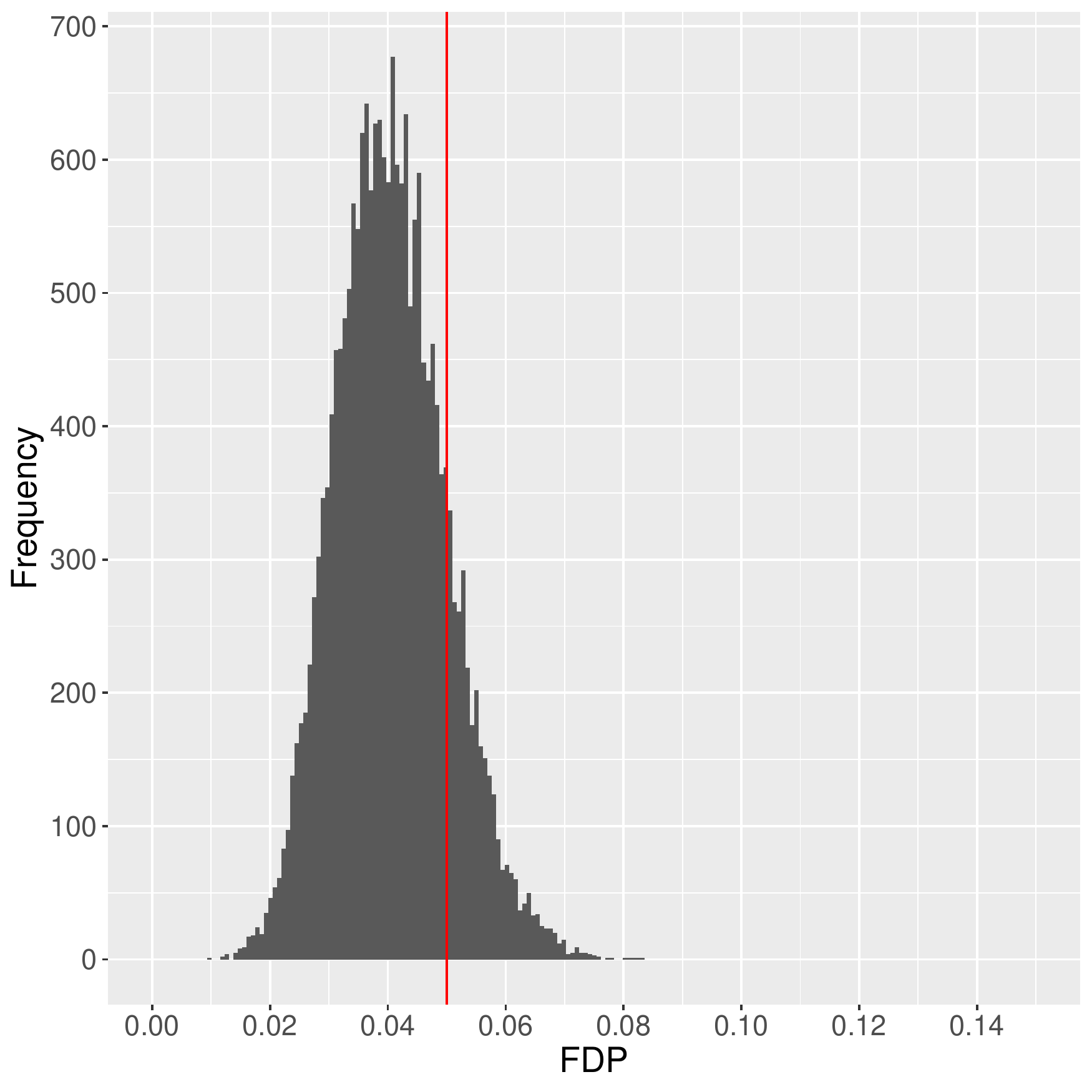}  & \includegraphics[width=2.5in]{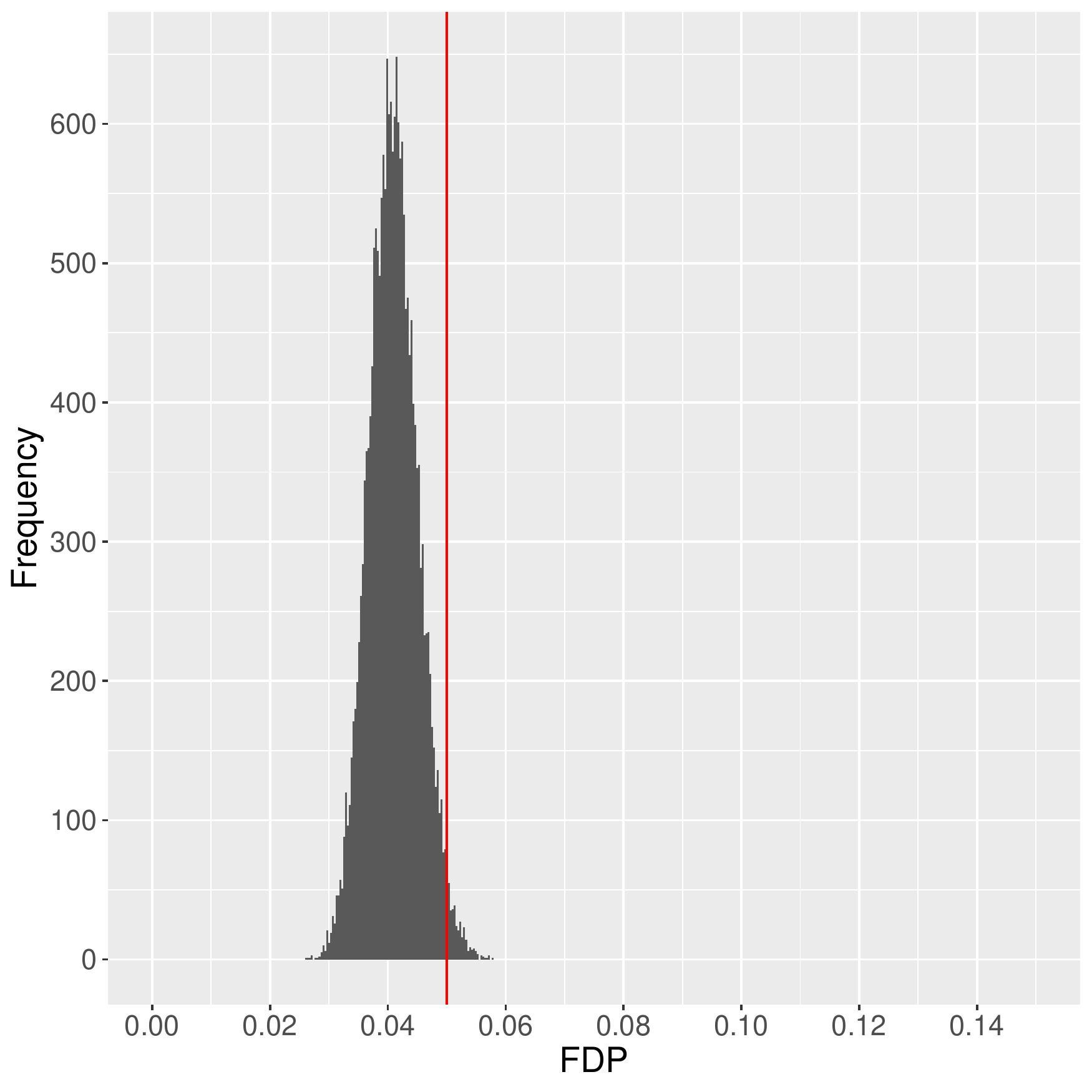} \tabularnewline
\hspace{-10ex}
\includegraphics[width=2.5in]{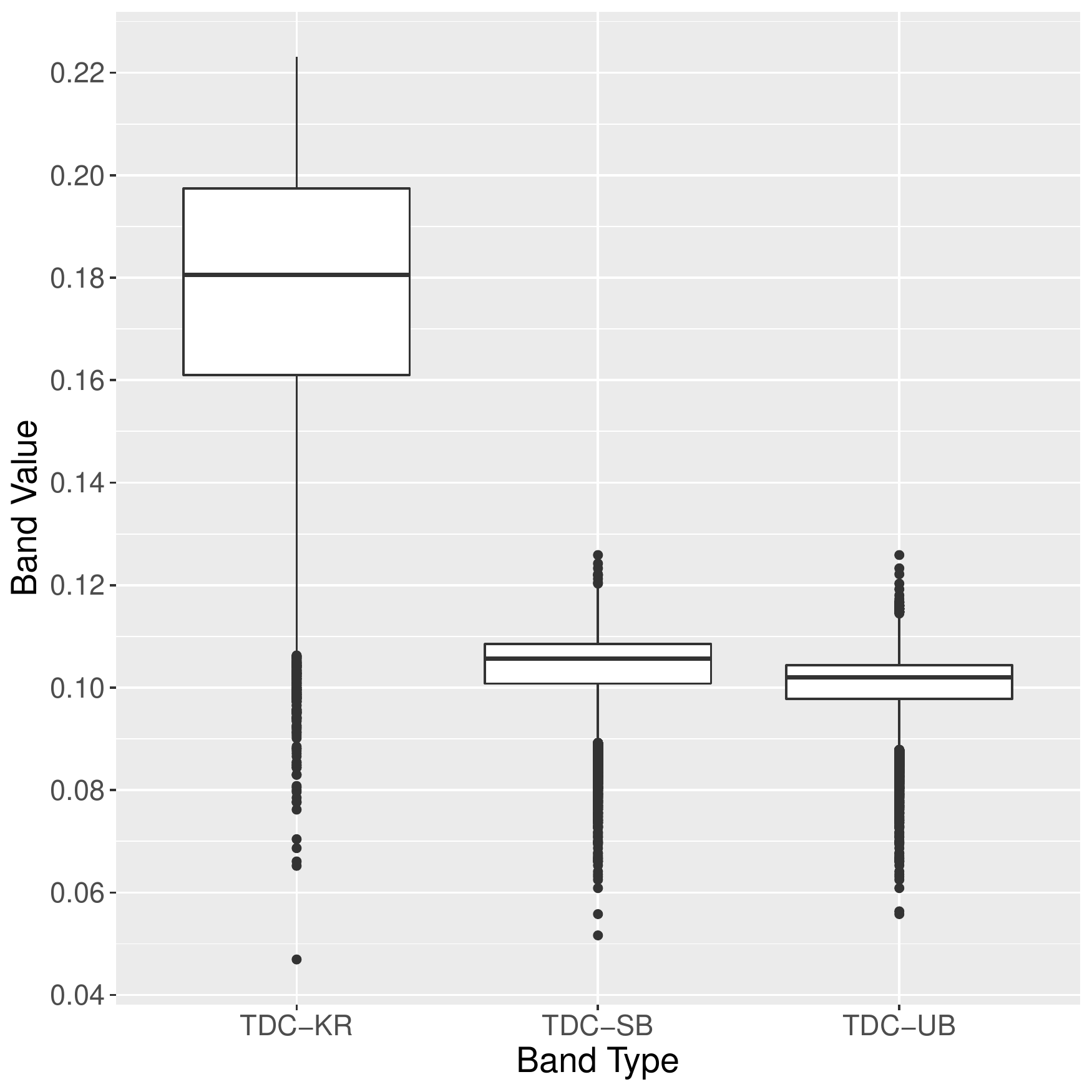}  & \includegraphics[width=2.5in]{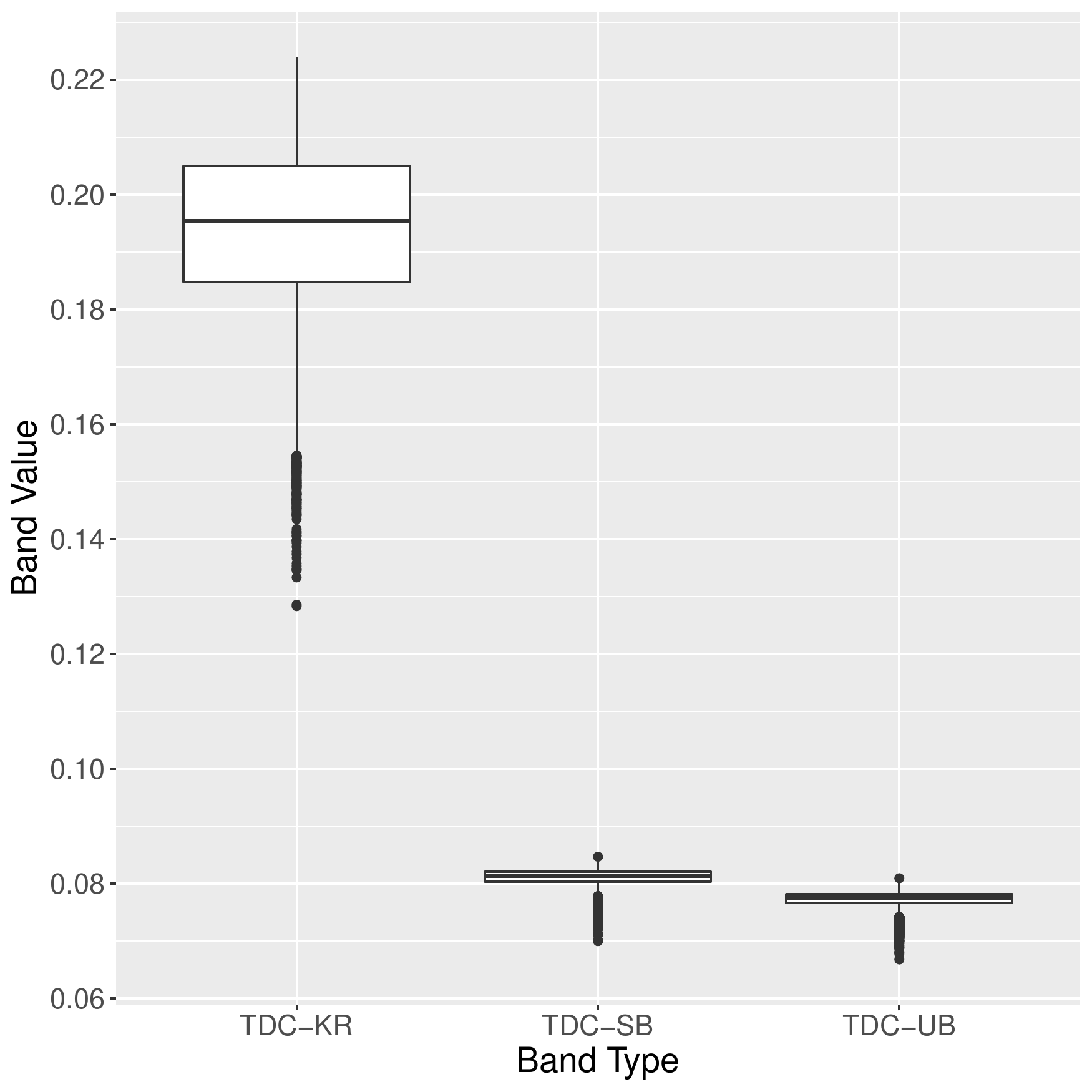}  & \includegraphics[width=2.5in]{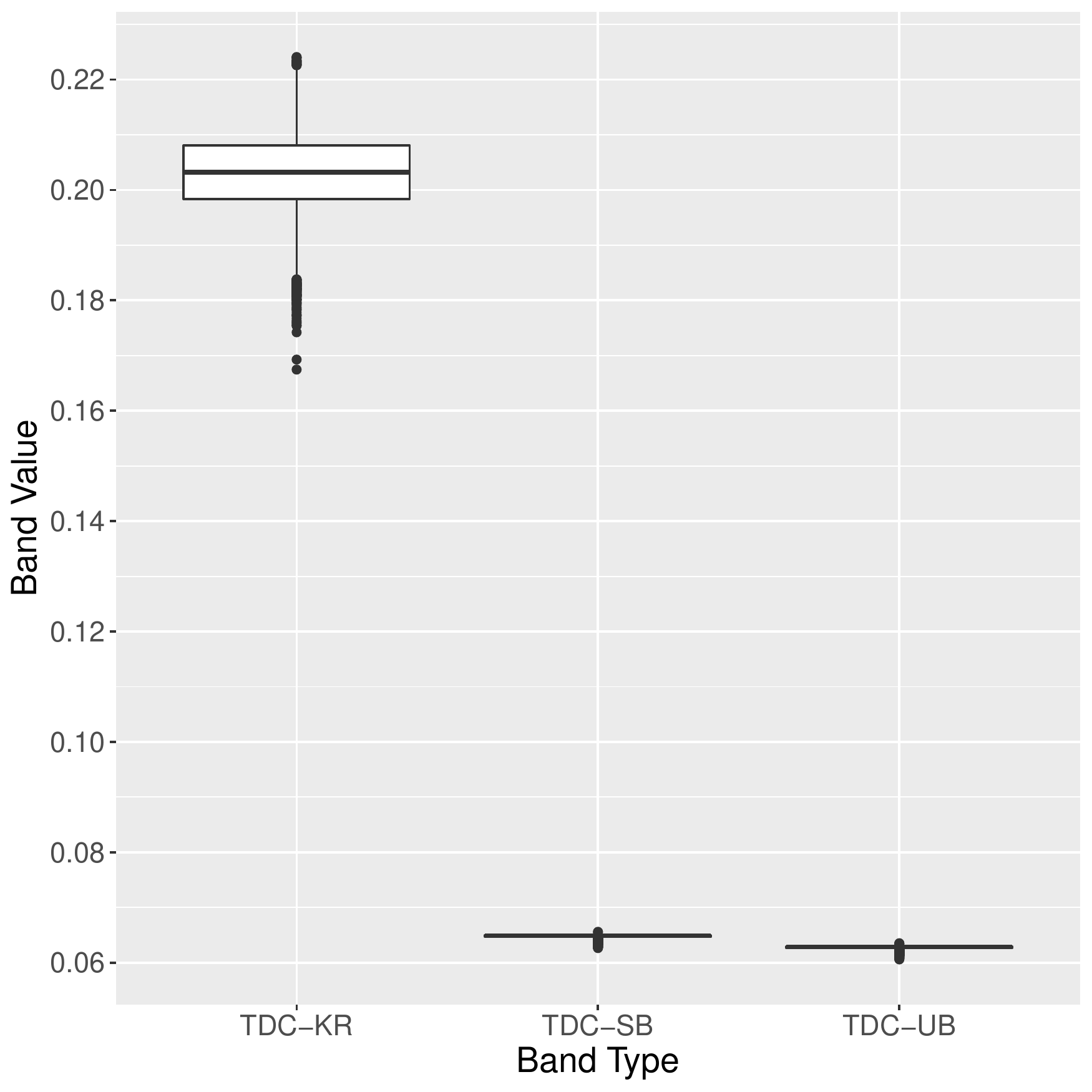} \tabularnewline
\hspace{-10ex}
\includegraphics[width=2.5in]{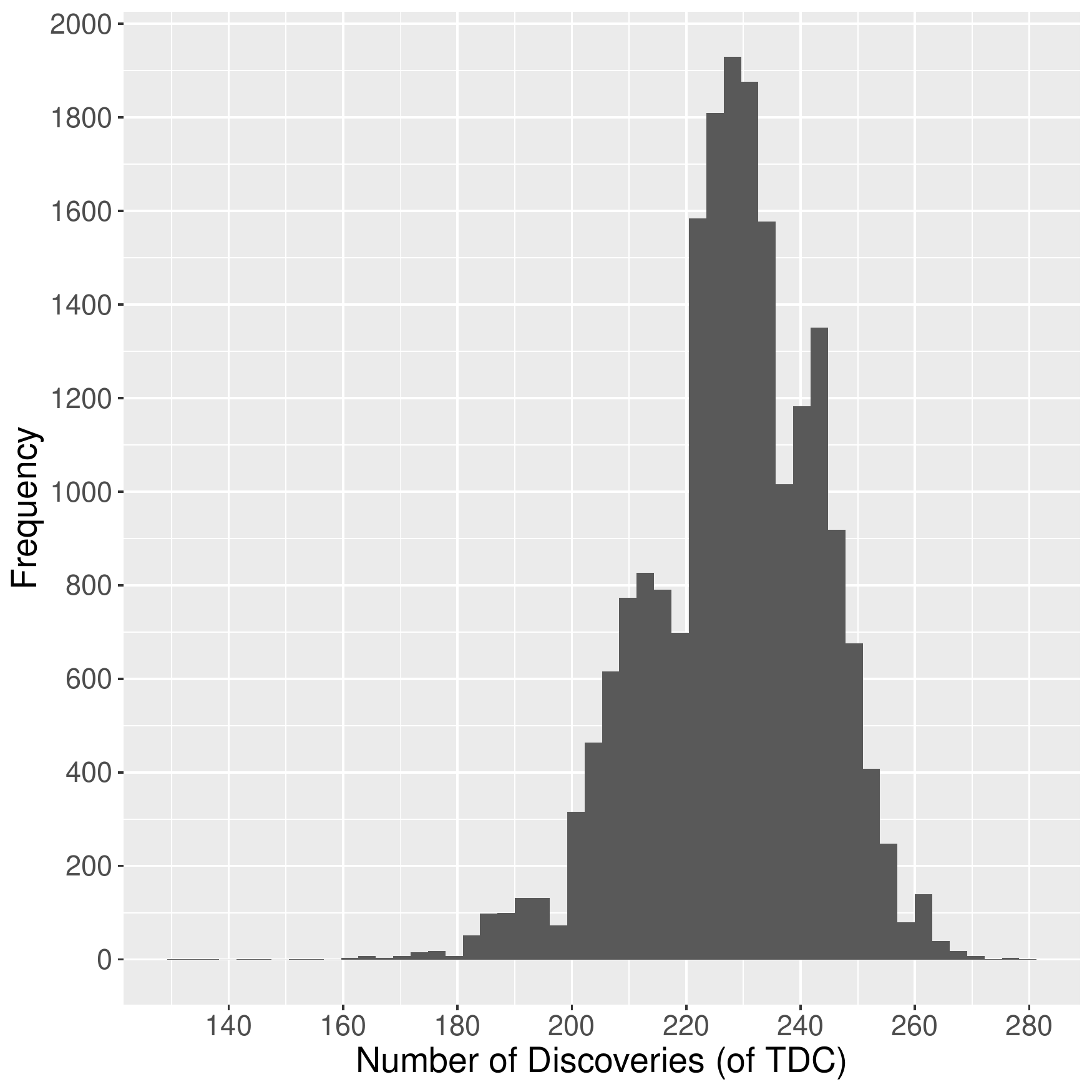}  & \includegraphics[width=2.5in]{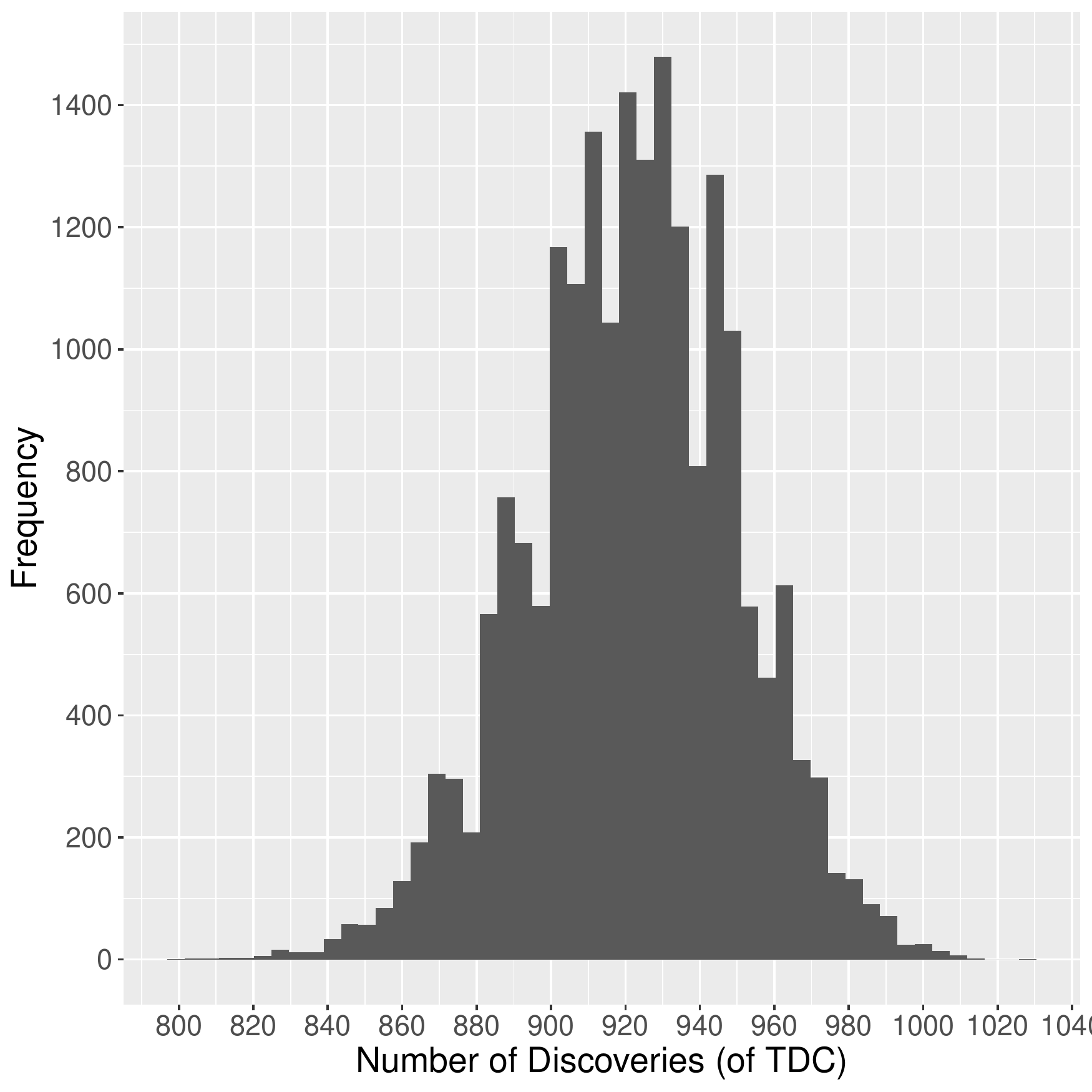}  & \includegraphics[width=2.5in]{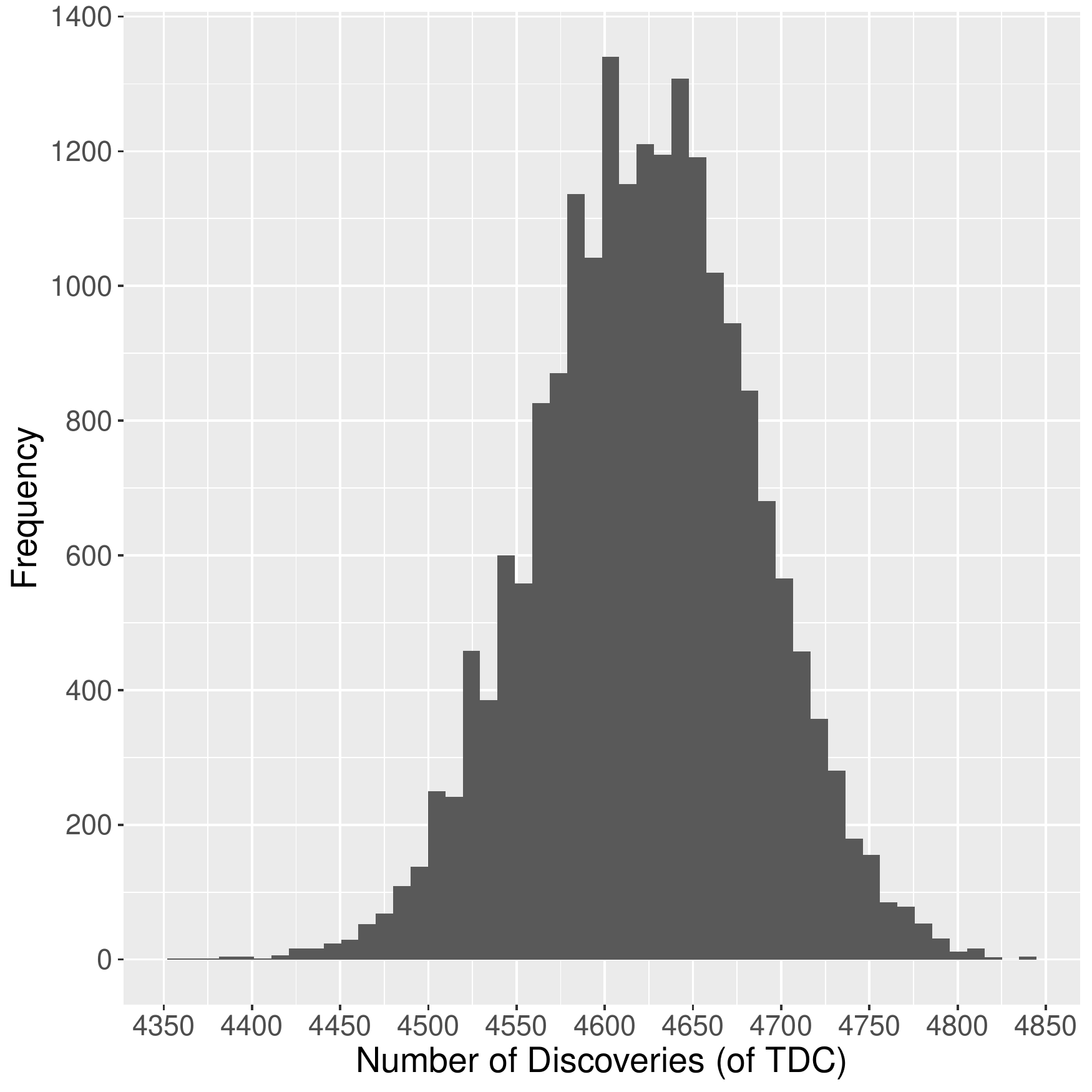} \tabularnewline
\end{tabular}
\caption{\textbf{Varying $m$ in simulated experiments.} Using 20K calibrated scores, we increase $m$ from 500 (left column) through 2K (middle column)
to 10K (right column) while looking at TDC's FDP (top row), the value of the FDP bound returned by the interpolated versions of TDC-KR, TDC-SB and TDC-UB (middle row) and the number of discoveries returned by TDC (bottom row).
The parameters $\alp=0.05,\gam=0.05$, $\rho = 3$ and $\pi_0=0.5$ were kept constant.  We observe the same trend of increasing TDC-KR bounds
and decreasing TDC-SB/UB bounds for other values of $\pi_0$ and $\rho$, including those studied in subsequent Figures \ref{supfig:vary_m_pi0_0.2} and \ref{supfig:vary_m_pi0_0.8}.
\label{supfig:vary_m}}
\end{figure}

\clearpage

\begin{figure}
\centering %
\begin{tabular}{ccc}
\hspace{-10ex}
$m = 500$  & $m = 2000$ & $m = 10000$ \tabularnewline
\hspace{-10ex}
\includegraphics[width=2.5in]{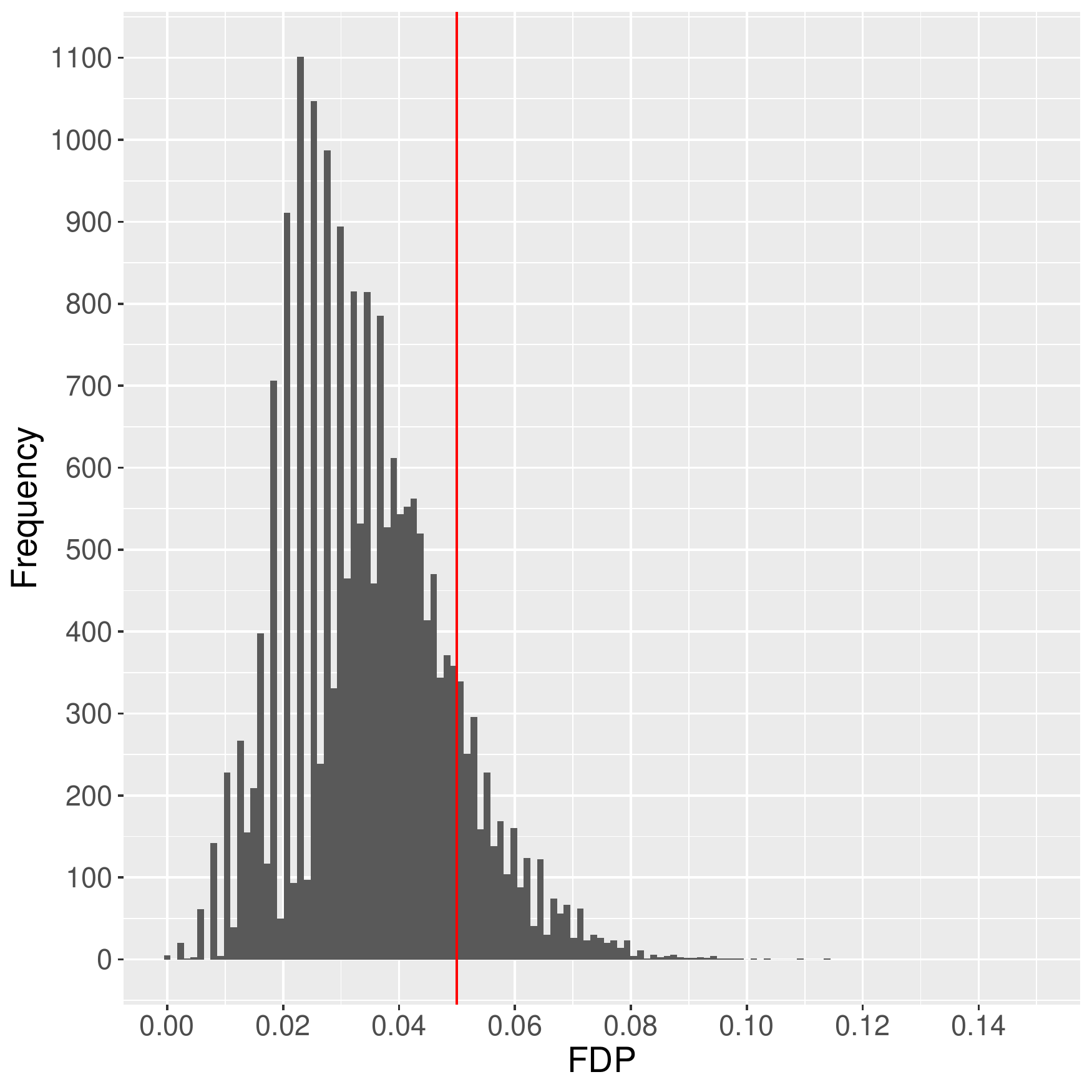}  & \includegraphics[width=2.5in]{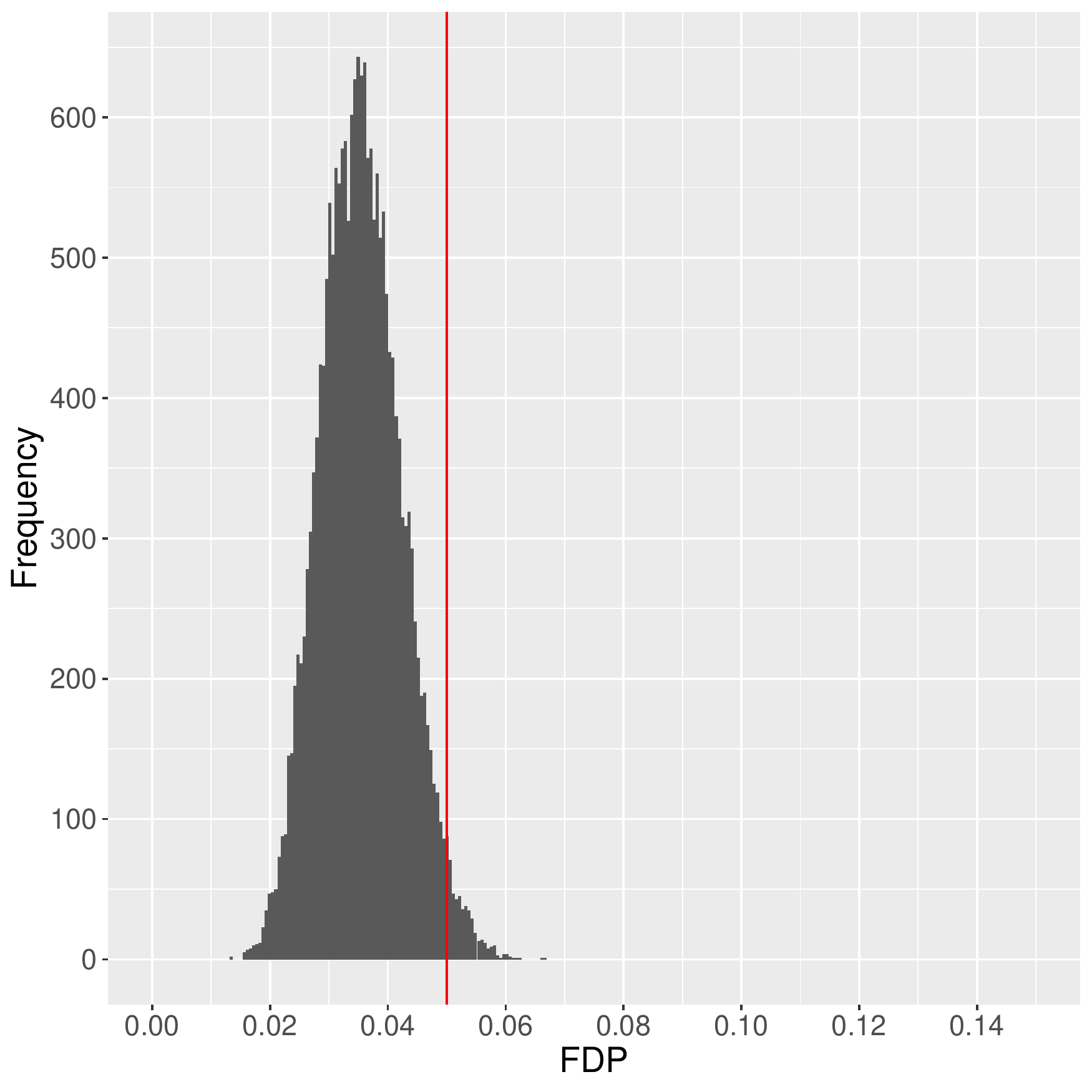}  & \includegraphics[width=2.5in]{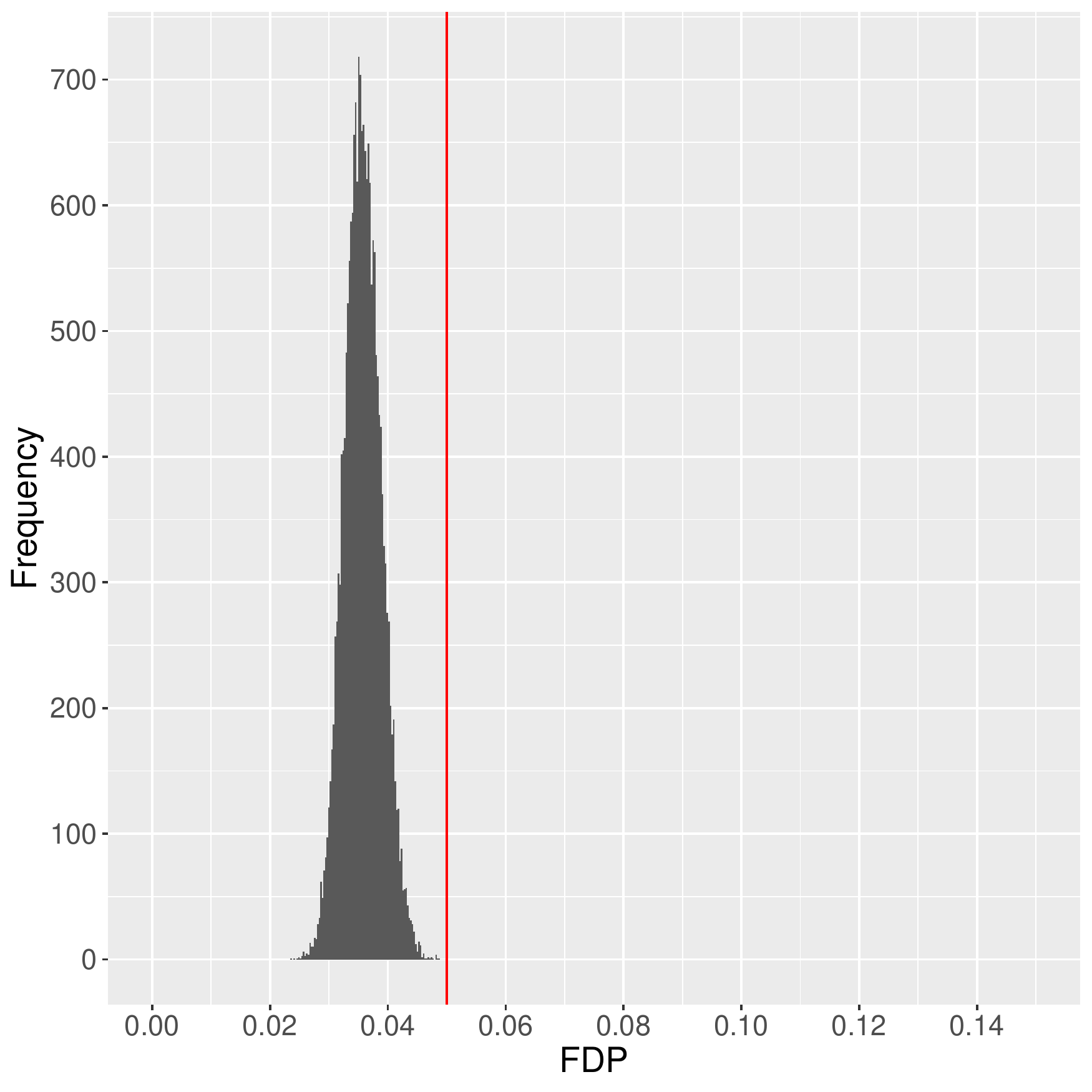} \tabularnewline
\hspace{-10ex}
\includegraphics[width=2.5in]{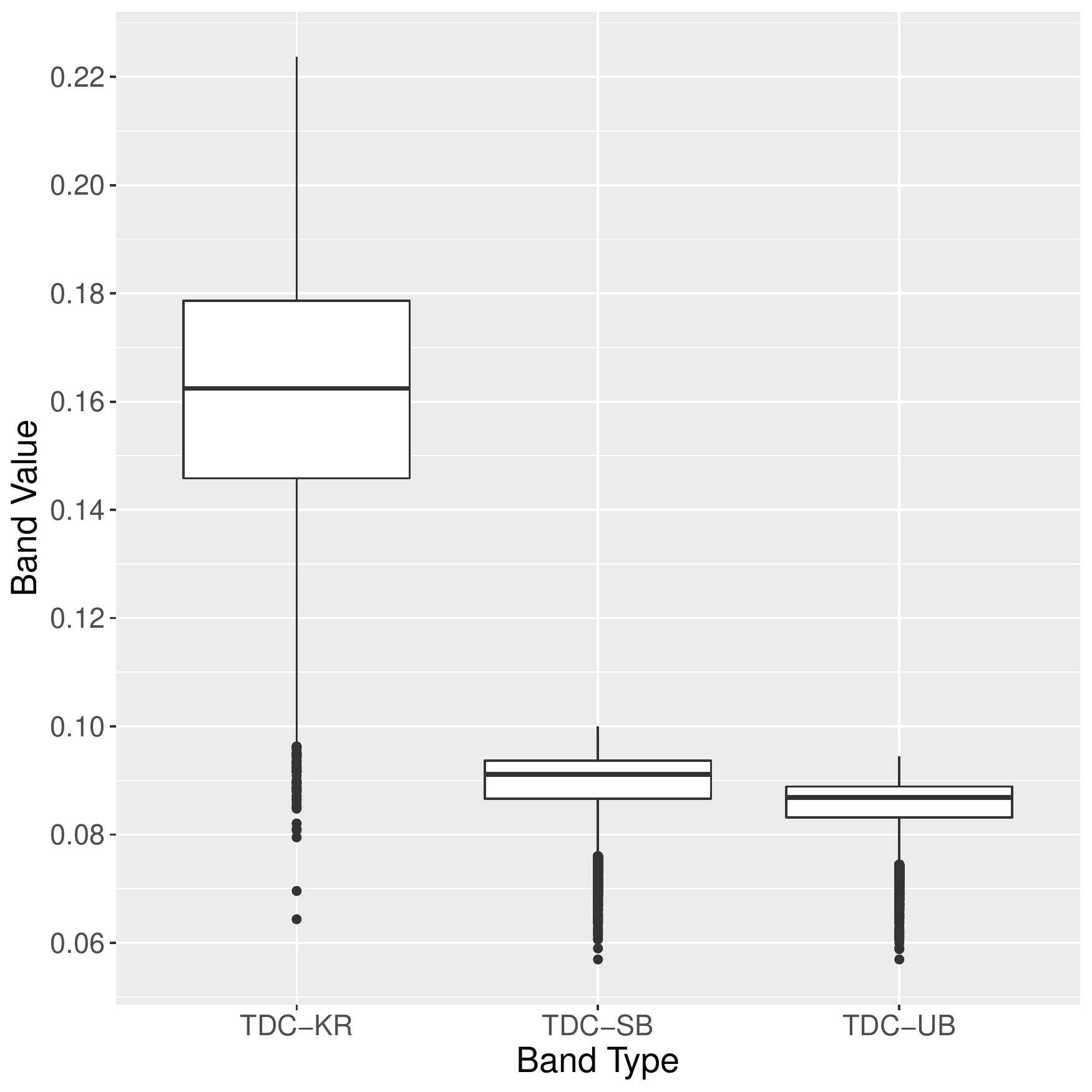}  & \includegraphics[width=2.5in]{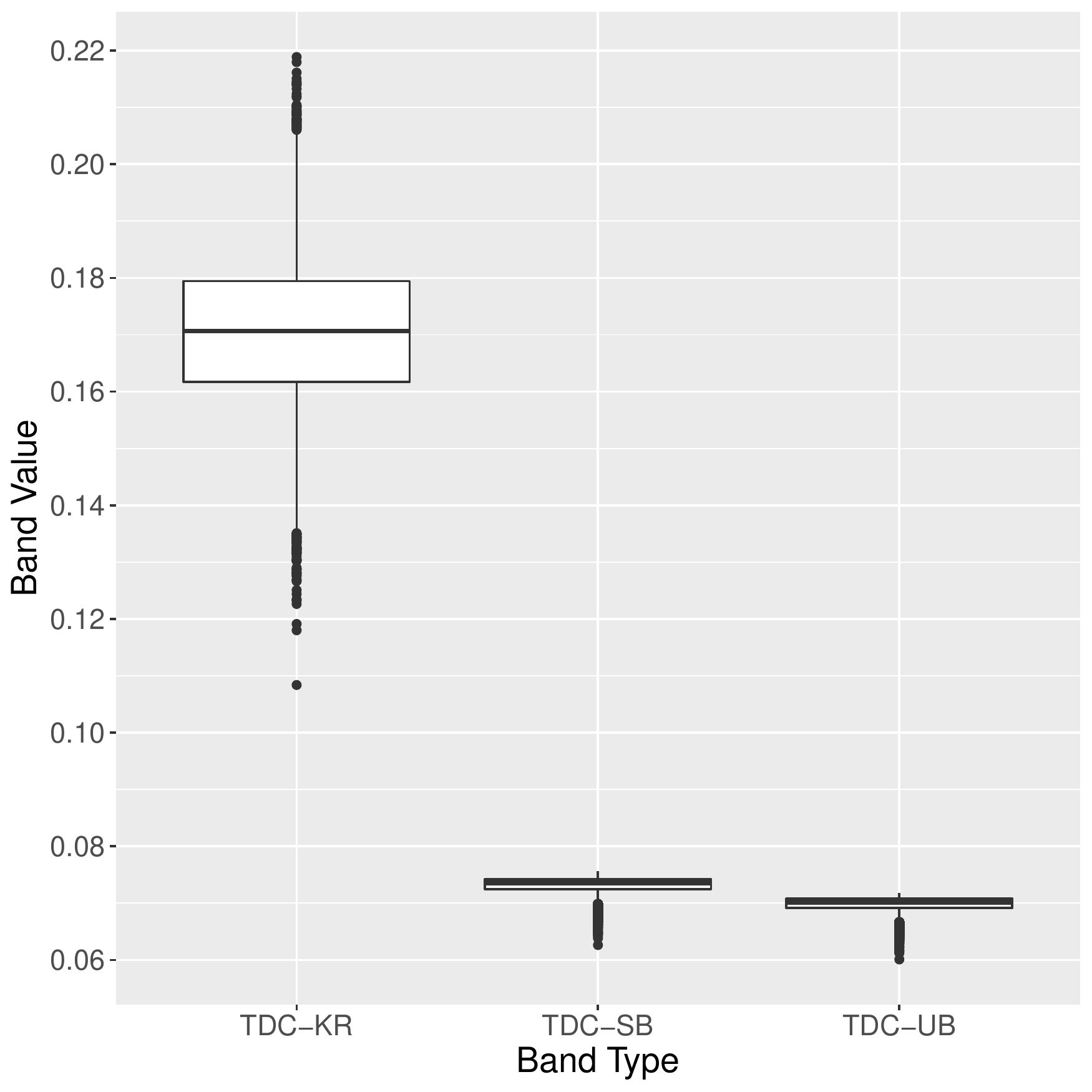}  & \includegraphics[width=2.5in]{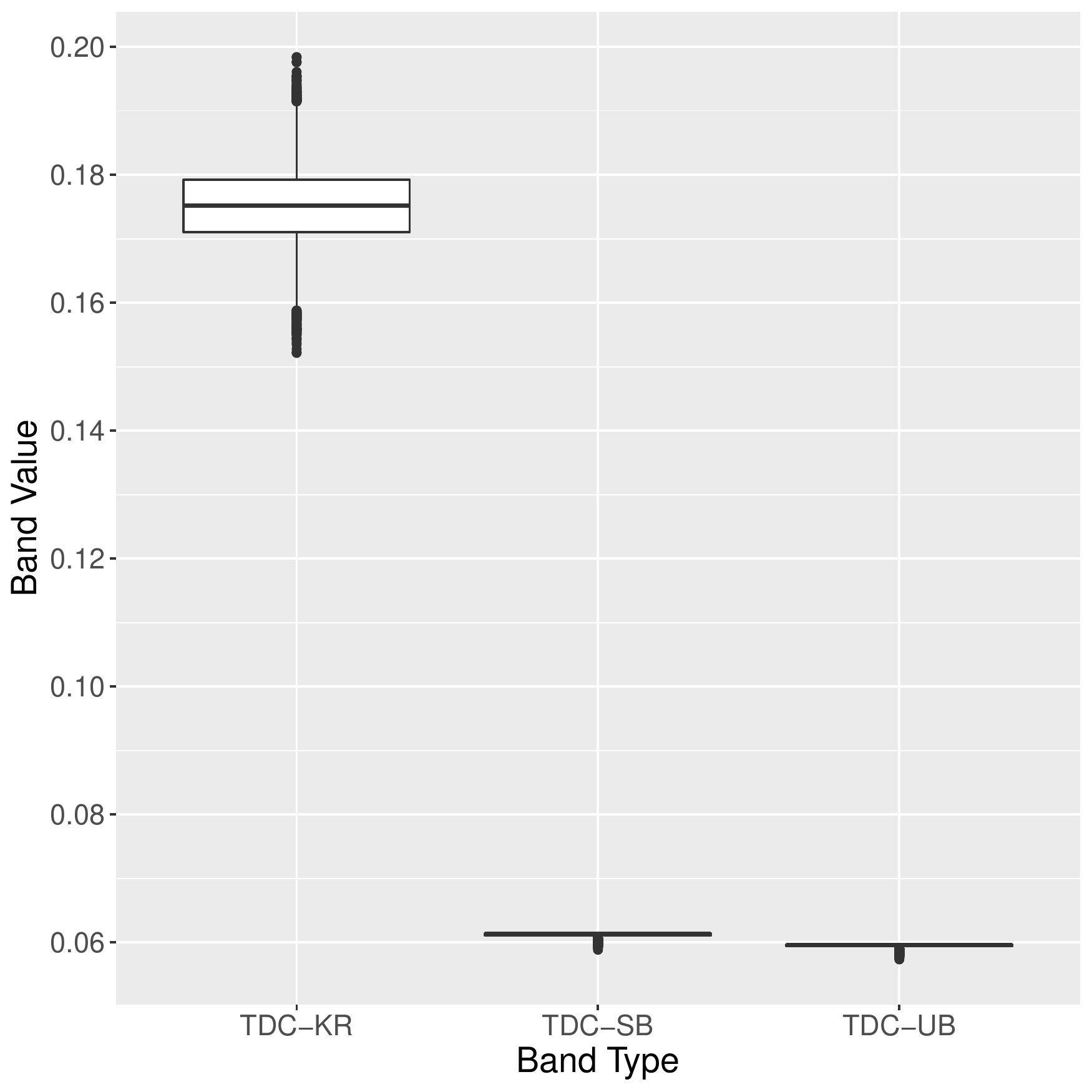} \tabularnewline
\hspace{-10ex}
\includegraphics[width=2.5in]{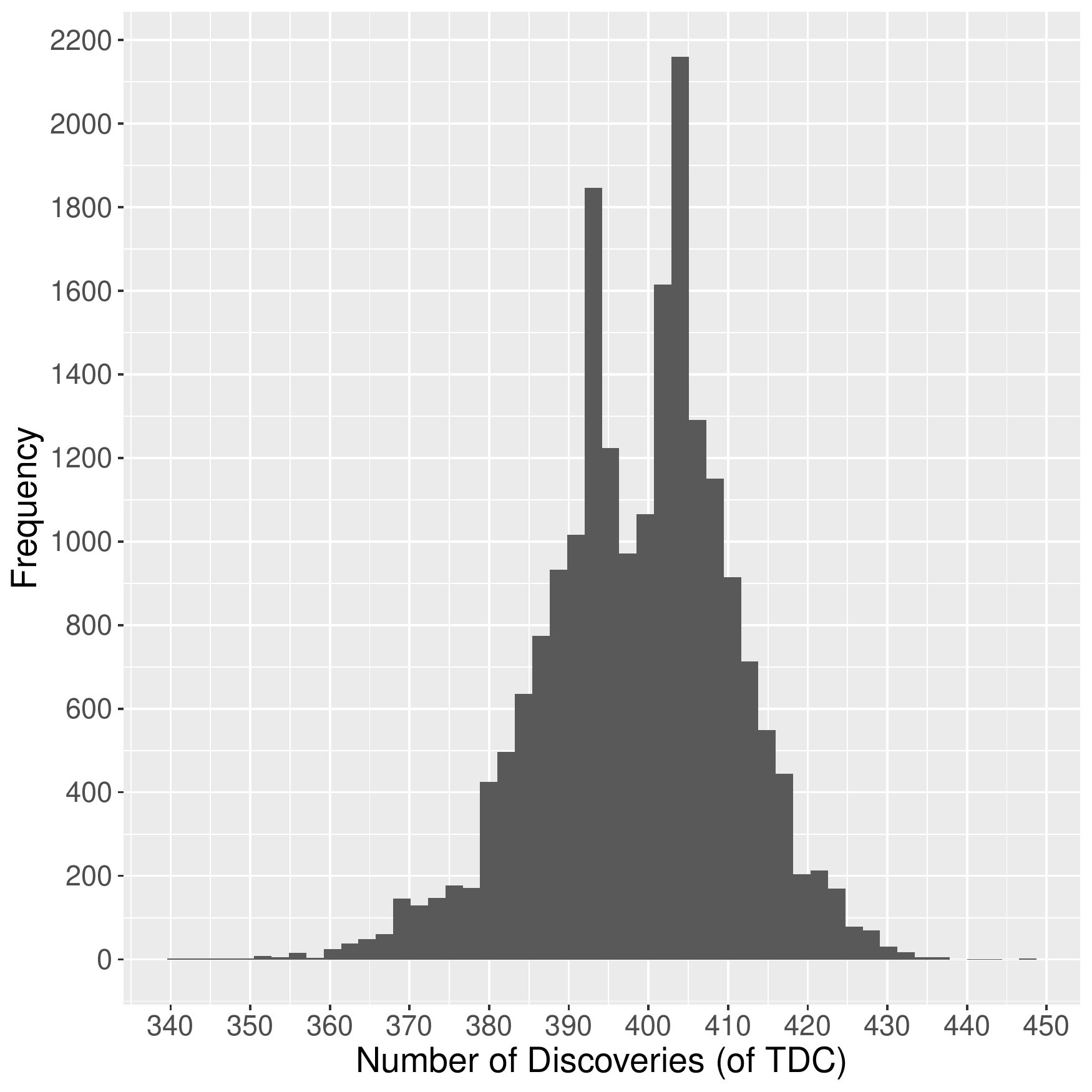}  & \includegraphics[width=2.5in]{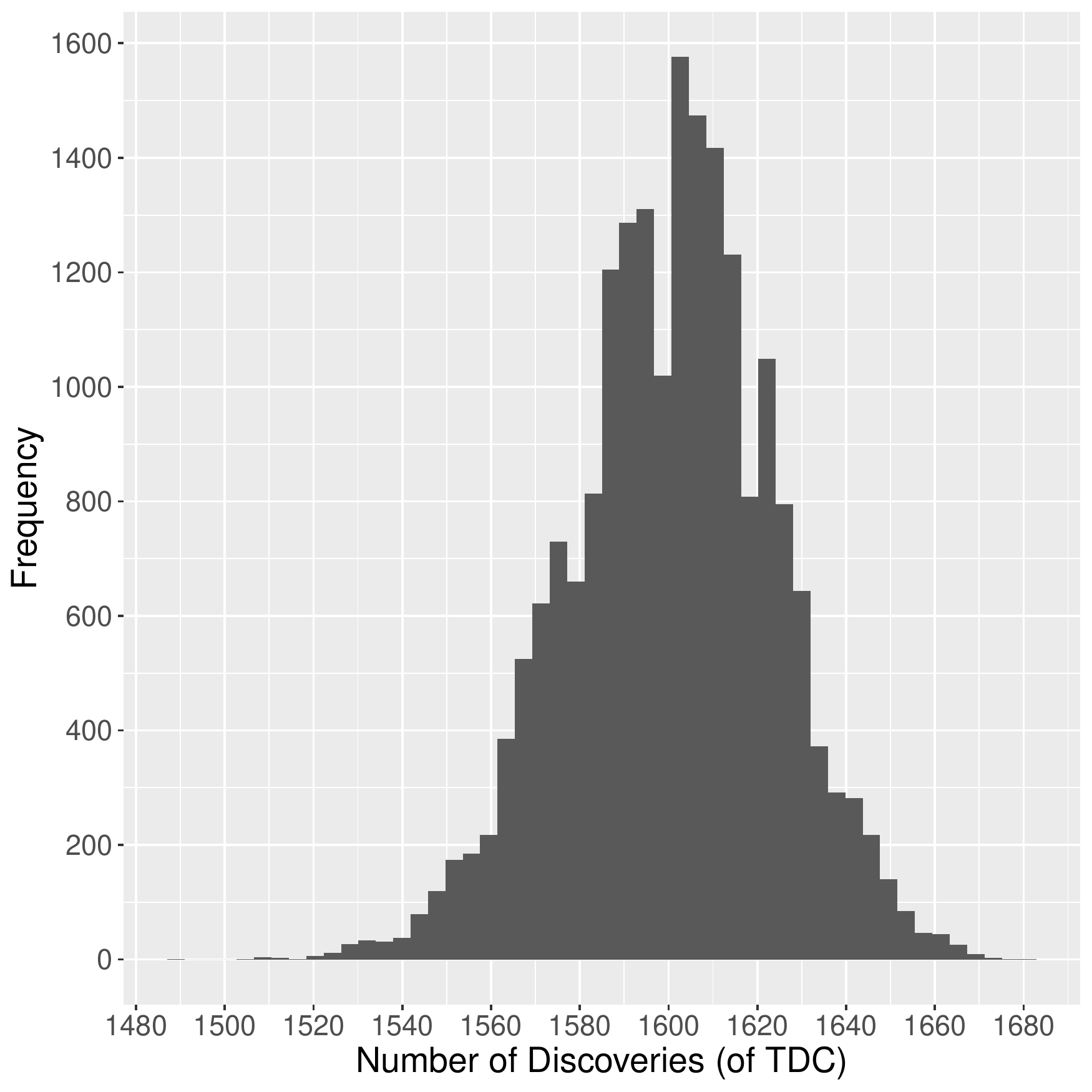}  & \includegraphics[width=2.5in]{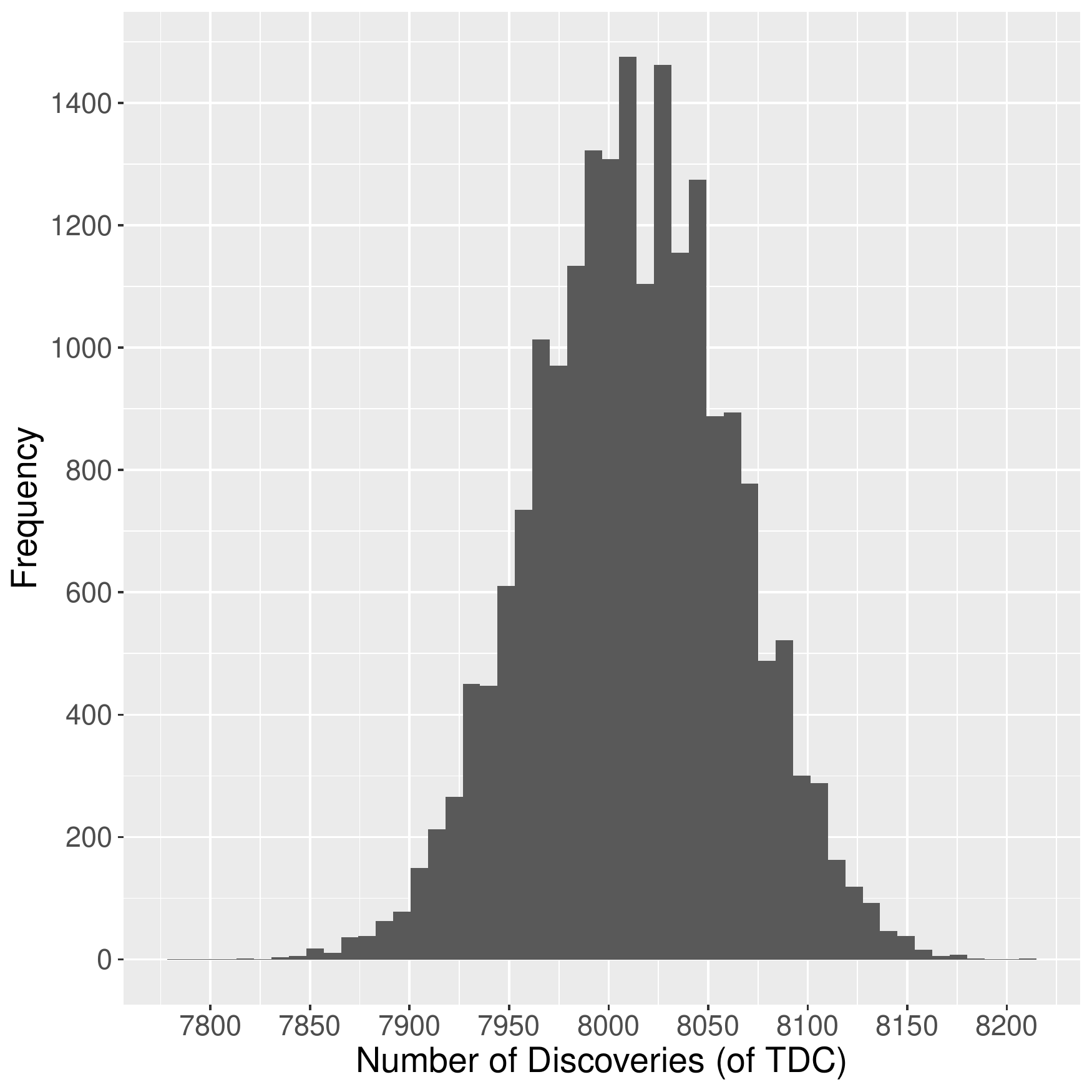} \tabularnewline
\end{tabular}
\caption{\textbf{The same as Figure 1, using $\rho = 3$, $\pi_0 = 0.2$.}
\label{supfig:vary_m_pi0_0.2}}
\end{figure}

\clearpage

\begin{figure}
\centering %
\begin{tabular}{ccc}
\hspace{-10ex}
$m = 500$  & $m = 2000$ & $m = 10000$ \tabularnewline
\hspace{-10ex}
\includegraphics[width=2.5in]{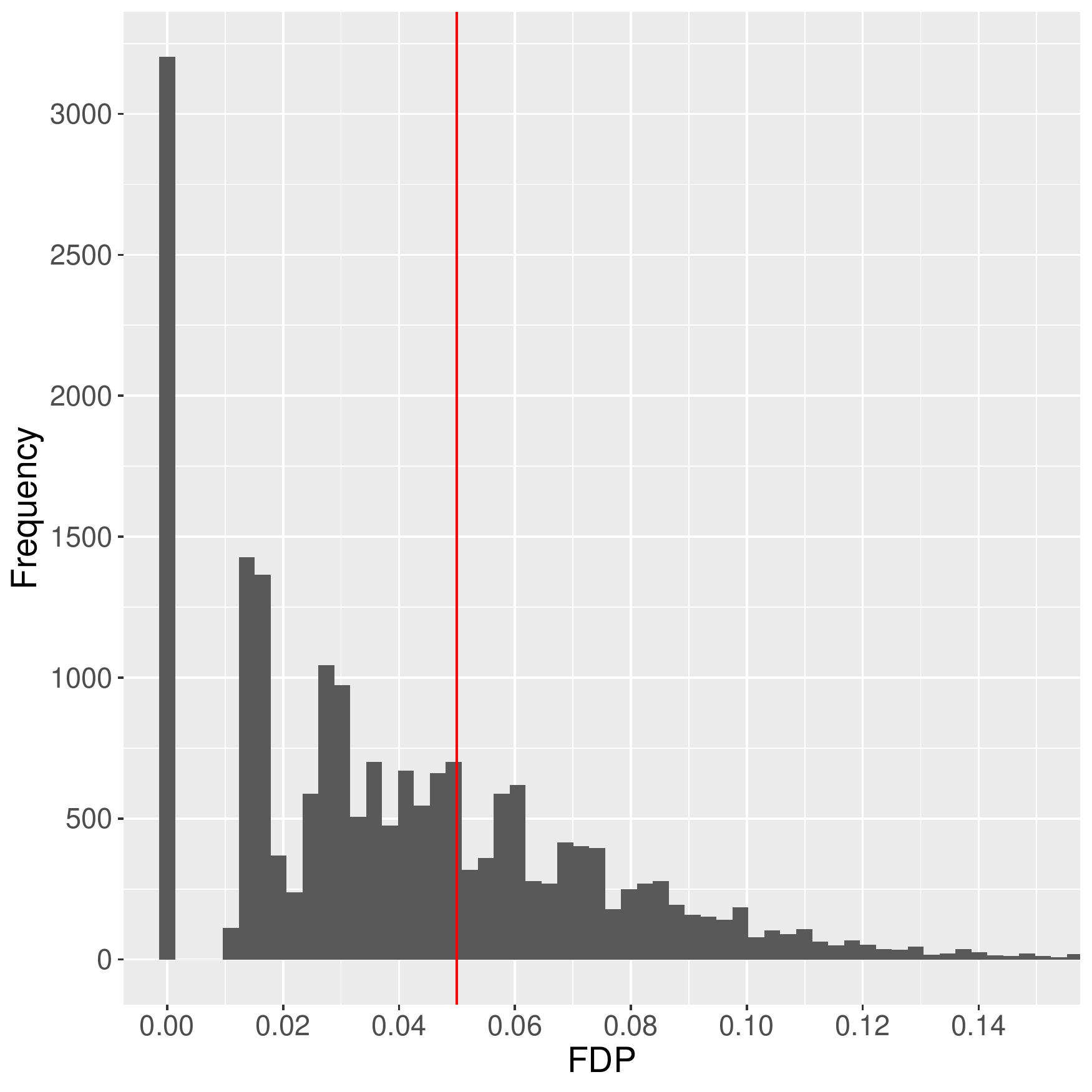}  & \includegraphics[width=2.5in]{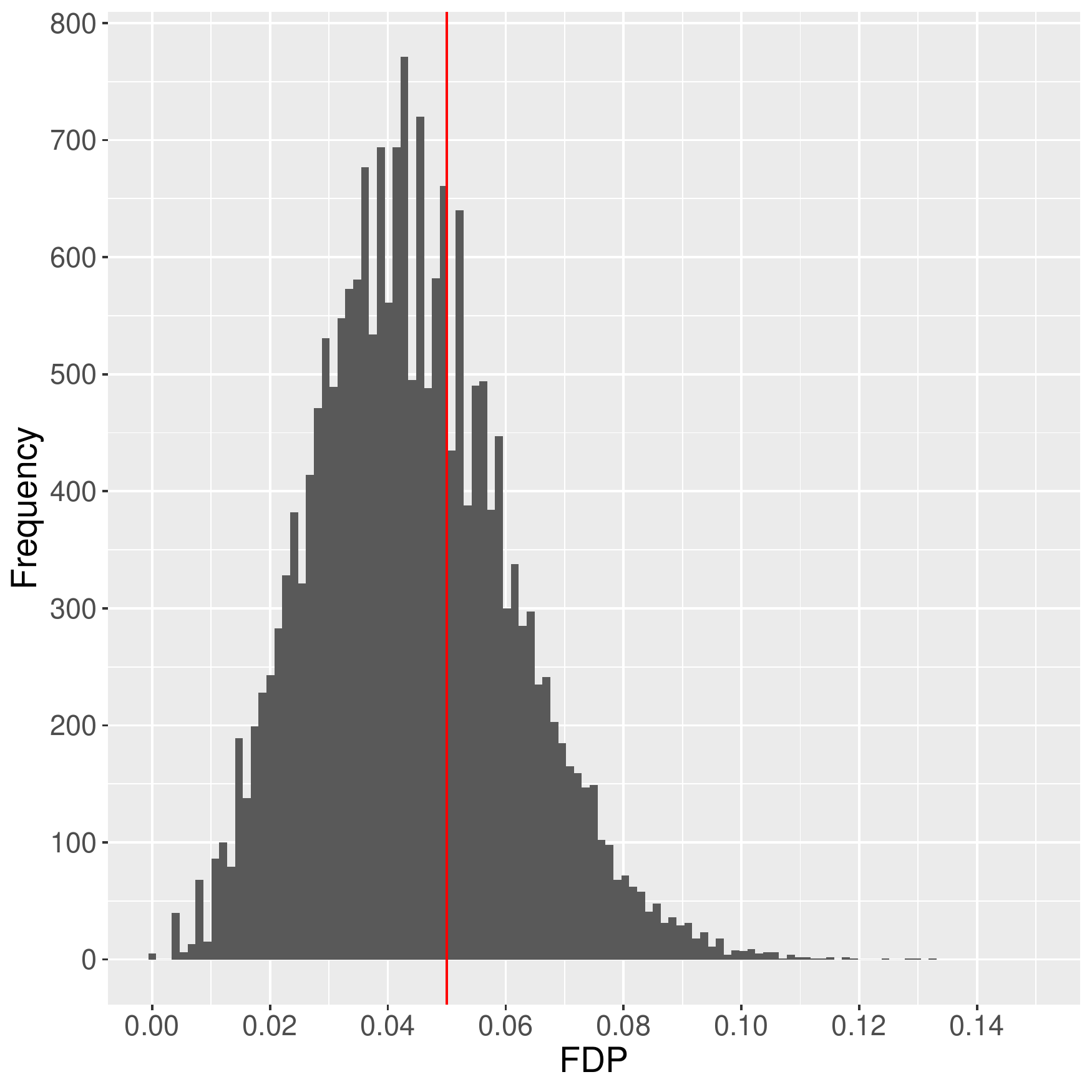}  & \includegraphics[width=2.5in]{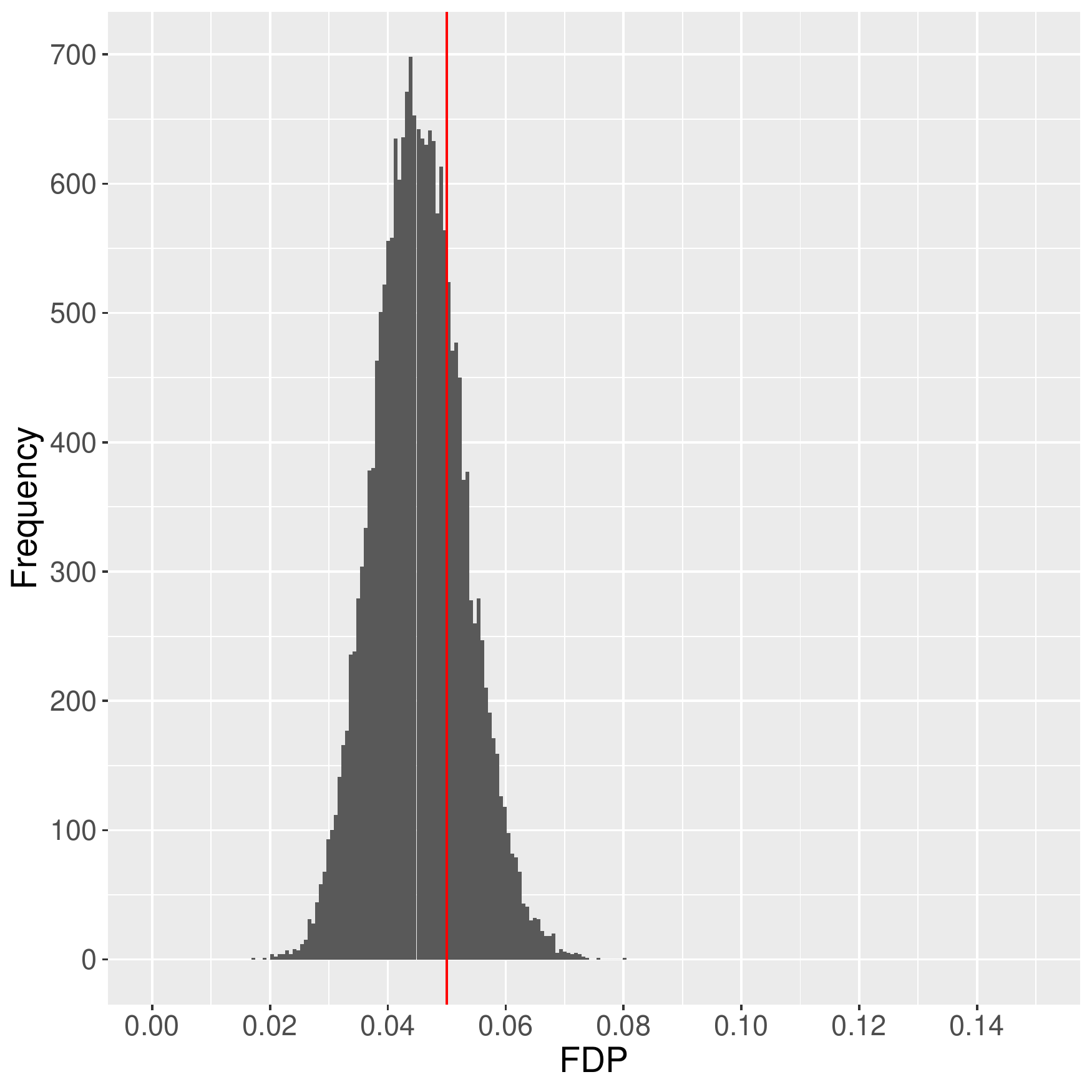} \tabularnewline
\hspace{-10ex}
\includegraphics[width=2.5in]{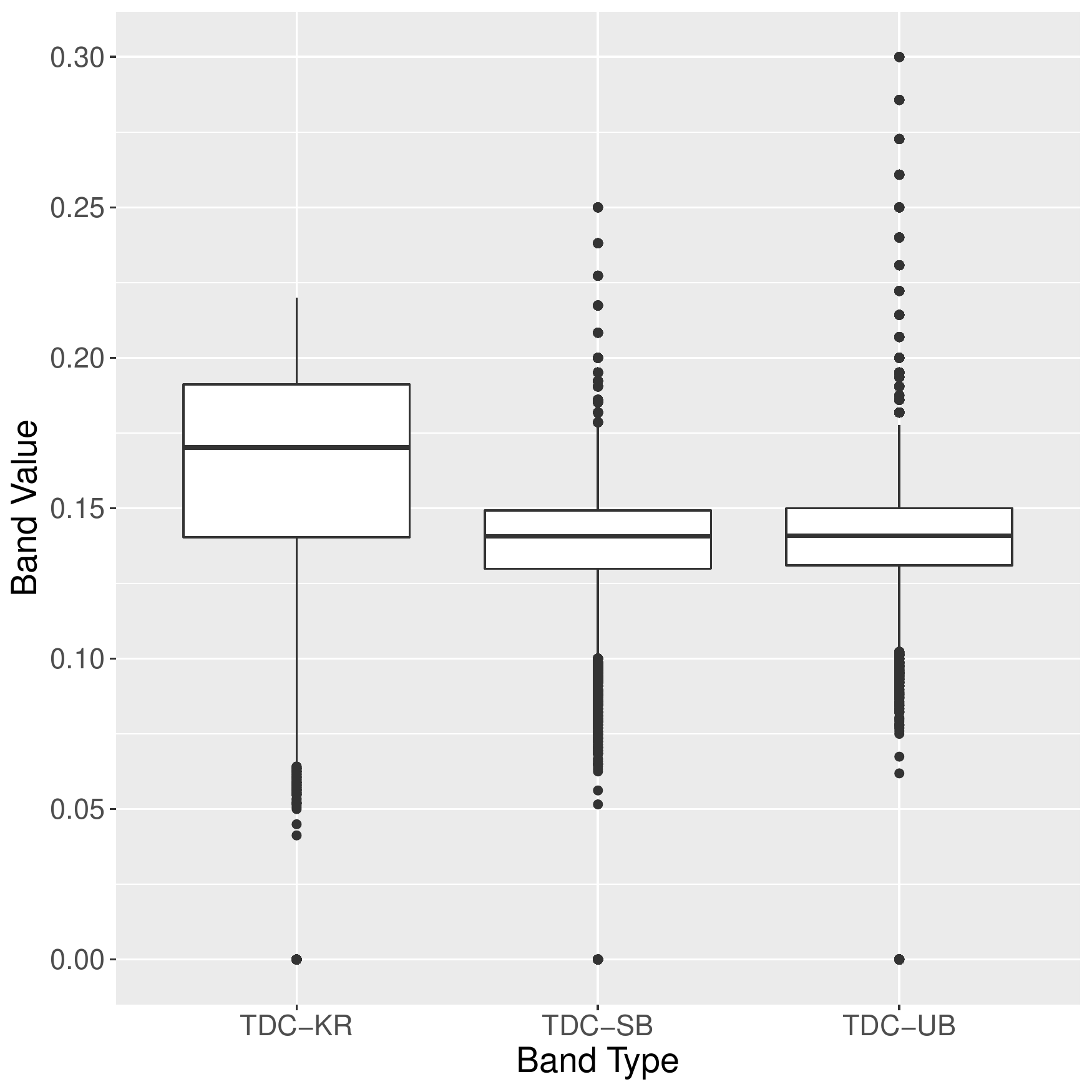}  & \includegraphics[width=2.5in]{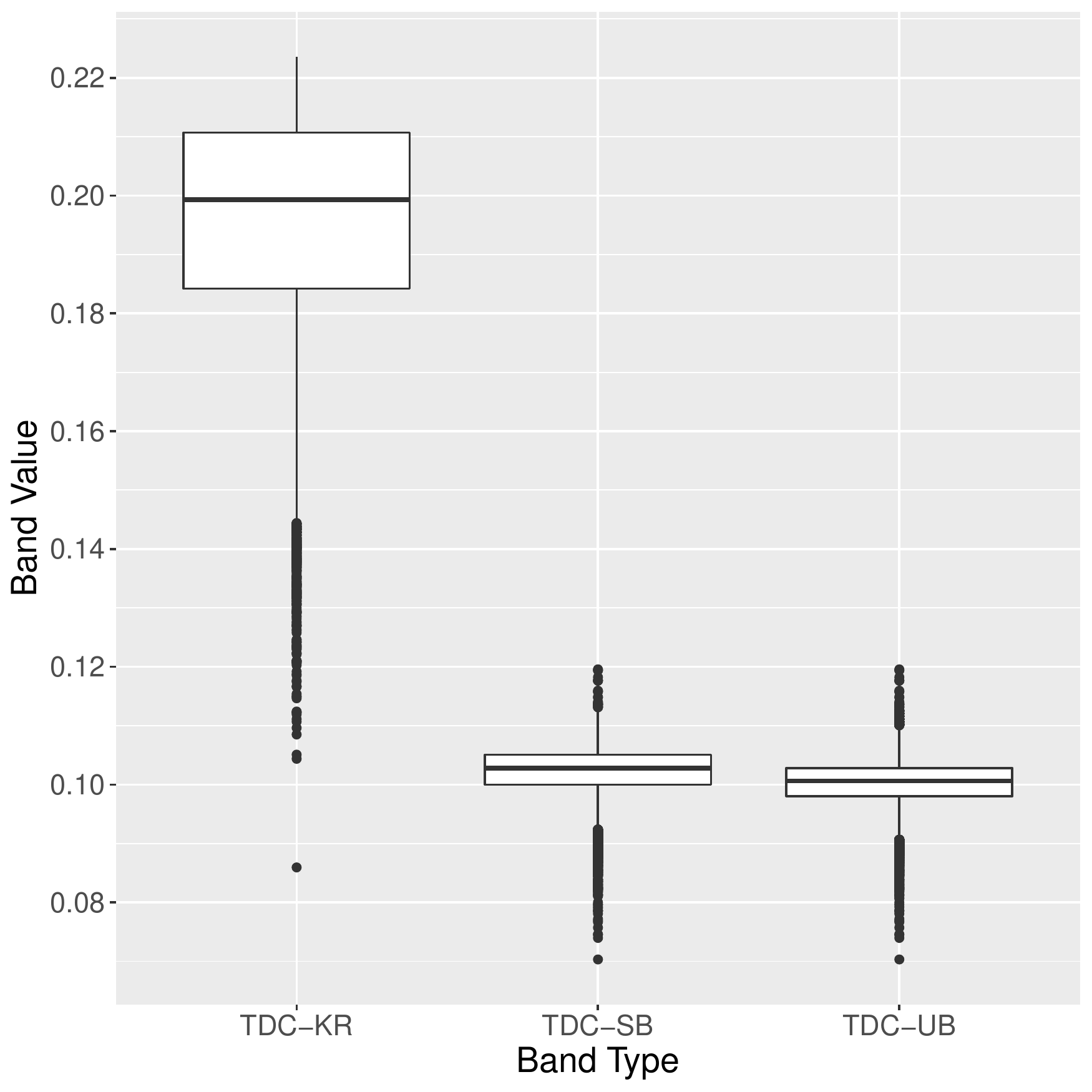}  & \includegraphics[width=2.5in]{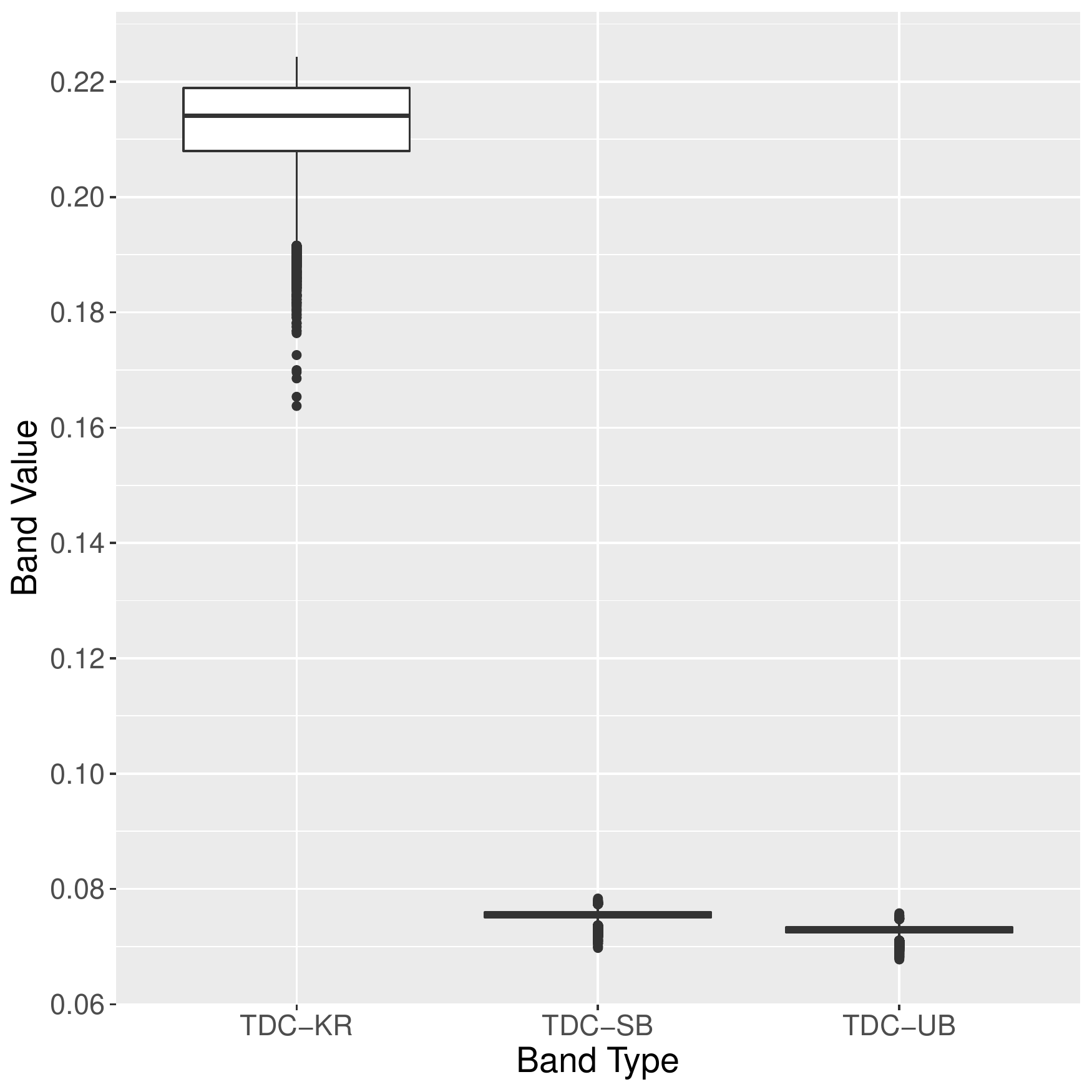} \tabularnewline
\hspace{-10ex}
\includegraphics[width=2.5in]{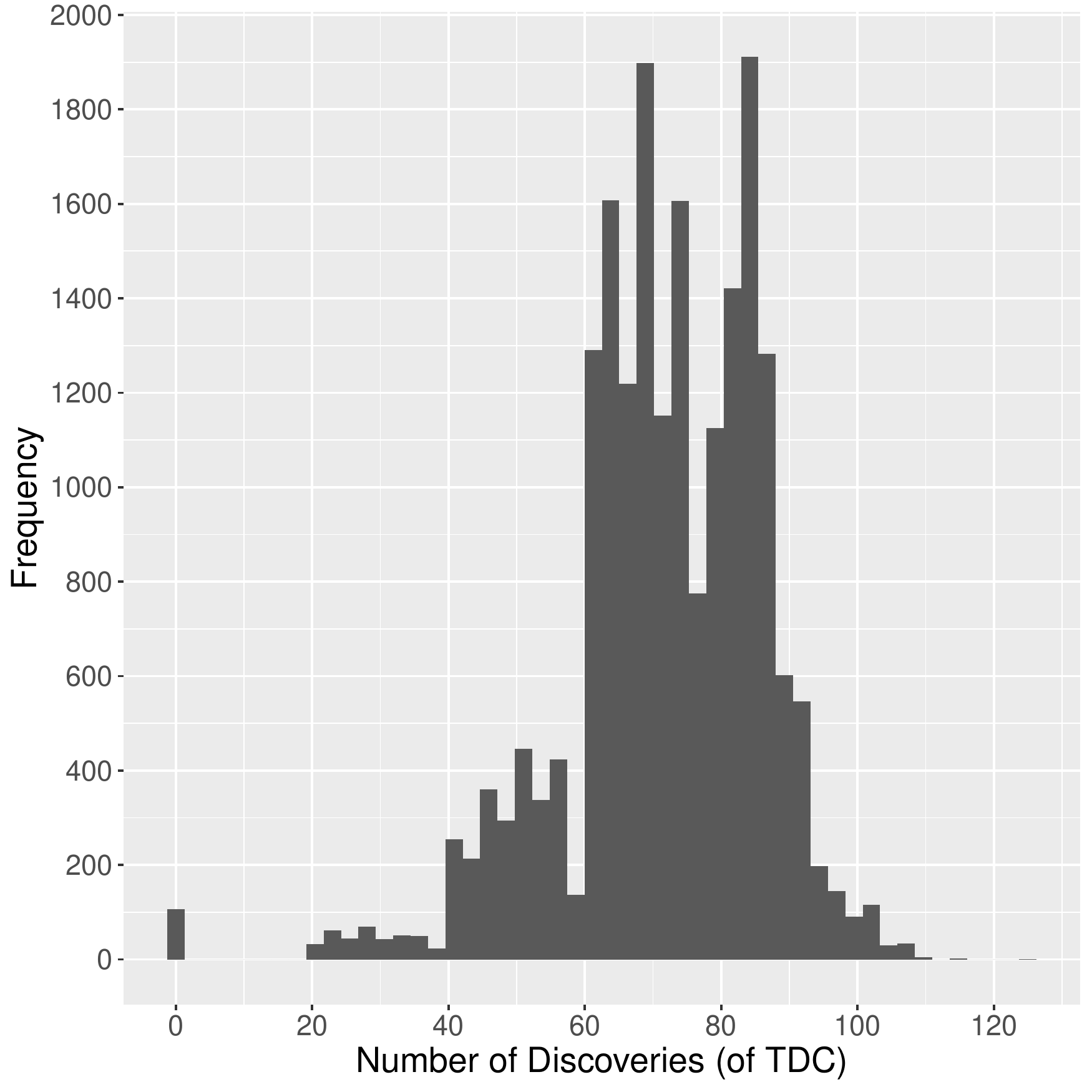}  & \includegraphics[width=2.5in]{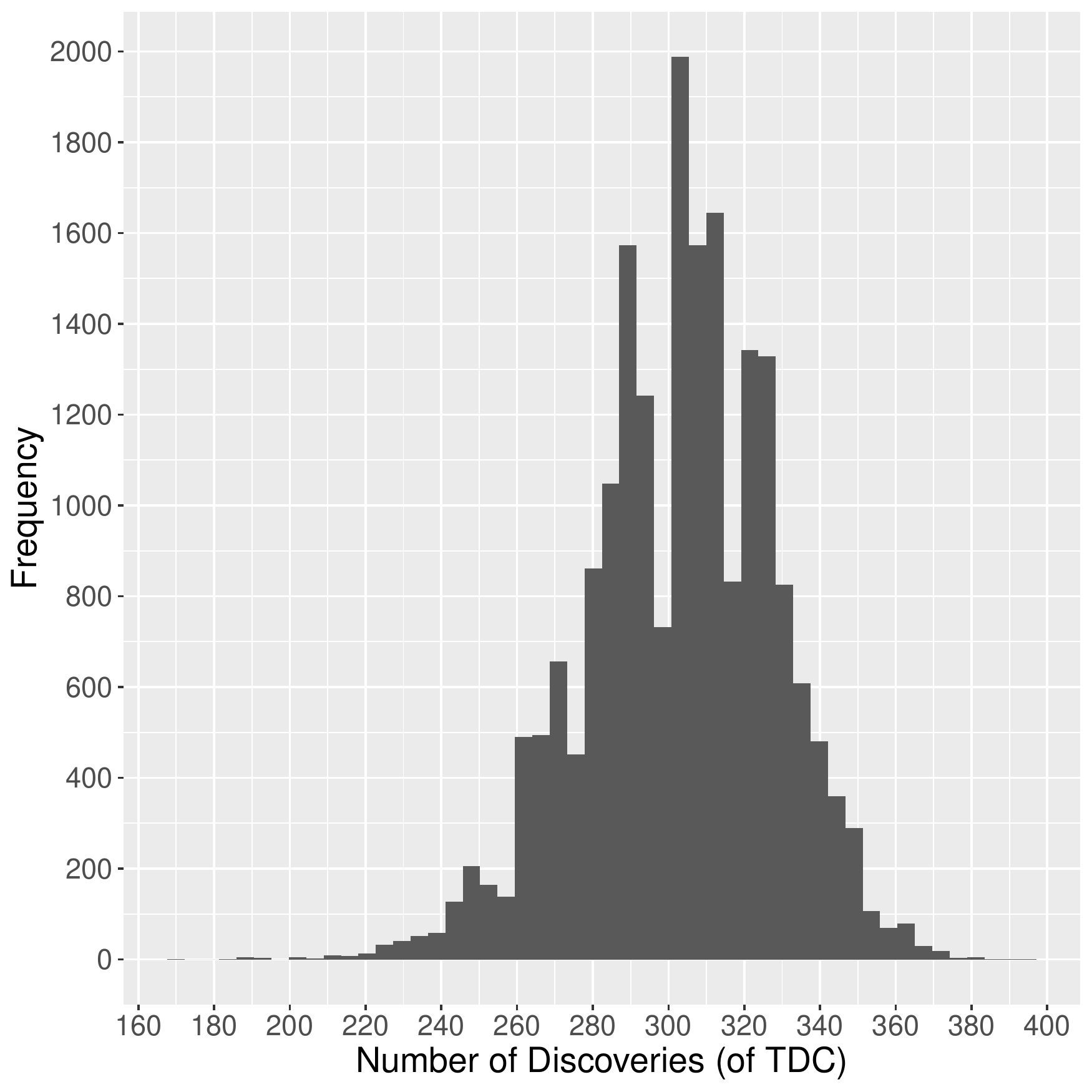}  & \includegraphics[width=2.5in]{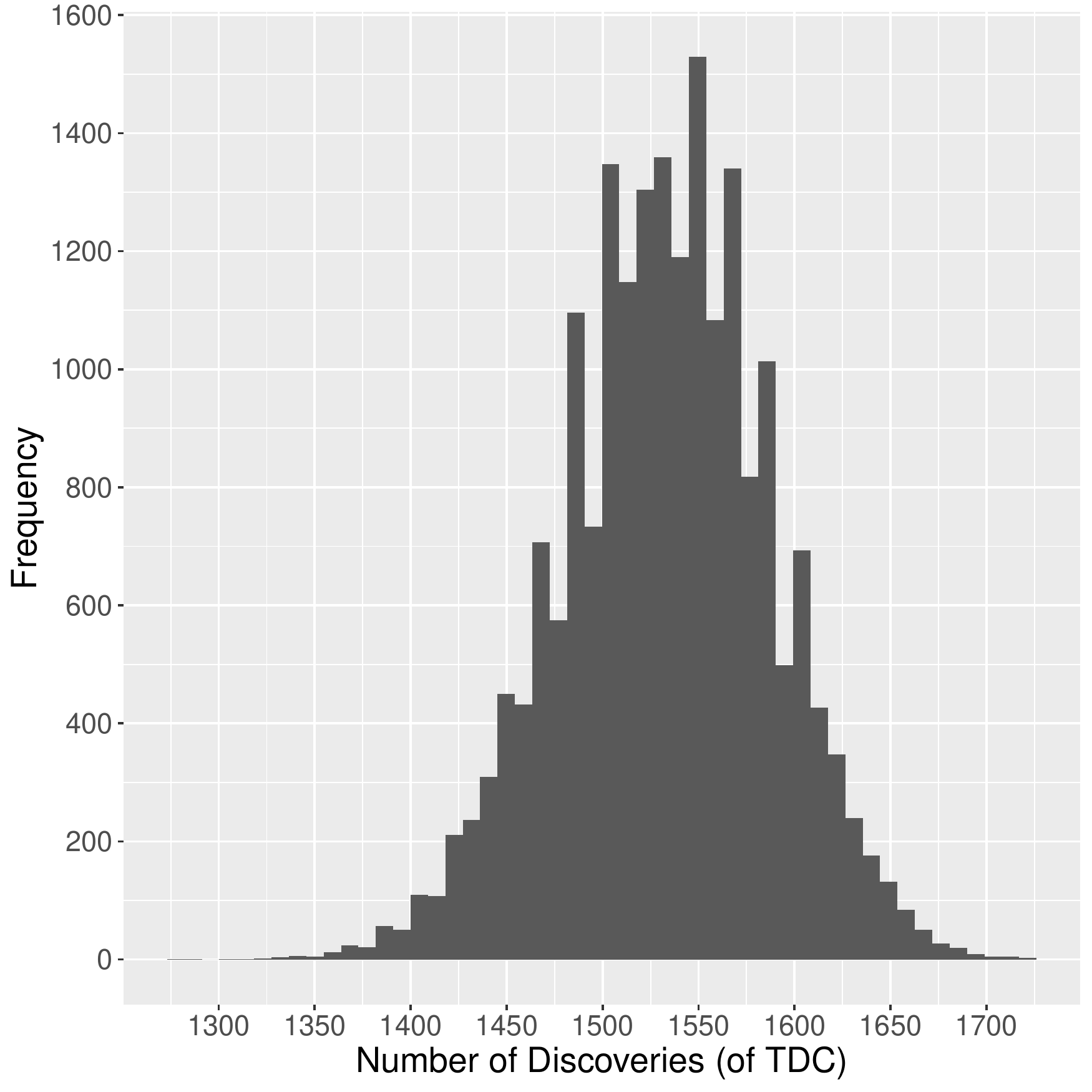} \tabularnewline
\end{tabular}
\caption{\textbf{The same as Figure 1, using $\rho = 3$, $\pi_0 = 0.8$.}
\label{supfig:vary_m_pi0_0.8}}
\end{figure}

\clearpage

\begin{figure}
\centering %
\begin{tabular}{ccc}
\hspace{-10ex}
$\pi_0=0.8$  & $\pi_0=0.5$ & $\pi_0=0.2$ \tabularnewline
\hspace{-10ex}
\includegraphics[width=2.5in]{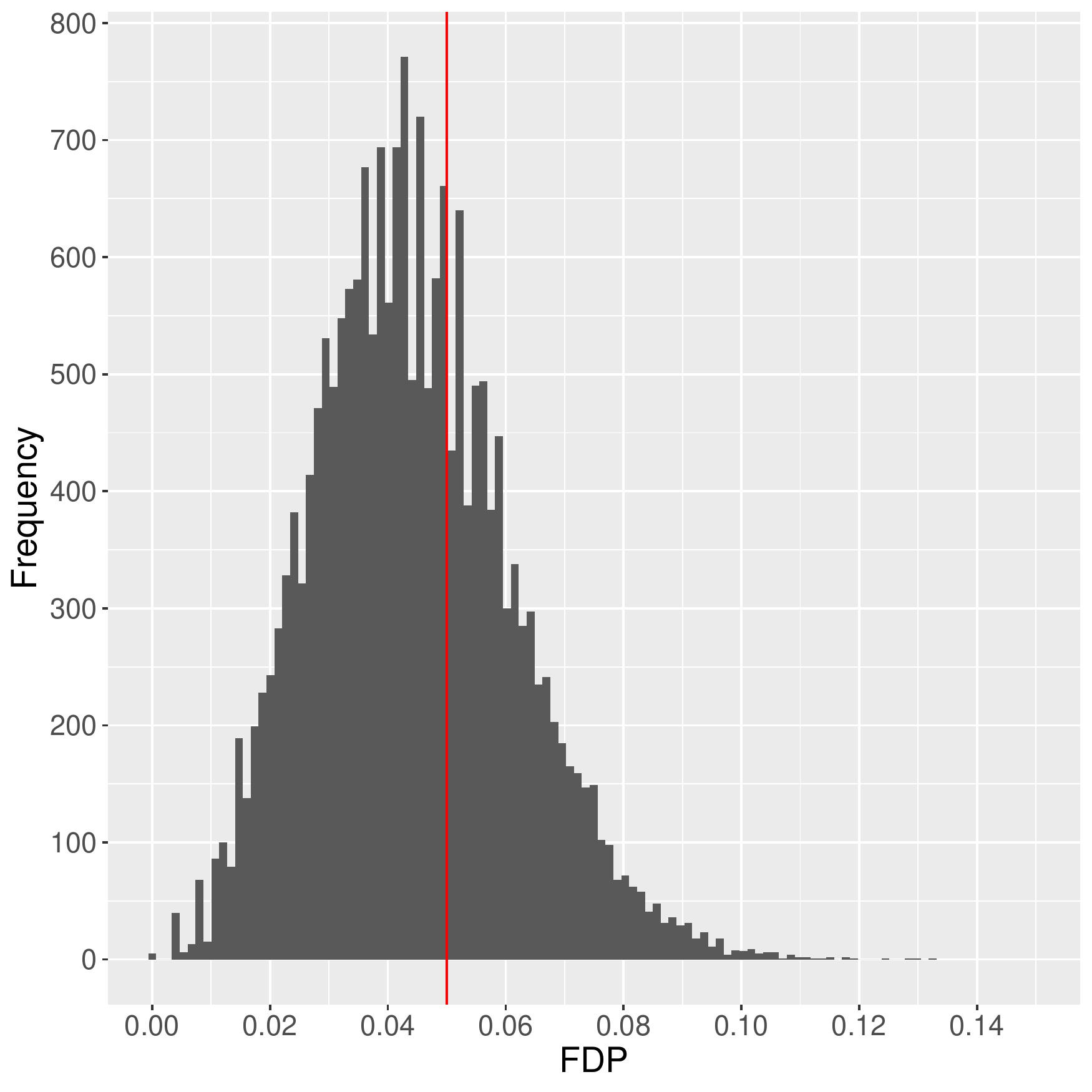}  & \includegraphics[width=2.5in]{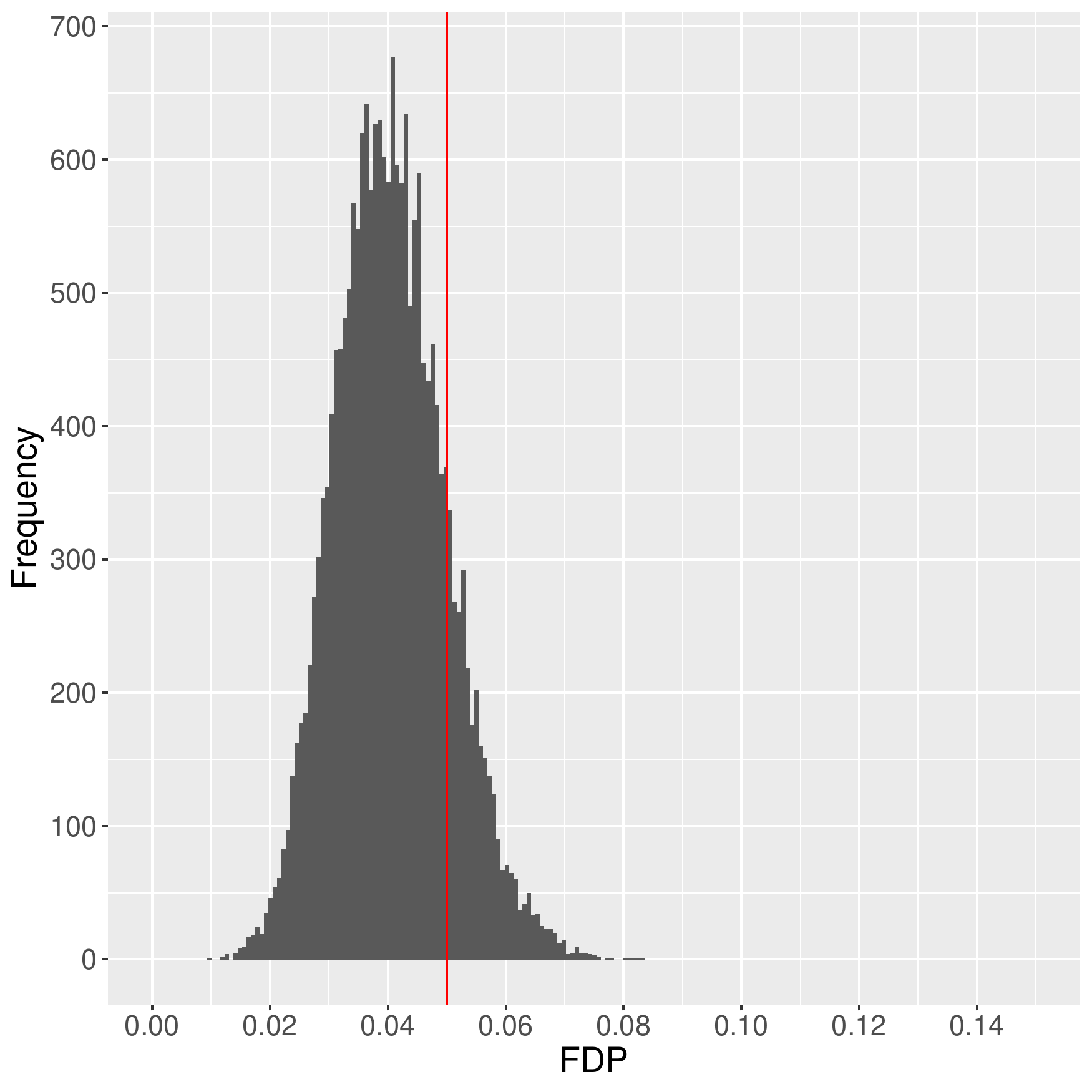}  & \includegraphics[width=2.5in]{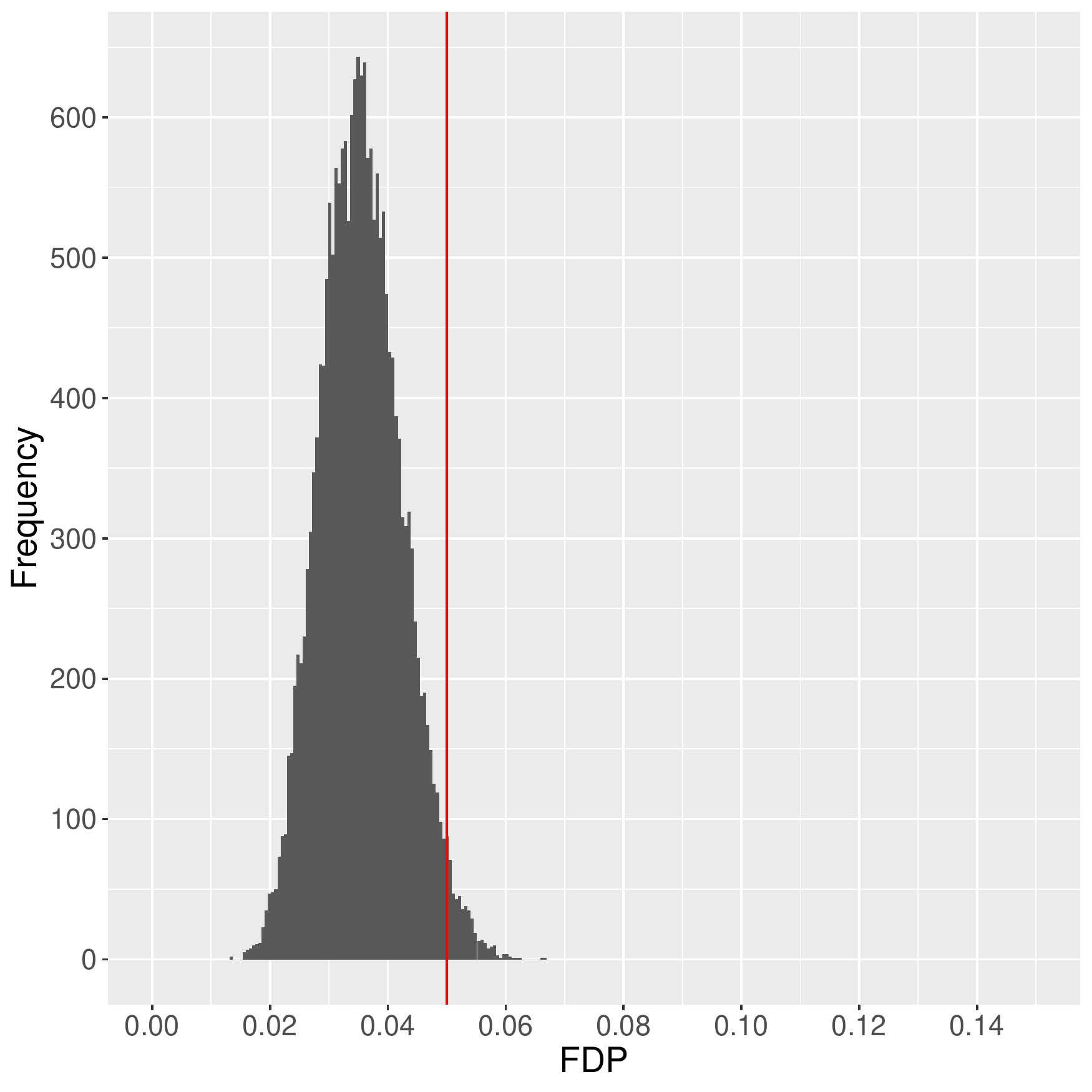} \tabularnewline
\hspace{-10ex}
\includegraphics[width=2.5in]{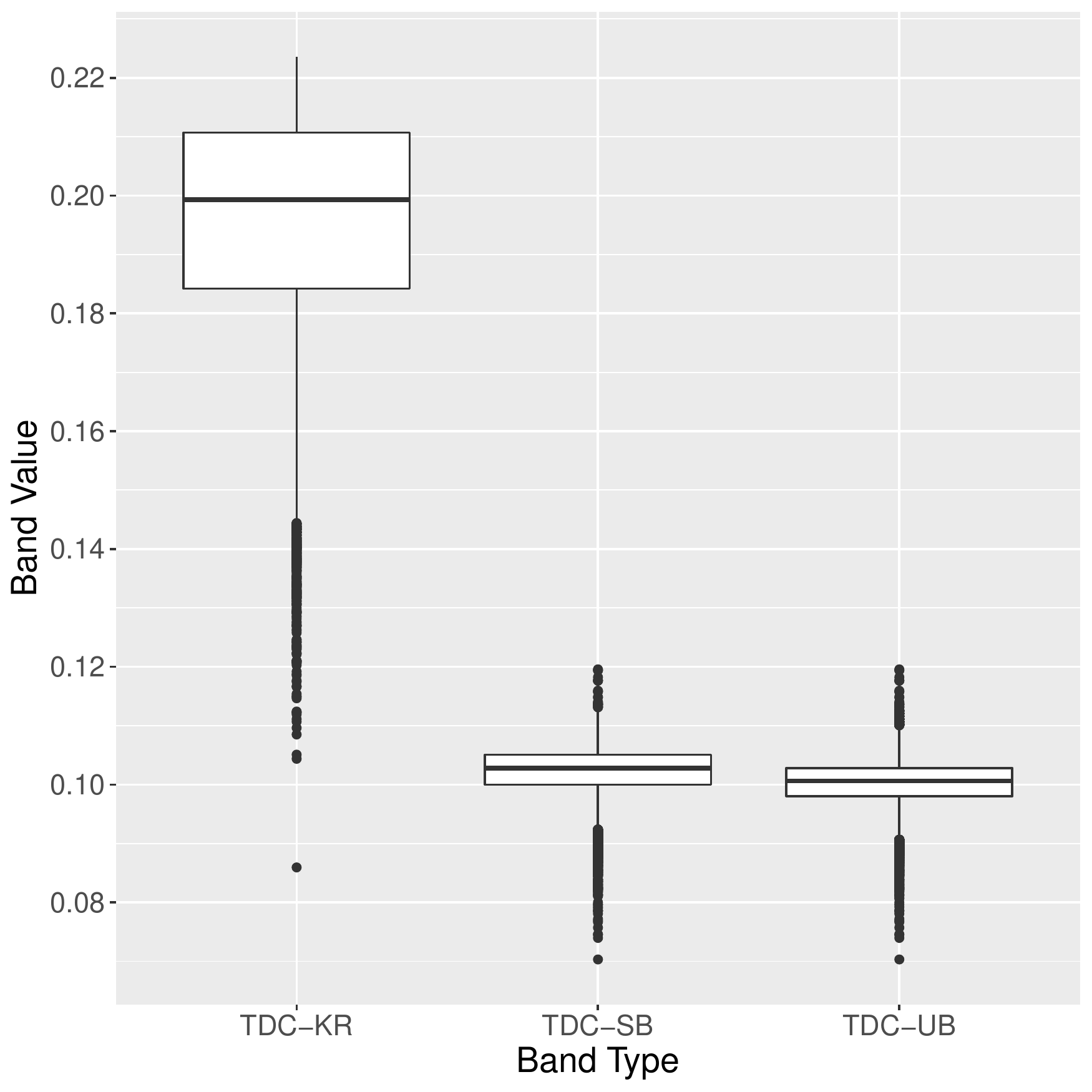}  & \includegraphics[width=2.5in]{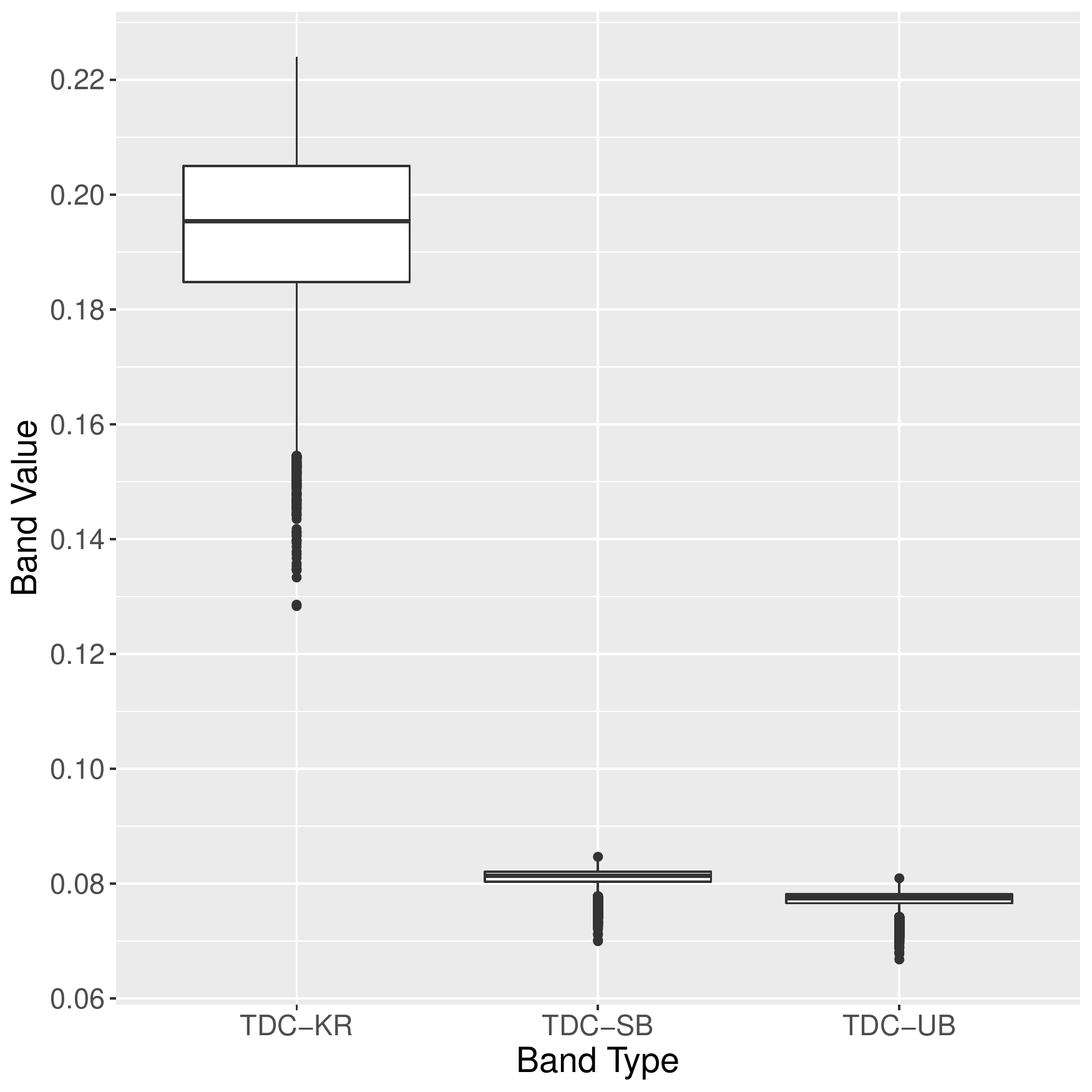}  & \includegraphics[width=2.5in]{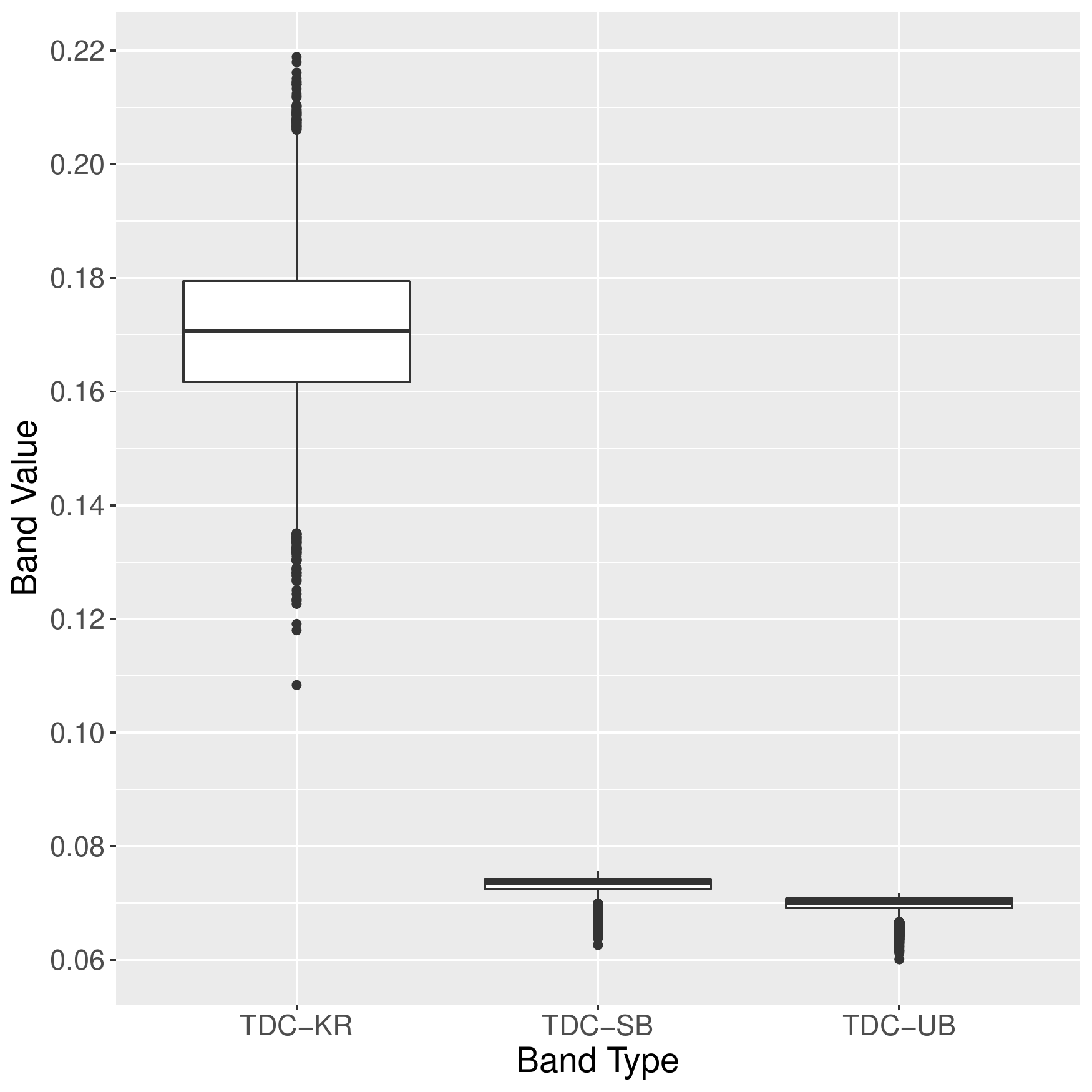} \tabularnewline
\hspace{-10ex}
\includegraphics[width=2.5in]{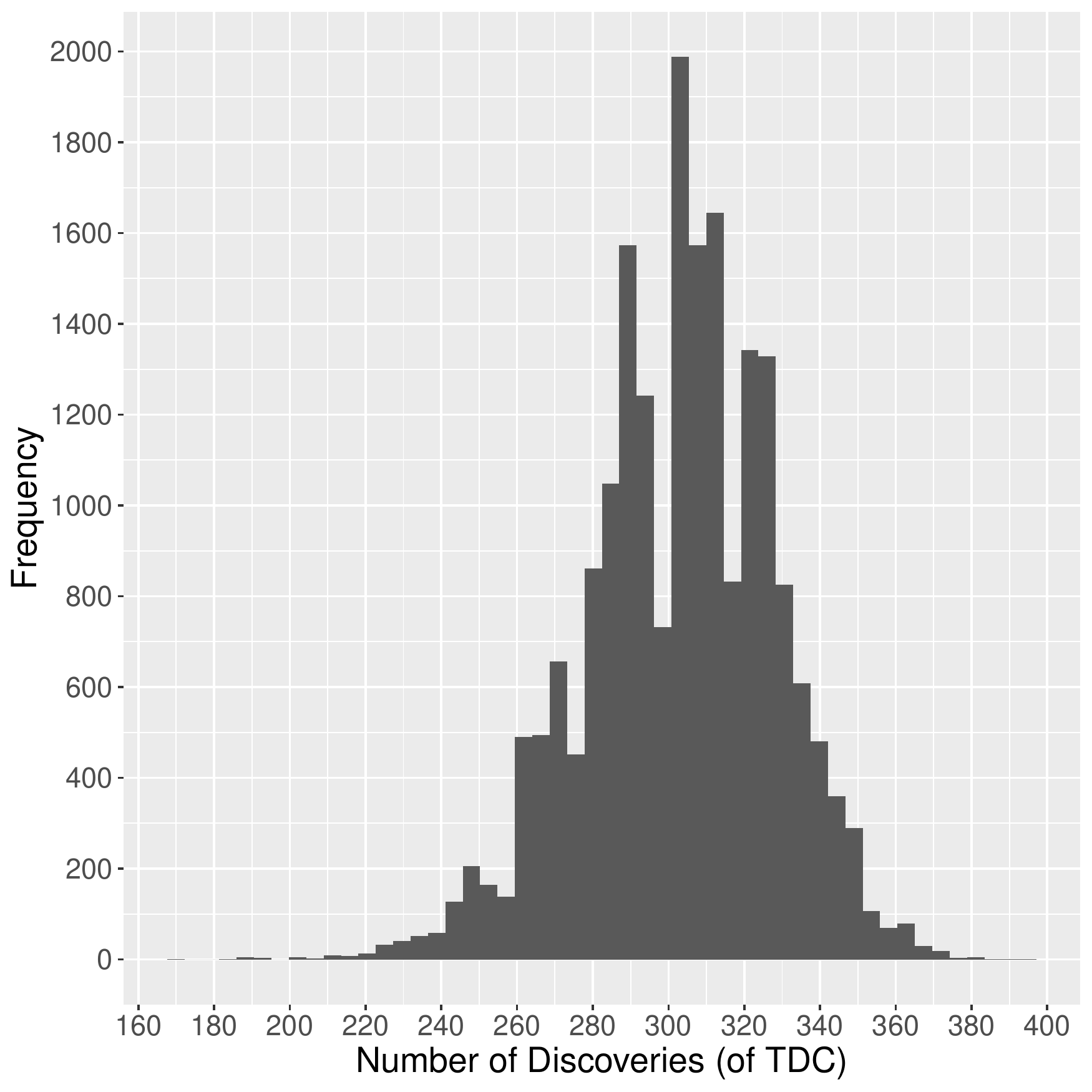}  & \includegraphics[width=2.5in]{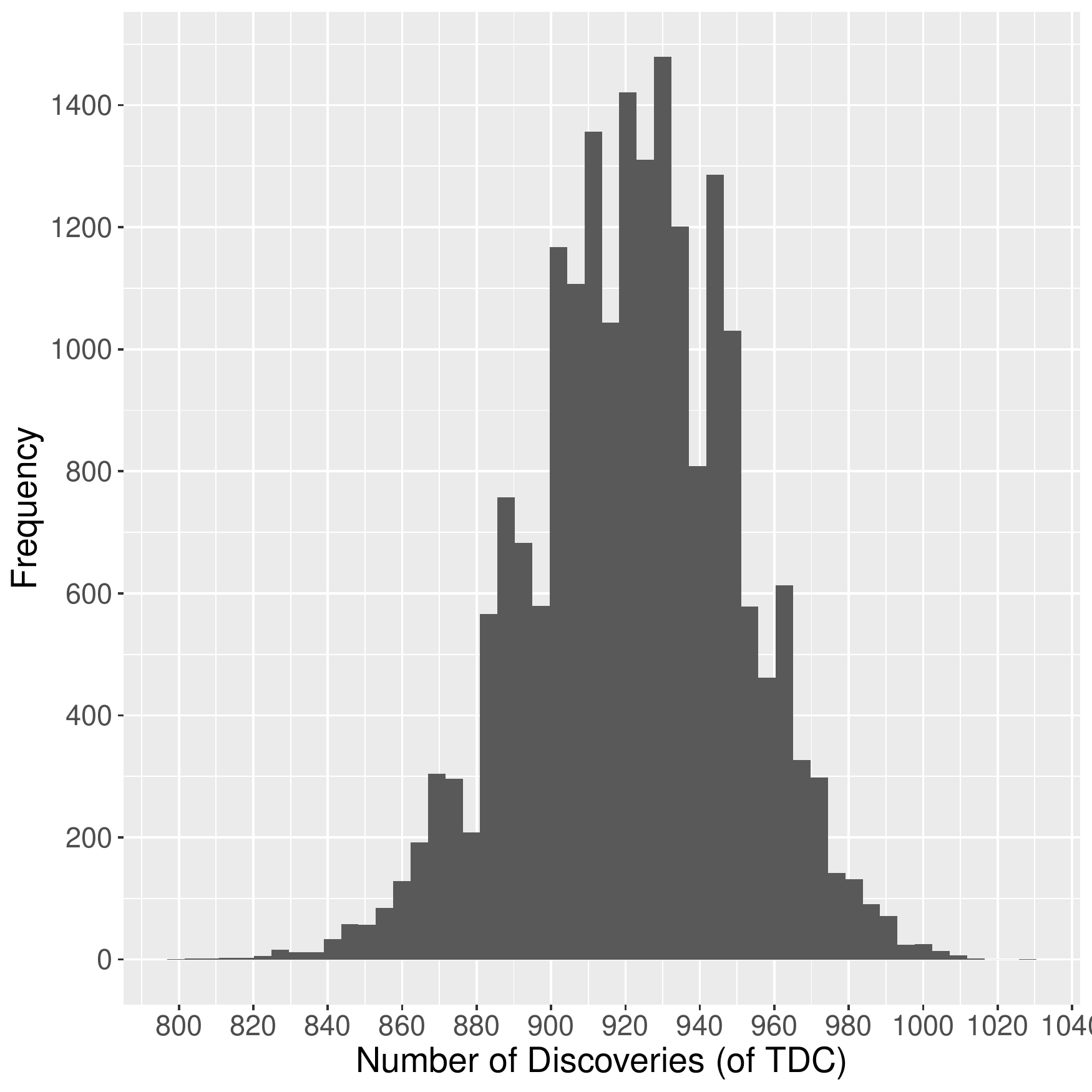}  & \includegraphics[width=2.5in]{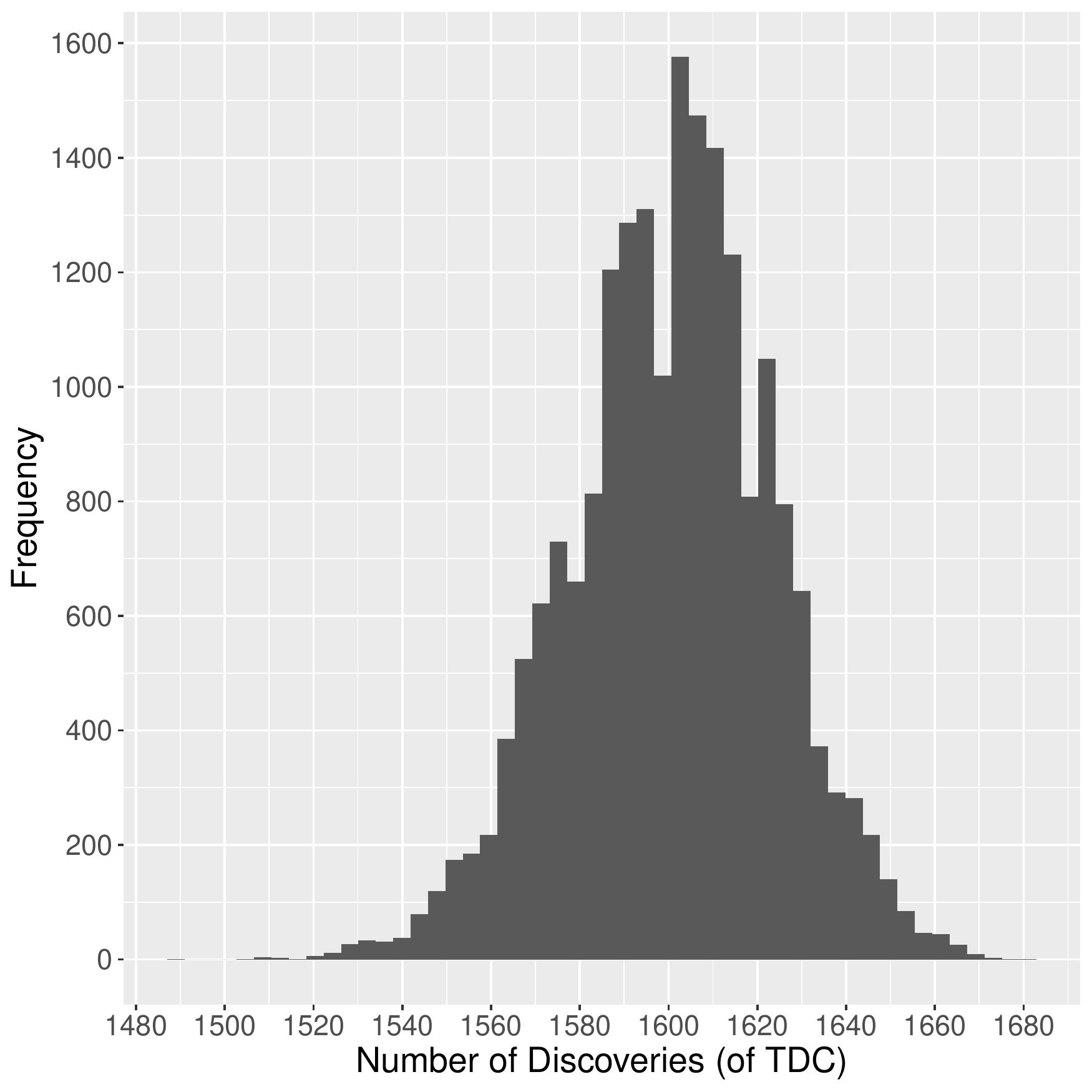} \tabularnewline
\end{tabular}
\caption{\textbf{Varying $\pi_0$ in simulated experiments.} Similar to \supfig~\ref{supfig:vary_m} only here we decrease $\pi_0$,
keeping $\alp=0.05$, $\gam=0.05$, and $m=$ 2K.  TDC's FDP (top row), the comparison of the three interpolated upper prediction bands (middle row) and the number of discoveries of TDC (bottom row).
\label{supfig:vary_pi0}}
\end{figure}

\clearpage

\begin{figure}
\centering %
\begin{tabular}{ccc}
\hspace{-10ex}
$\rho = 2.5$  & $\rho = 3$ & $\rho = 3.5$ \tabularnewline
\hspace{-10ex}
\includegraphics[width=2.5in]{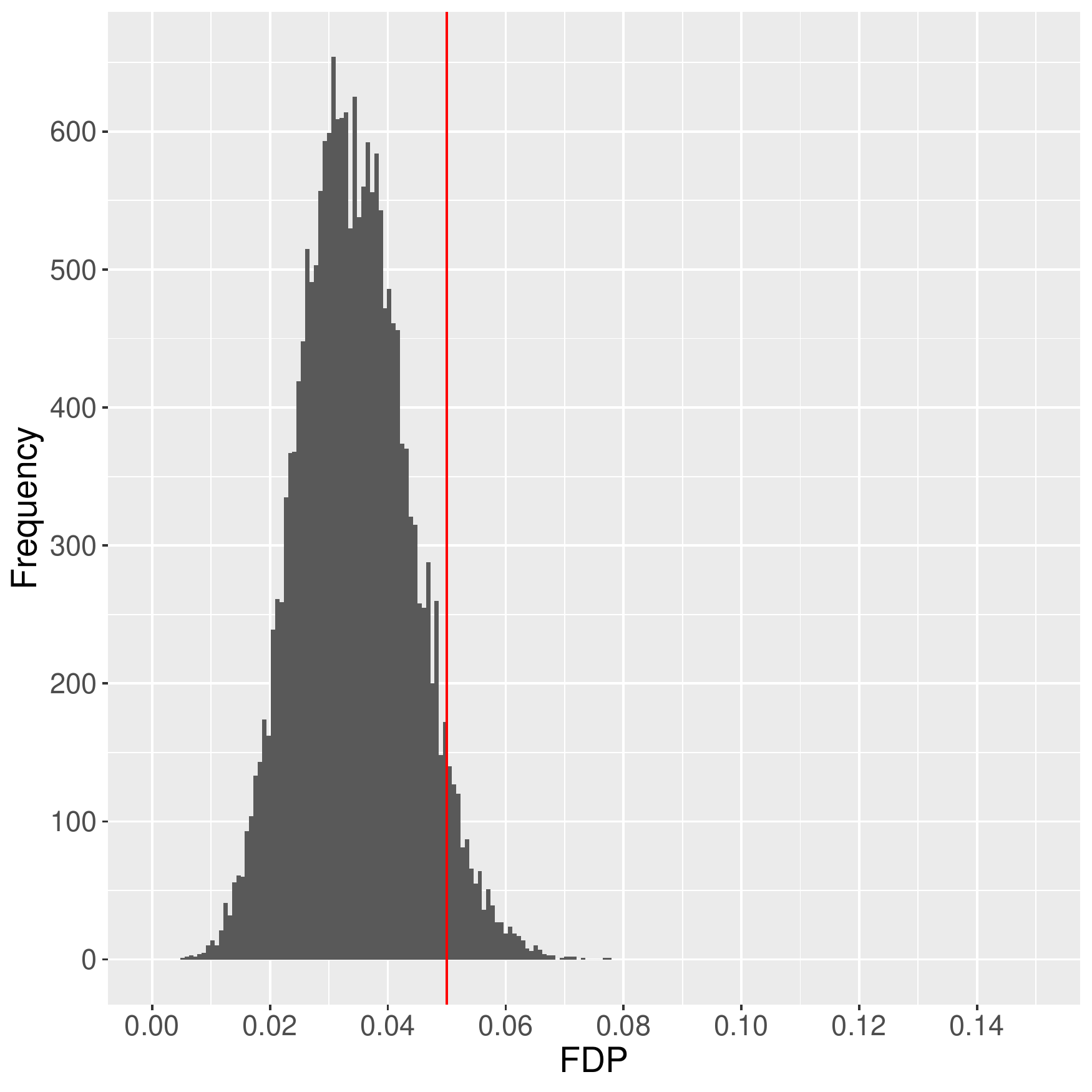}  & \includegraphics[width=2.5in]{{figures_Arya/grid_figure_fdp_d1_c0.5_lambda0.5_m2000_pi0.5_calTRUE_sep3_conf0.05}.pdf}  & \includegraphics[width=2.5in]{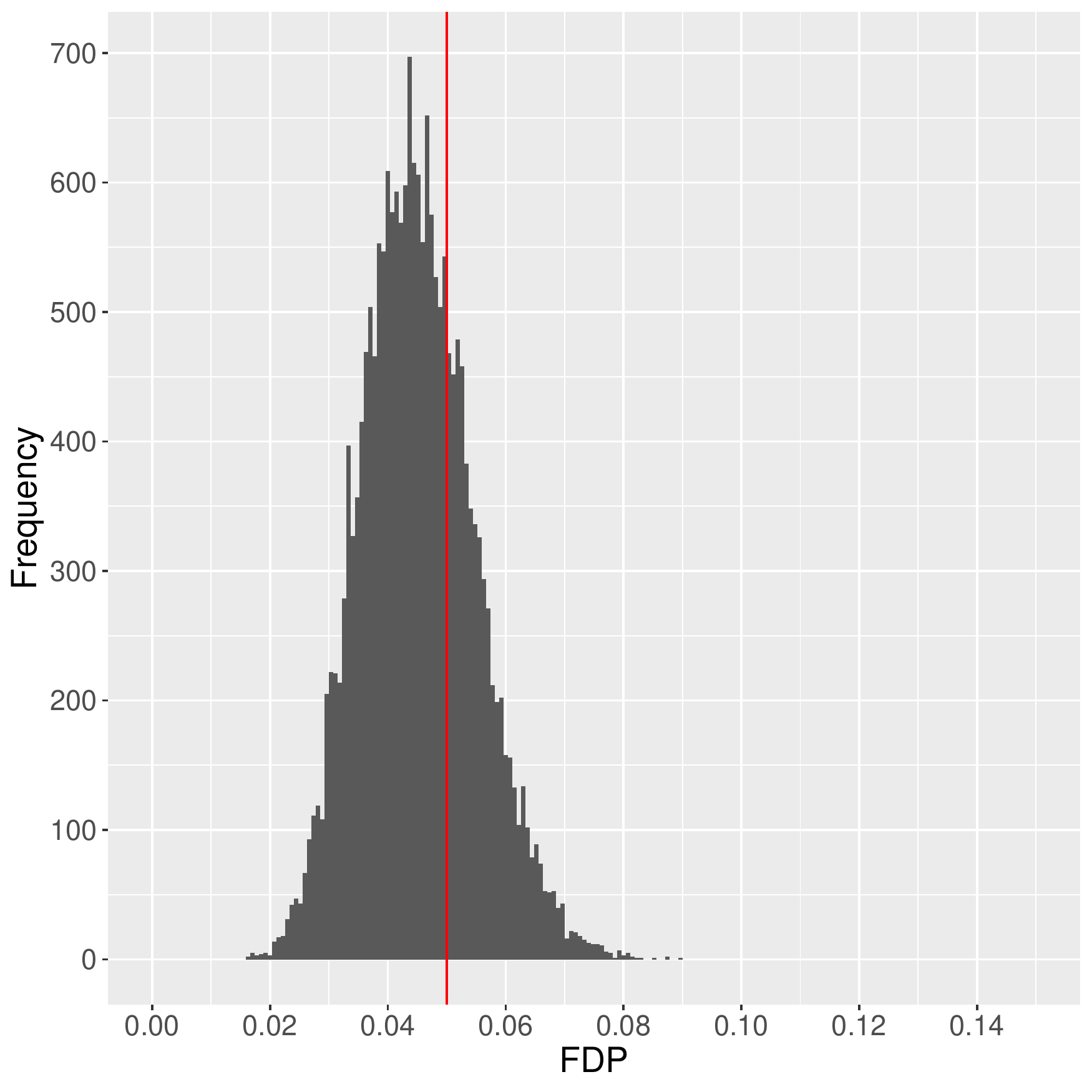} \tabularnewline
\hspace{-10ex}
\includegraphics[width=2.5in]{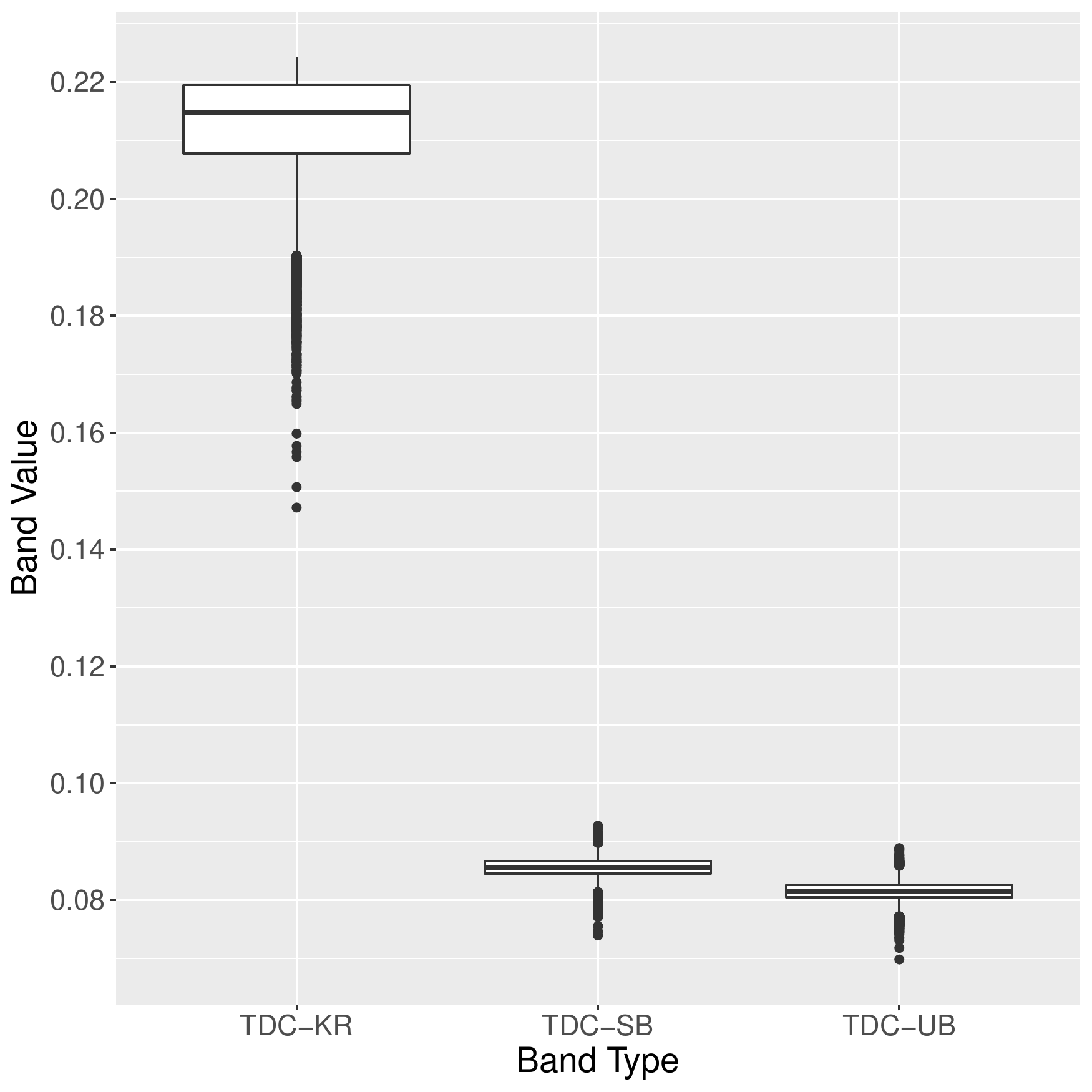}  & \includegraphics[width=2.5in]{{figures_Arya/grid_figure_box_d1_c0.5_lambda0.5_m2000_pi0.5_calTRUE_sep3_conf0.05}.pdf}  & \includegraphics[width=2.5in]{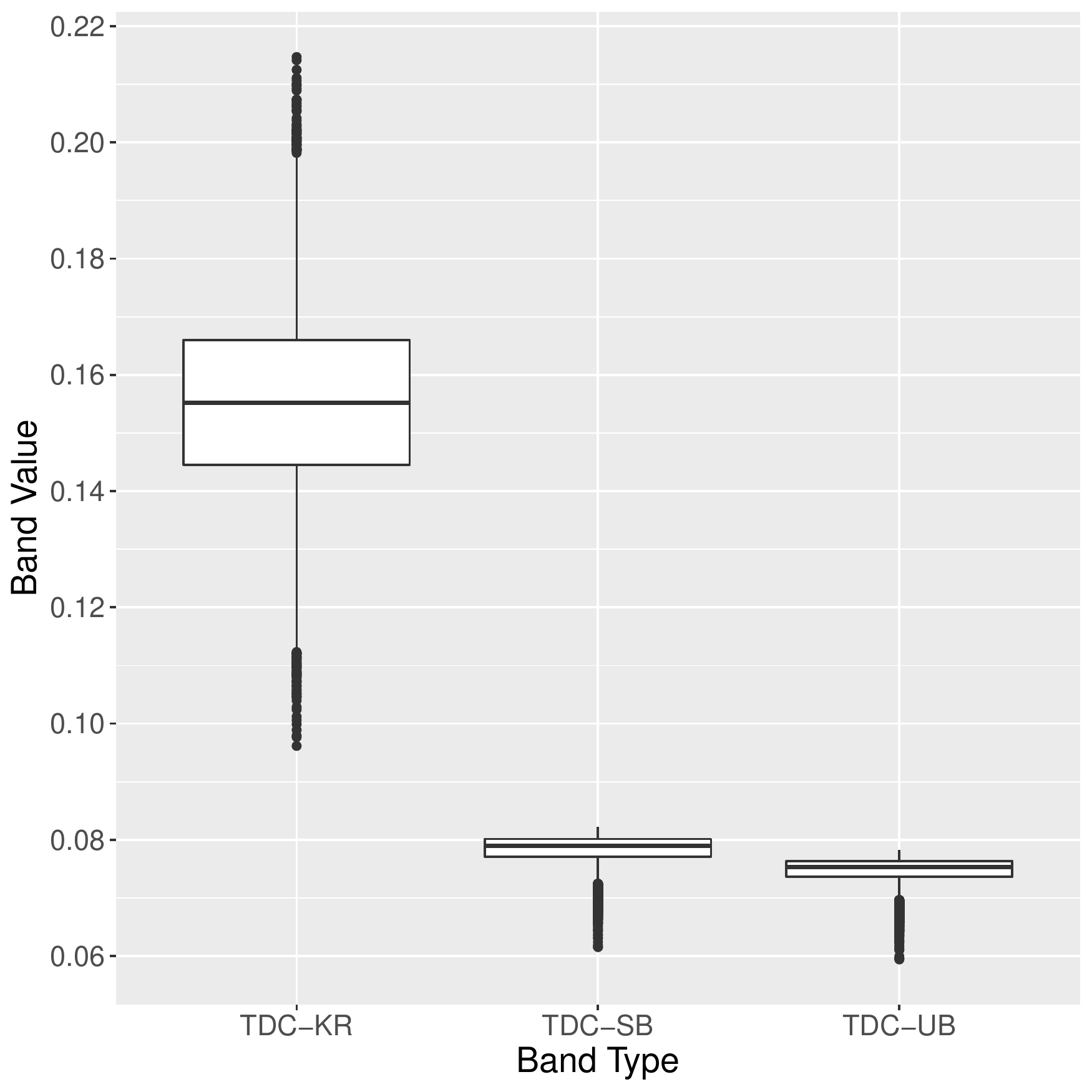} \tabularnewline
\hspace{-10ex}
\includegraphics[width=2.5in]{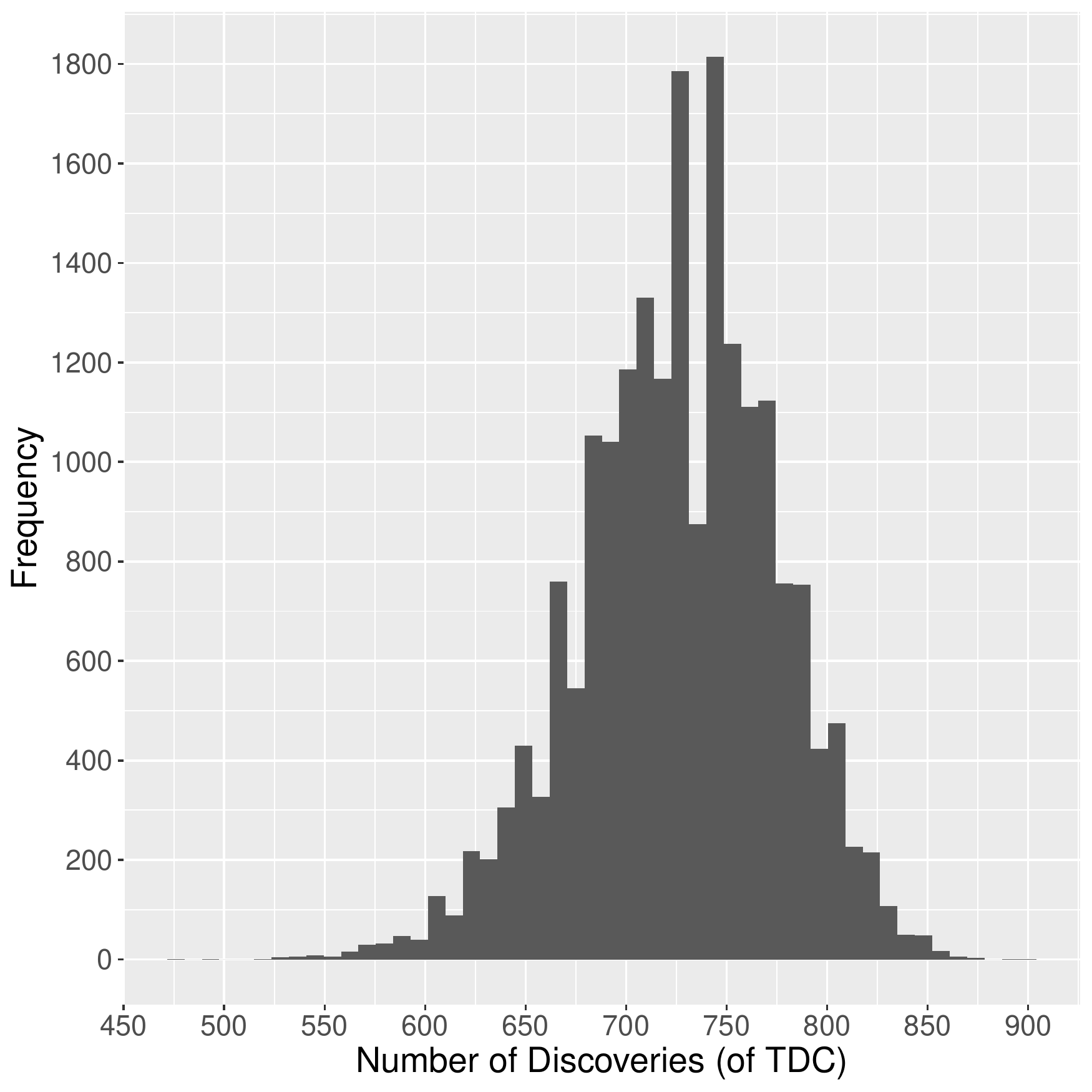}  & \includegraphics[width=2.5in]{{figures_Arya/grid_figure_disc_d1_c0.5_lambda0.5_m2000_pi0.5_calTRUE_sep3_conf0.05}.pdf}  & \includegraphics[width=2.5in]{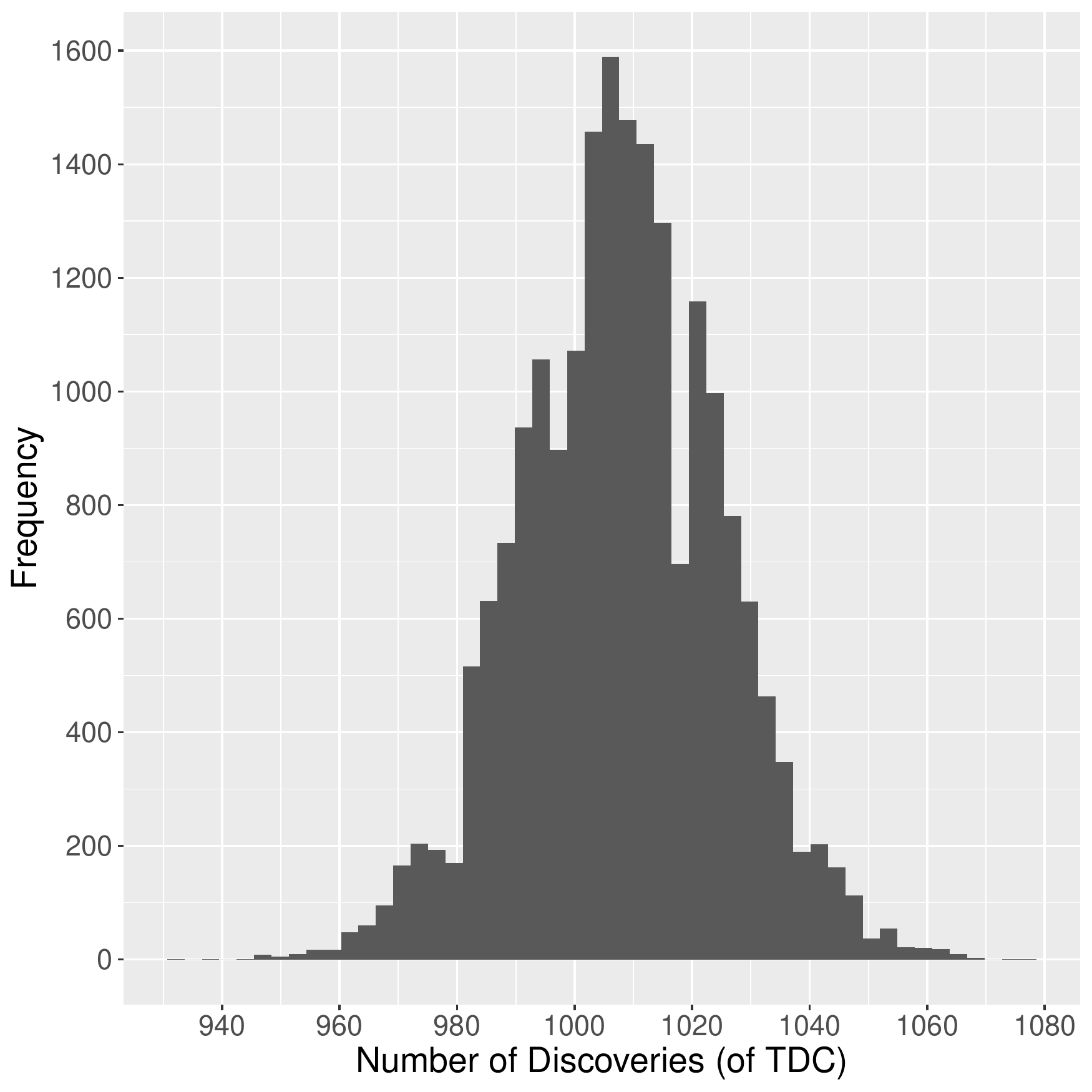} \tabularnewline
\end{tabular}
\caption{\textbf{Varying $\rho$ in simulated experiments.} Similar to \supfig~\ref{supfig:vary_m} only here we increase $\rho$
keeping $\alp=0.05$, $\gam=0.05$, $\pi_0 = 0.5$ and $m=$ 2K.  TDC's FDP (top row), the comparison of the three interpolated upper prediction bands (middle row) and the number of discoveries of TDC (bottom row).
\label{supfig:vary_rho}}
\end{figure}

\clearpage

\begin{figure}
\centering %
\begin{tabular}{ccc}
\hspace{-10ex}
$\color{blue}m = 500$, $\pi_0 = 0.5$, $\rho = 3$  & $\color{blue}m = 2000$, $\pi_0 = 0.5$, $\rho = 3$ & $\color{blue}m = 10000$, $\pi_0 = 0.5$, $\rho = 3$ \tabularnewline
\hspace{-10ex}
\includegraphics[width=2.5in]{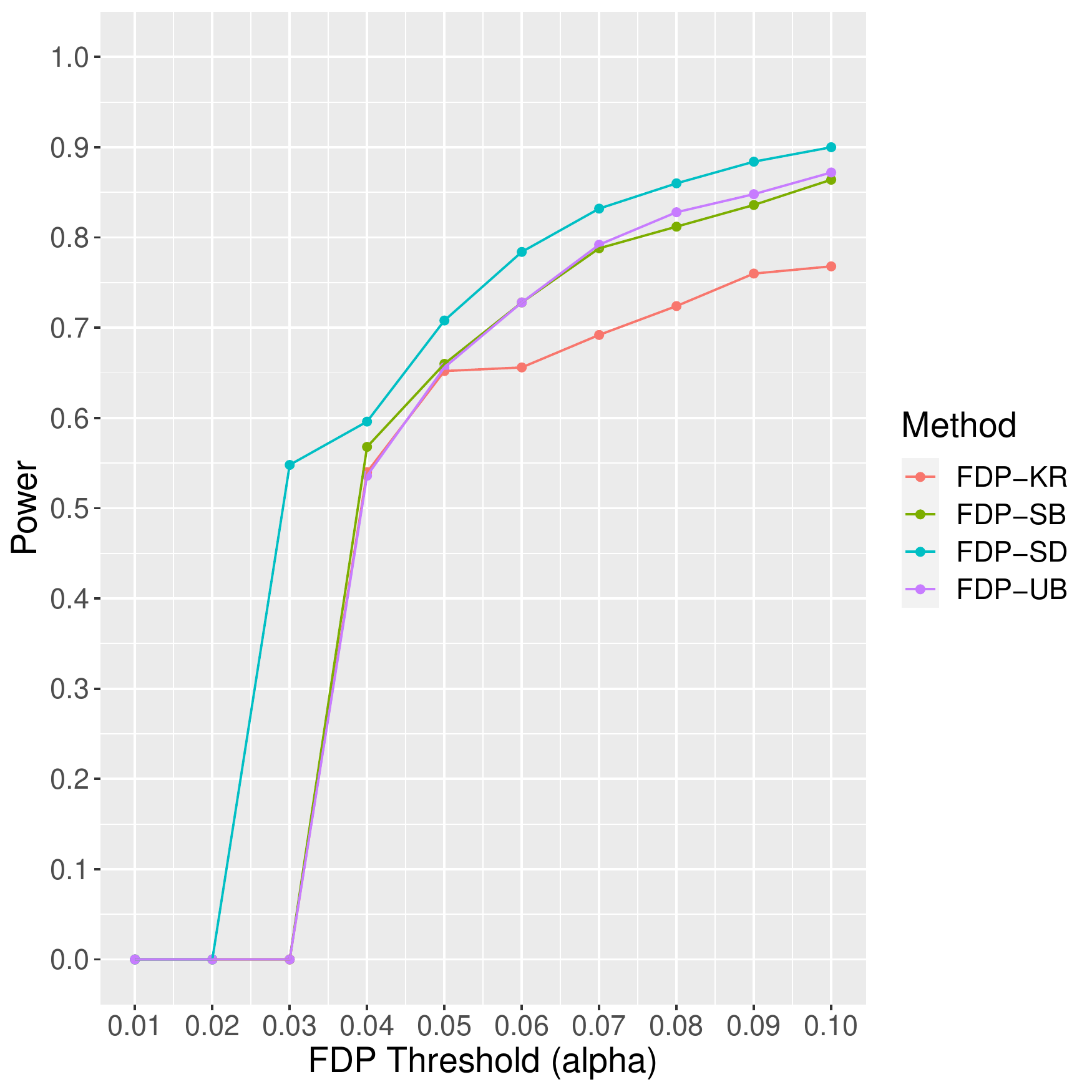}  & \includegraphics[width=2.5in]{{figures_Arya/power_figure_m2000_pi0.5_calTRUE_sep3_conf0.05}.pdf}  & \includegraphics[width=2.5in]{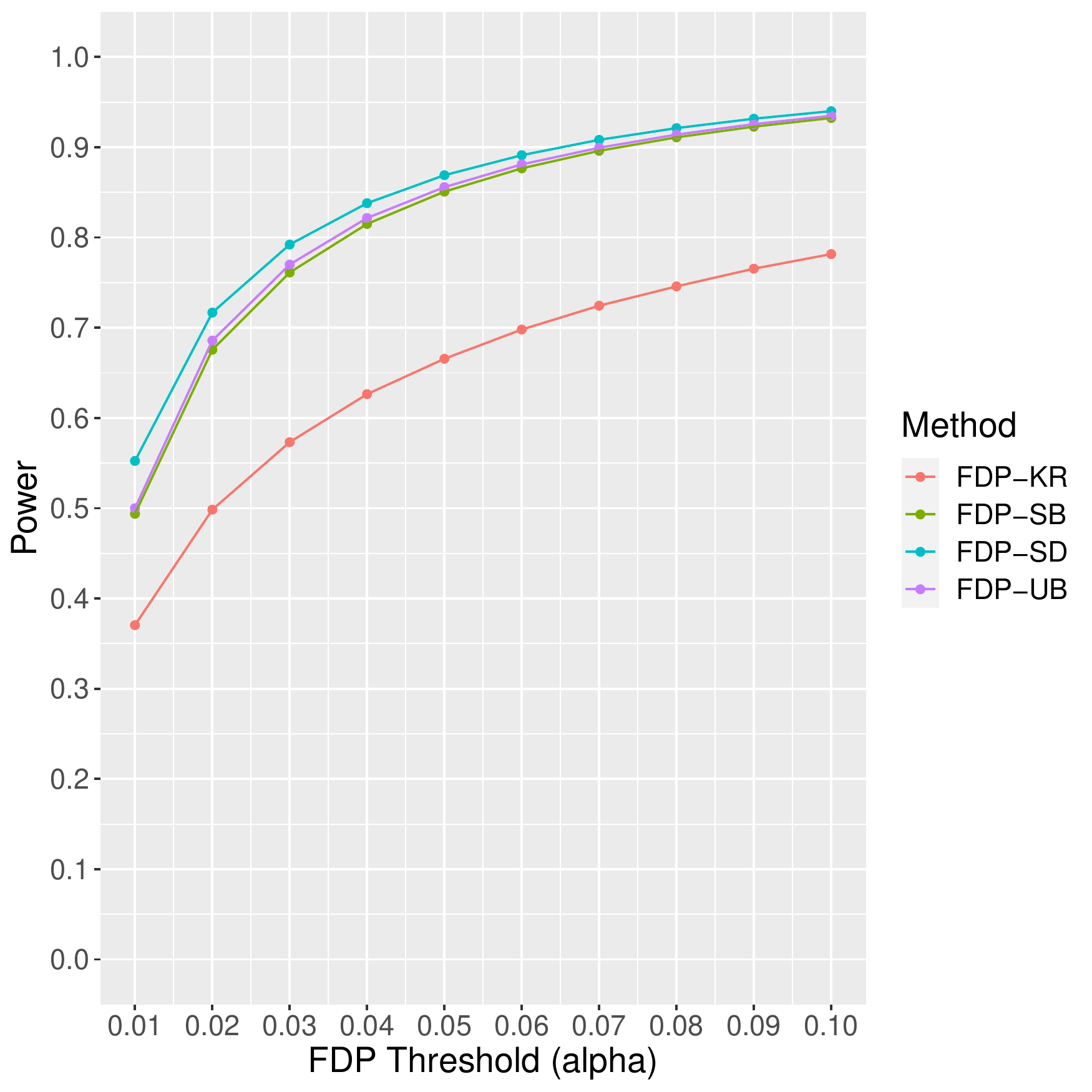} \tabularnewline
\hspace{-10ex}
$m = 2000$, $\color{blue}\pi_0 = 0.8$, $\rho = 3$  & $m = 2000$, $\color{blue}\pi_0 = 0.5$, $\rho = 3$ & $m = 2000$, $\color{blue}\pi_0 = 0.2$, $\rho = 3$ \tabularnewline
\hspace{-10ex}
\includegraphics[width=2.5in]{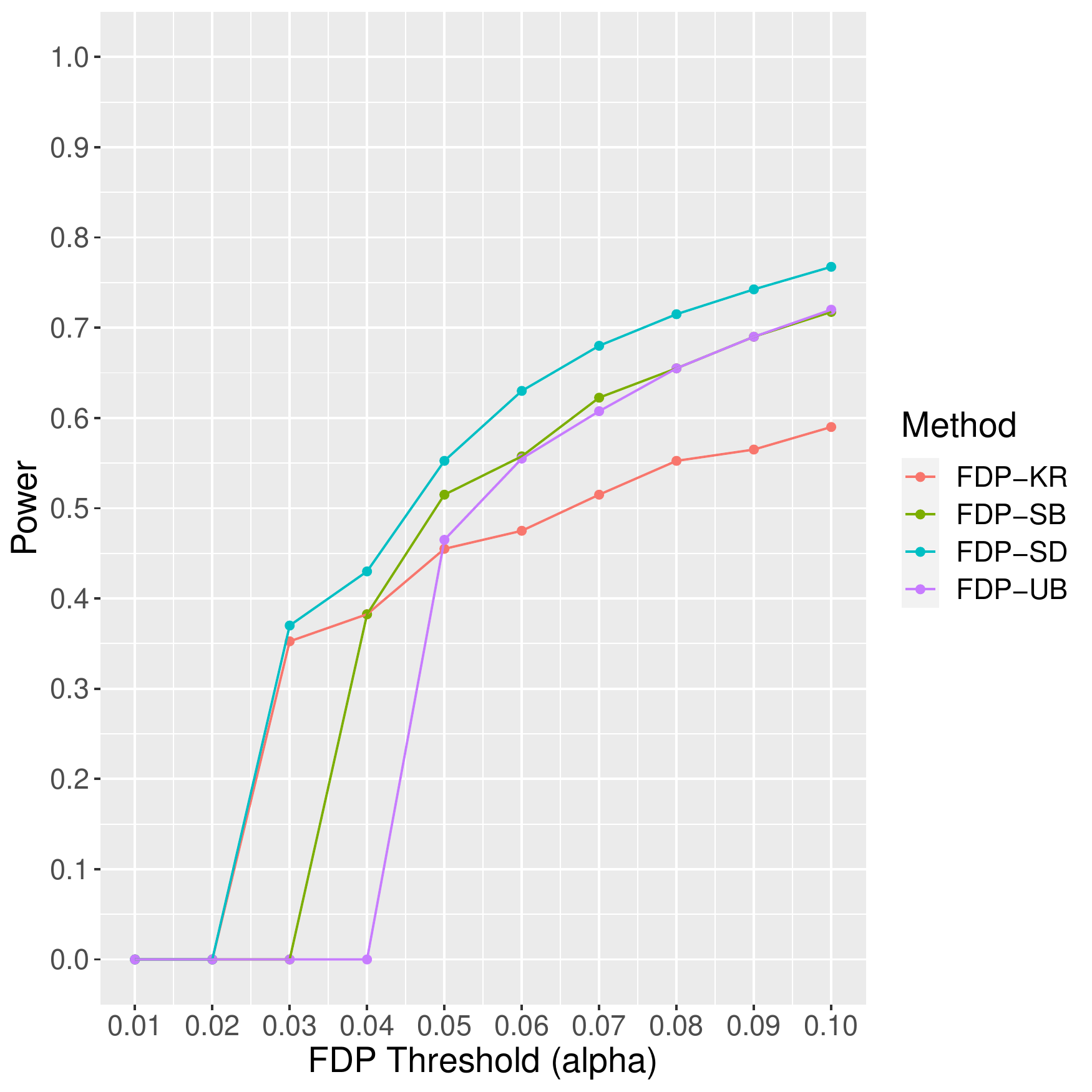}  & \includegraphics[width=2.5in]{{figures_Arya/power_figure_m2000_pi0.5_calTRUE_sep3_conf0.05}.pdf}  & \includegraphics[width=2.5in]{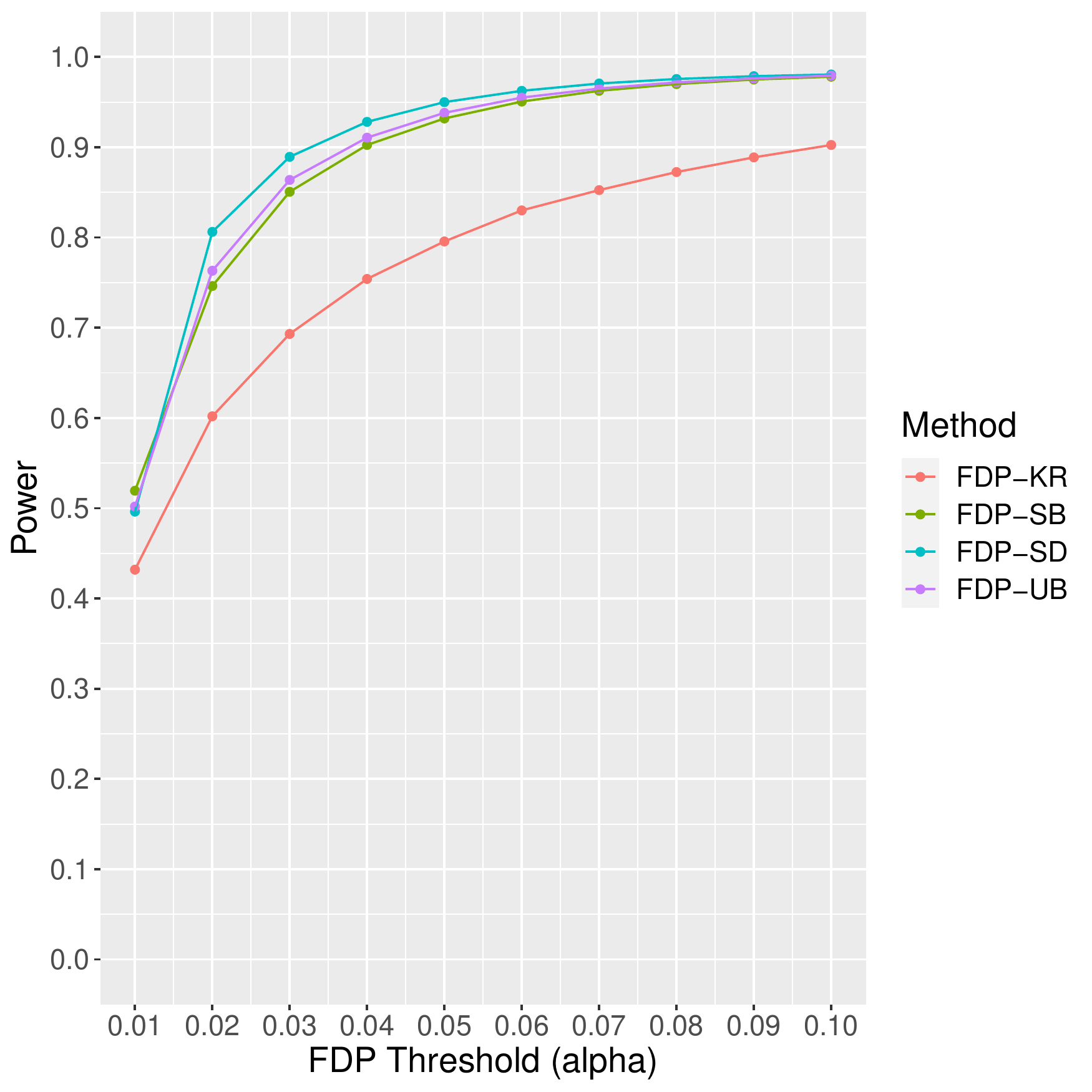} \tabularnewline
\hspace{-10ex}
$m = 2000$, $\pi_0 = 0.5$, $\color{blue}{\rho = 2.5}$  & $m = 2000$, $\pi_0 = 0.5$, $\color{blue}{\rho = 3}$ & $m = 2000$, $\pi_0 = 0.5$, $\color{blue}{\rho = 3.5}$ \tabularnewline
\hspace{-10ex}
\includegraphics[width=2.5in]{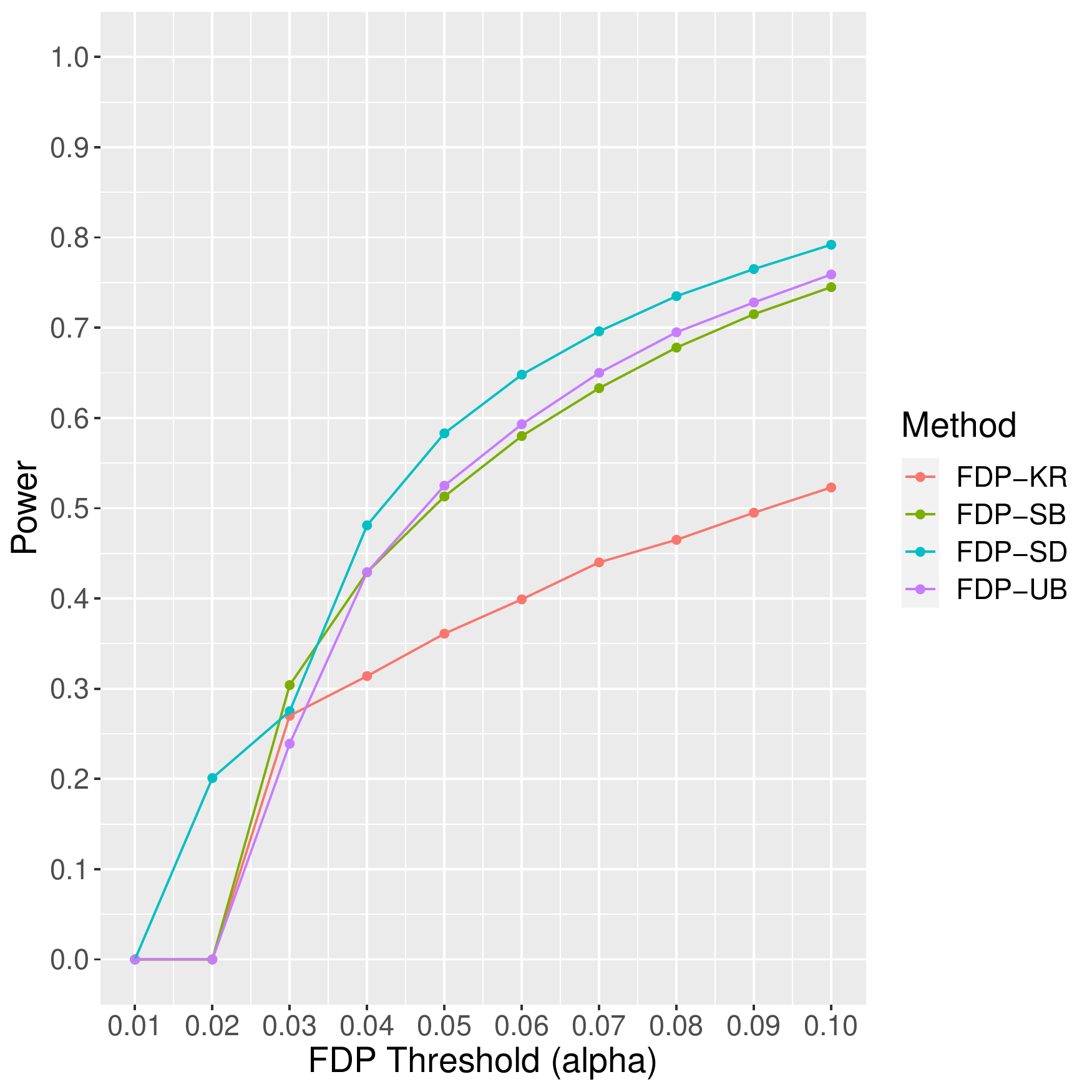}  & \includegraphics[width=2.5in]{{figures_Arya/power_figure_m2000_pi0.5_calTRUE_sep3_conf0.05}.pdf}  & \includegraphics[width=2.5in]{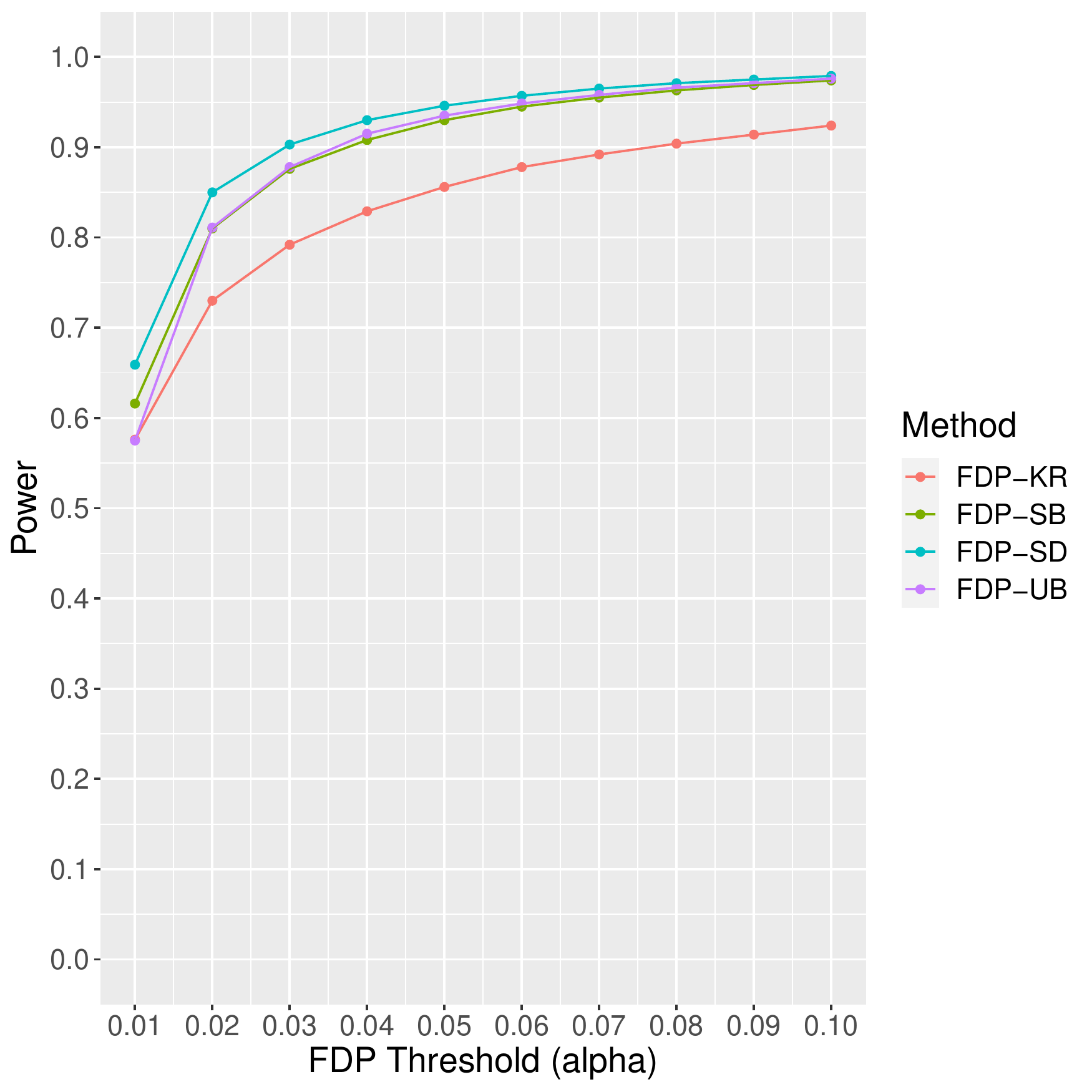} \tabularnewline
\end{tabular}
\caption{\textbf{Median power of FDP-controlling procedures.} 
Plotted are the median power (over 20k datasets) of FDP-controlling procedures with  $\alpha \in \{0.01, 0.02, \ldots, 0.1\}$
and a fixed confidence level of $1-\gamma = 0.95$. 
The datasets were generated using our normal mixture model with parameters that are displayed above the plots (the one highlighted in blue varies across each row).
\label{supfig:fdp-control}}
\end{figure}

\clearpage

\begin{figure}
\centering
\begin{tabular}{ccc}
\includegraphics[width=0.45\textwidth]{figures_Arya/pride_upbands_d1_c0.5_lambda0.5.pdf} \\
\textbf{1 decoy} \\[6pt]
\end{tabular} \tabularnewline
\begin{tabular}{cccc}
\includegraphics[width=0.45\textwidth]{figures_Arya/pride_upbands_d3_c0.5_lambda0.5.pdf} &
\includegraphics[width=0.45\textwidth]{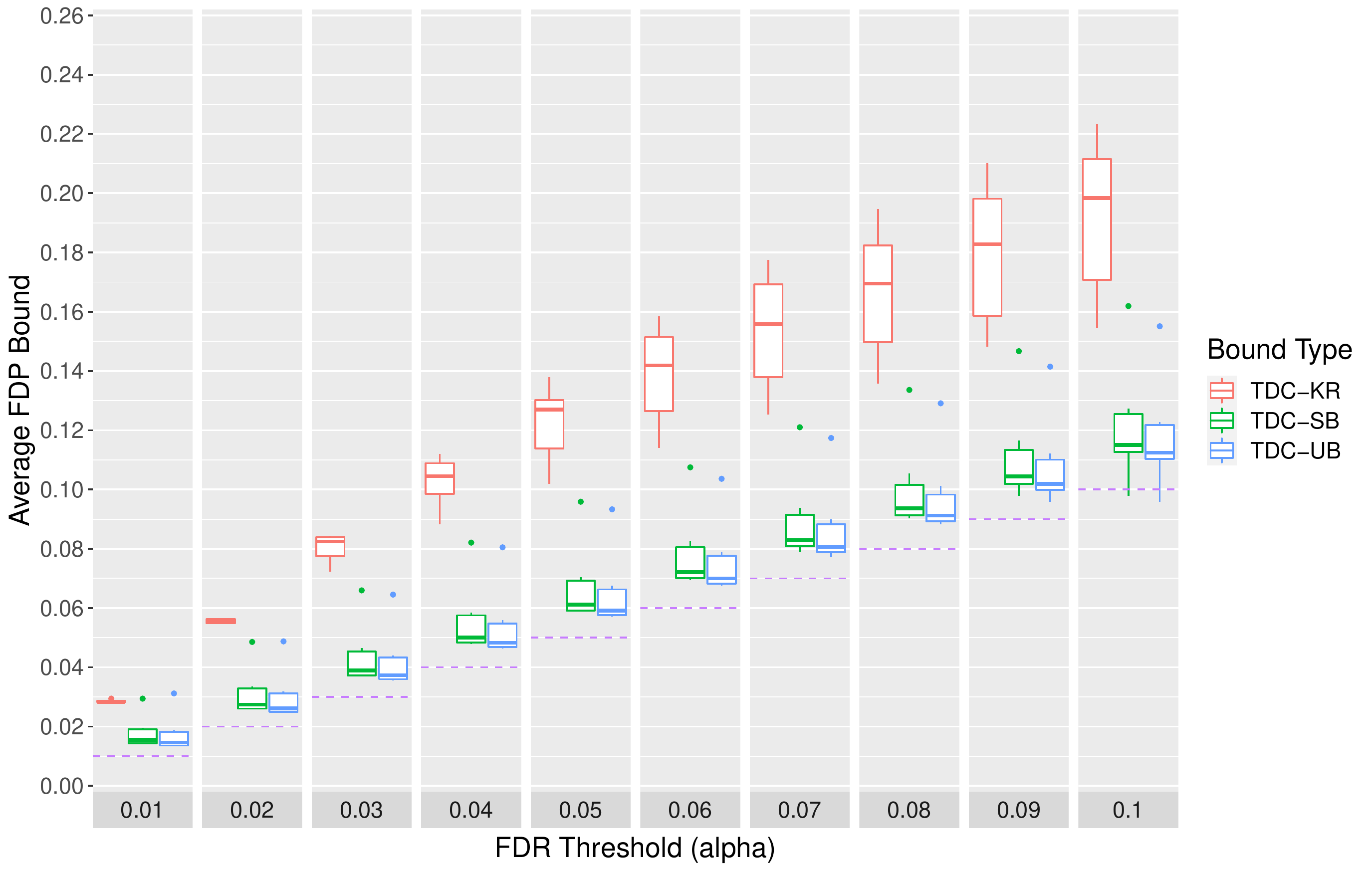} \\
\textbf{3 decoys (mirror)}  & \textbf{3 decoys (max)}  \\[6pt]
\end{tabular}
\begin{tabular}{cccc}
\includegraphics[width=0.45\textwidth]{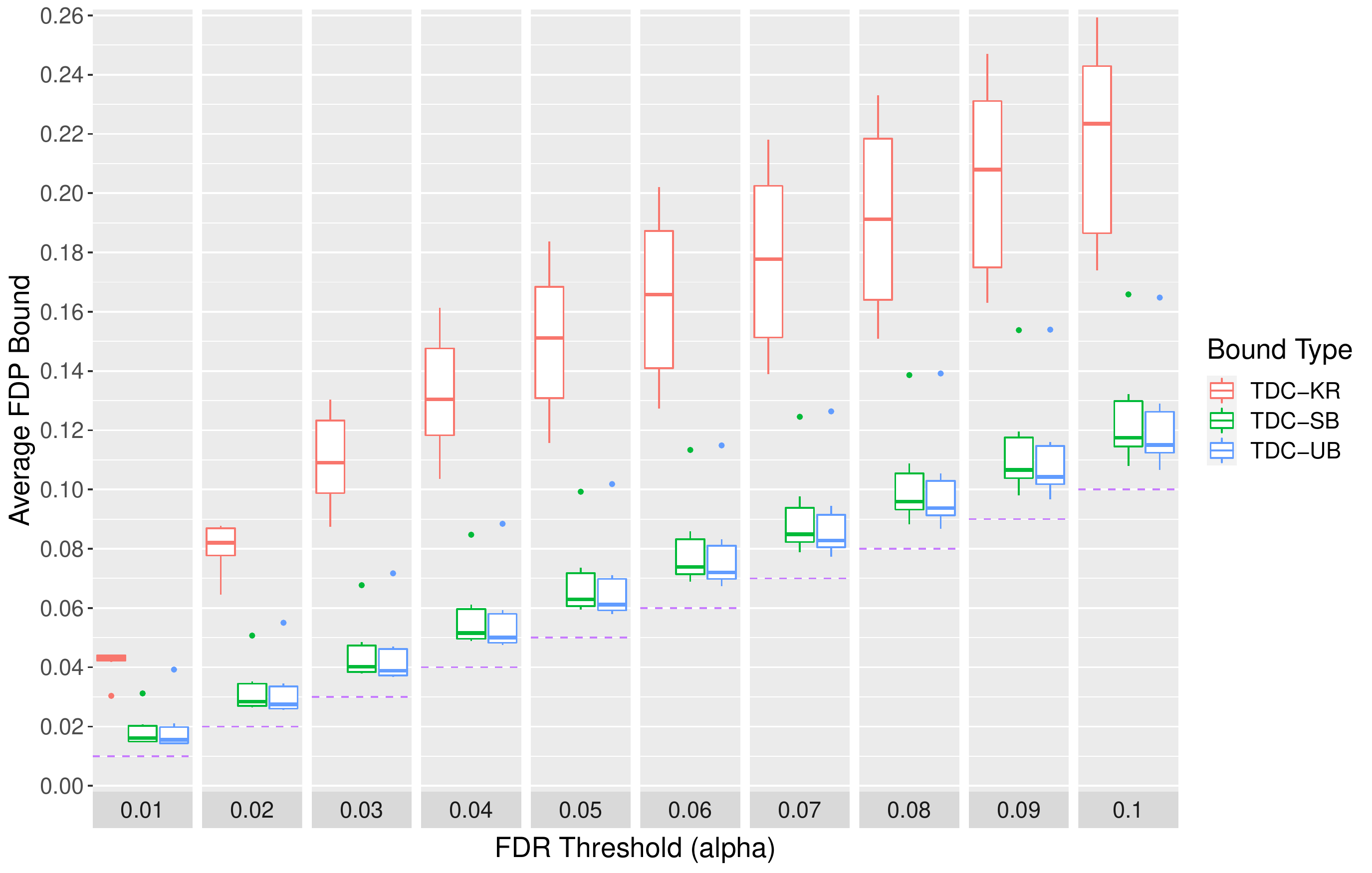} &
\includegraphics[width=0.45\textwidth]{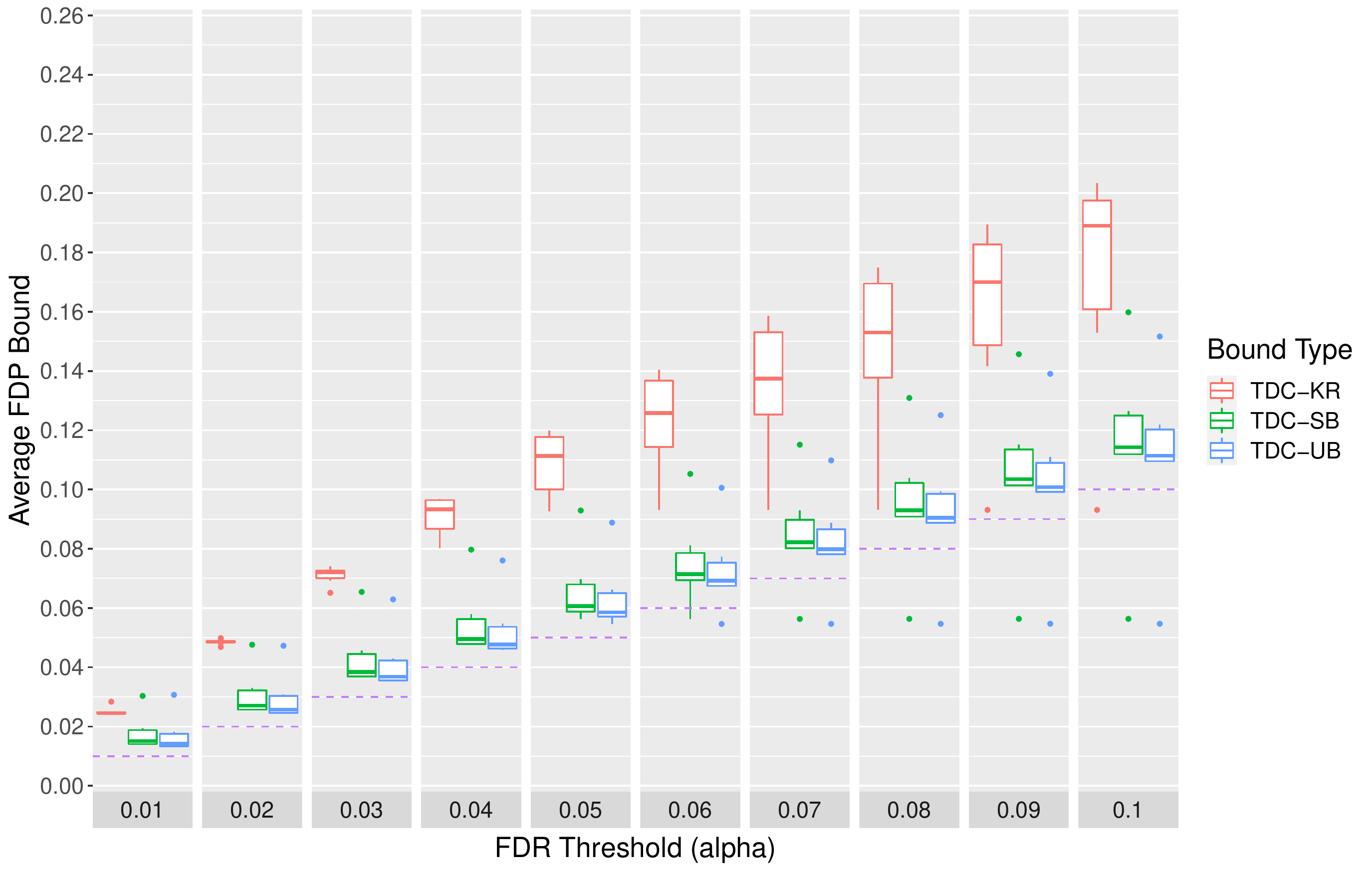} \\
\textbf{7 decoys (mirror)}  & \textbf{7 decoys (max)}  \\[6pt]
\end{tabular}
\caption{\textbf{Applications of TDC-KR, TDC-SB and TDC-UB in peptide detection.} For each of the 10 PRIDE datasets, we computed the average, across 20 decoys, of the upper prediction bound on FDP the list of peptides discovered by TDC or its multiple-decoy version. The number of decoys and the method used are written below each figure. Shown are boxplots of the 10 averages, across FDR tolerances $\alpha \in \{0.01, \ldots, 0.1\}$, as well as a purple dashed line corresponding to $y = \alpha$. Here we use the bands TDC-KR (red), TDC-SB (green) and TDC-UB (blue). We note that, except in cases where there are few discoveries to be made (in particular, PXD029319), TDC-SB and TDC-UB offer significantly tighter bounds, with TDC-UB generally the tightest.}
\label{supfig:upband-boxplot}
\end{figure}

\clearpage

\begin{figure}
\centering %
\begin{tabular}{ccc}
\hspace{-10ex}
$\color{blue}m = 500$, $\pi_0 = 0.5$, $\rho = 3$  & $\color{blue}m = 2000$, $\pi_0 = 0.5$, $\rho = 3$ & $\color{blue}m = 10000$, $\pi_0 = 0.5$, $\rho = 3$ \tabularnewline
\hspace{-10ex}
\includegraphics[width=2.5in]{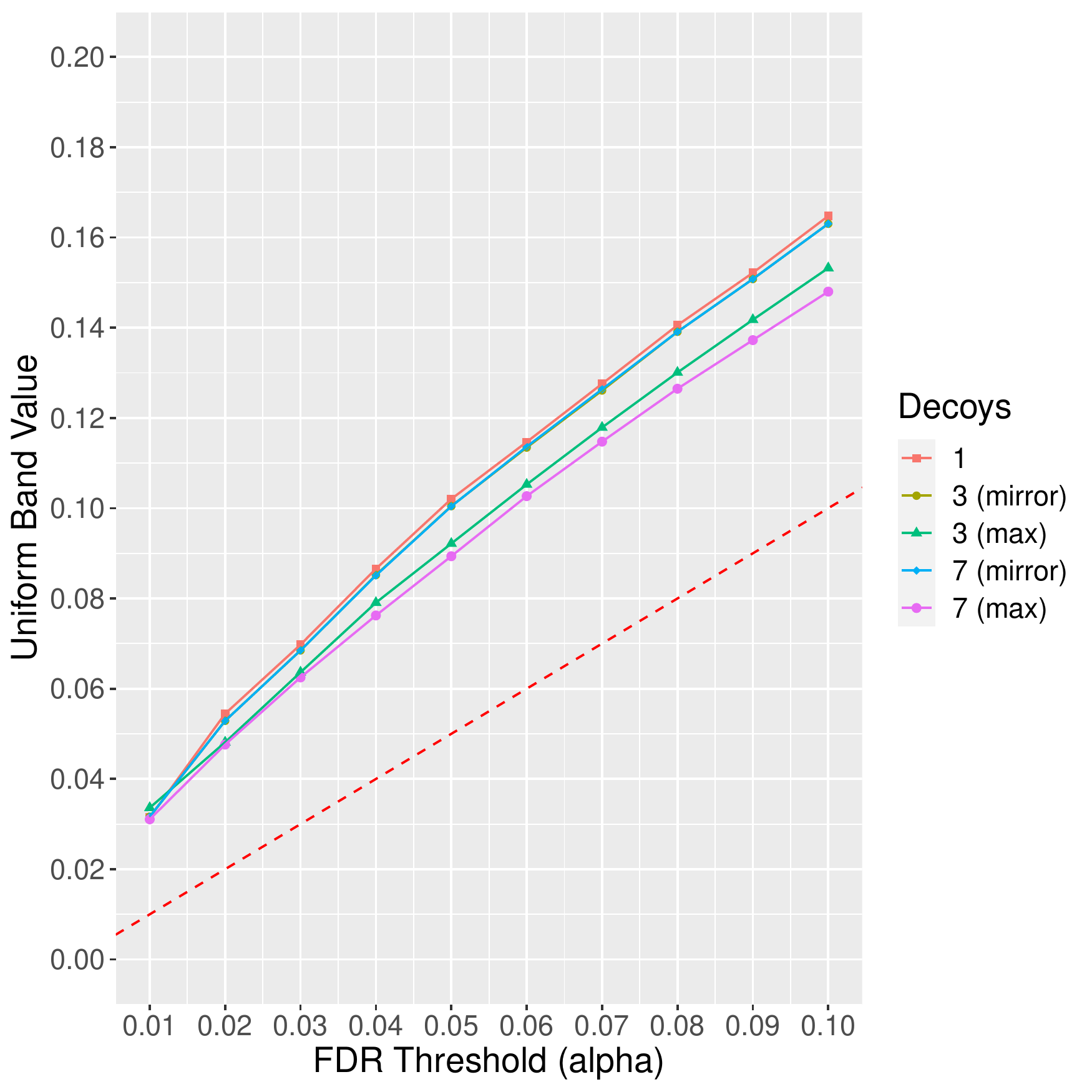}  & \includegraphics[width=2.5in]{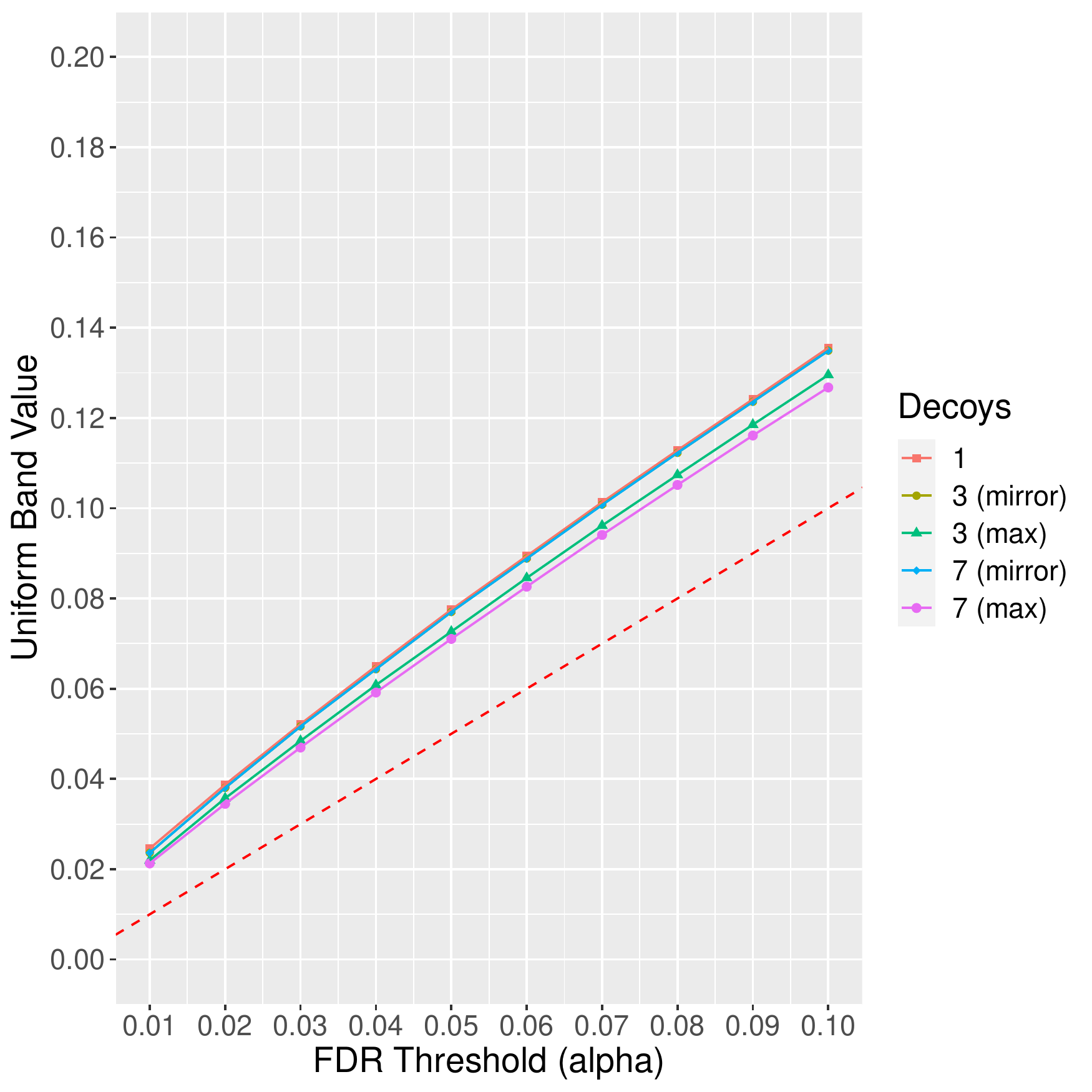}  & \includegraphics[width=2.5in]{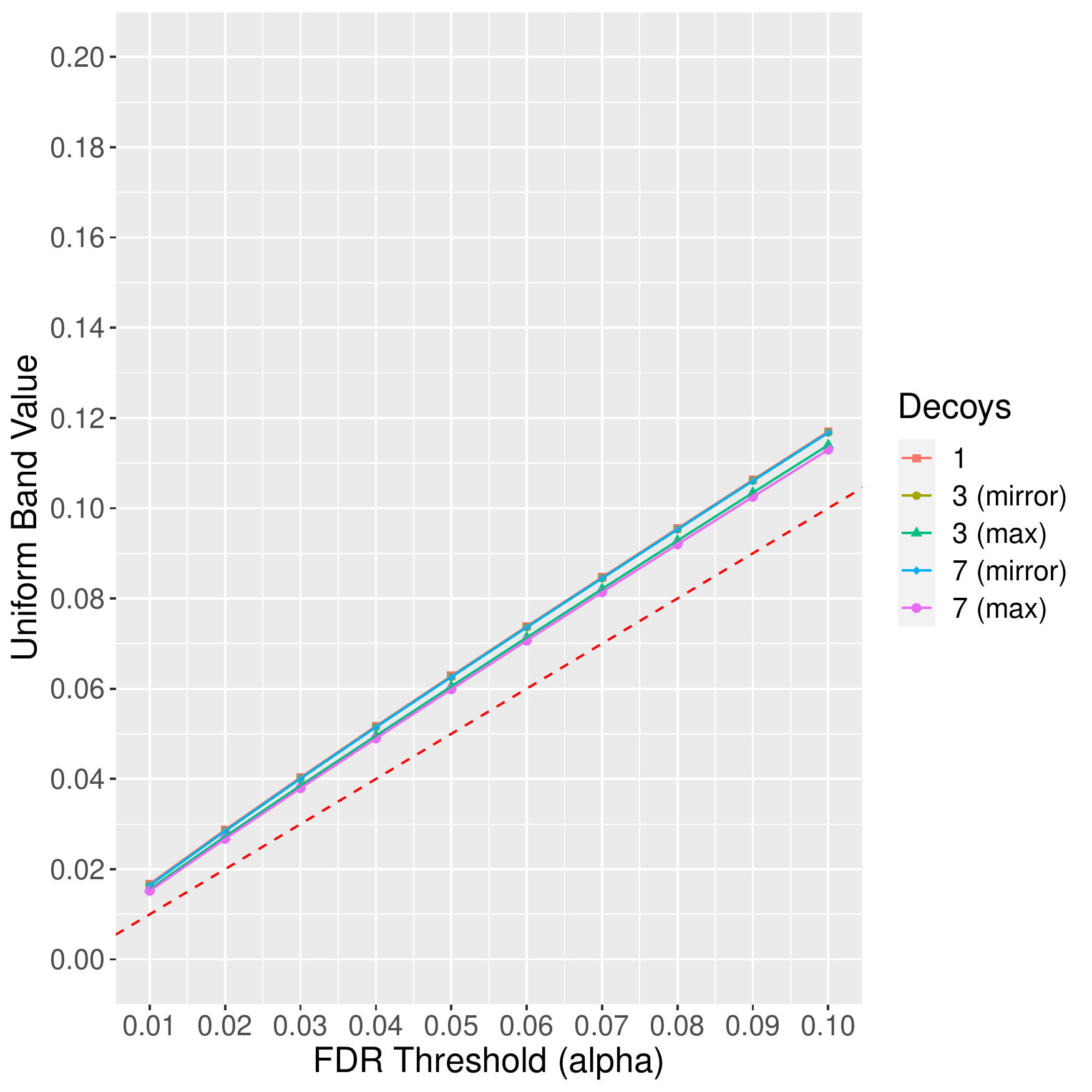} \tabularnewline
\hspace{-10ex}
$m = 2000$, $\color{blue}\pi_0 = 0.8$, $\rho = 3$  & $m = 2000$, $\color{blue}\pi_0 = 0.5$, $\rho = 3$ & $m = 2000$, $\color{blue}\pi_0 = 0.2$, $\rho = 3$ \tabularnewline
\hspace{-10ex}
\includegraphics[width=2.5in]{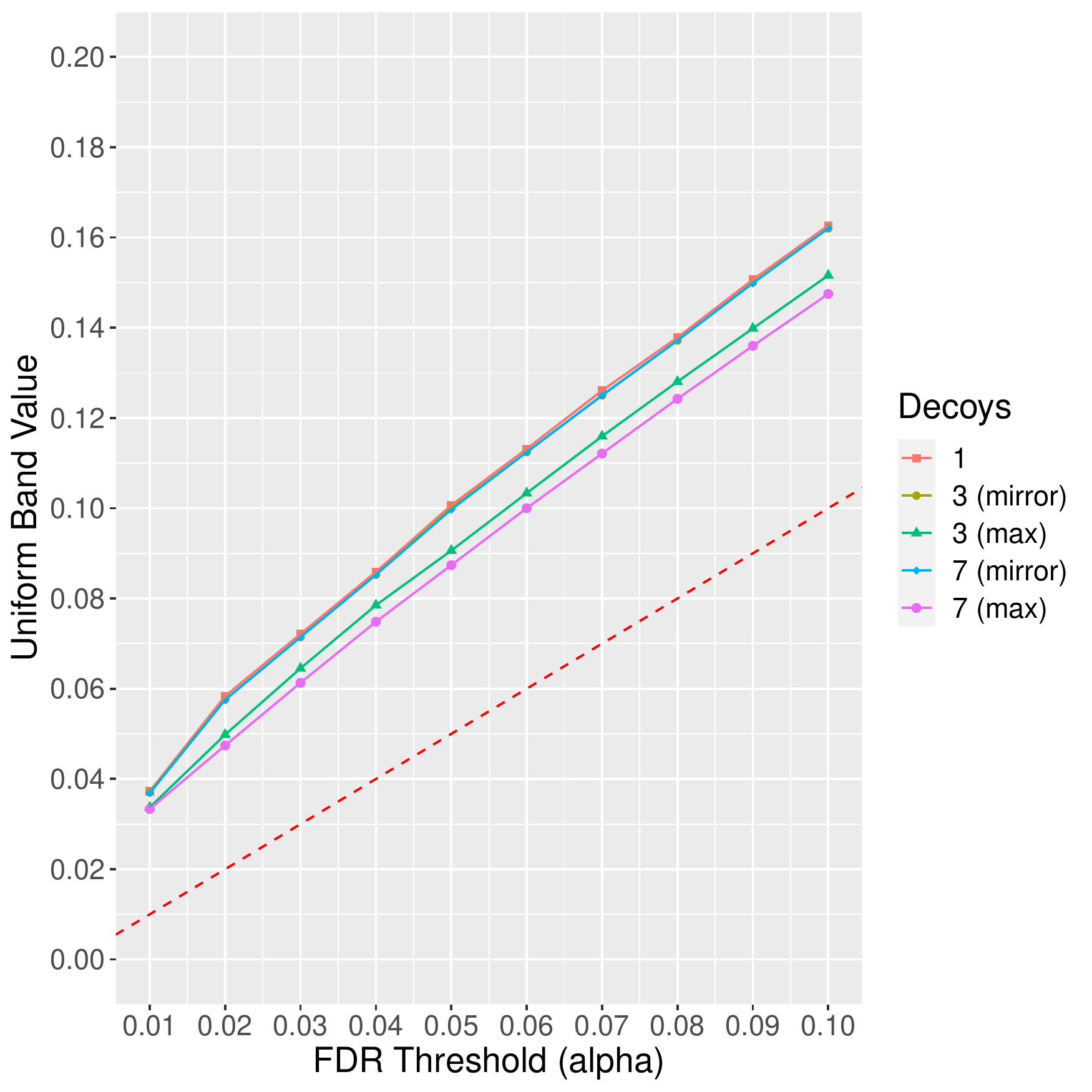}  & \includegraphics[width=2.5in]{{figures_Arya/multiple_decoys_method_figure_m2000_pi0.5_calTRUE_sep3_alpha0.05_conf0.05}.pdf}  & \includegraphics[width=2.5in]{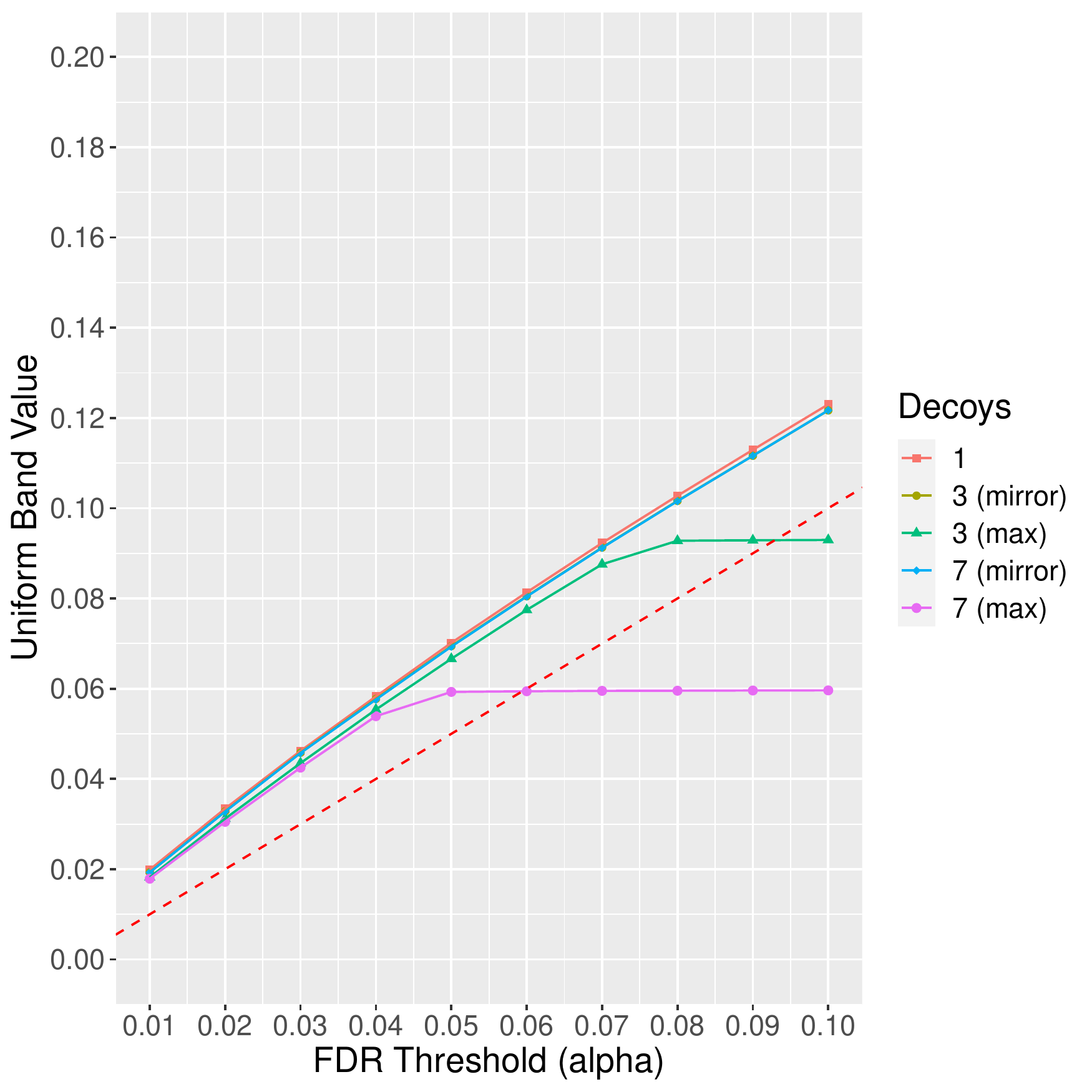} \tabularnewline
\hspace{-10ex}
$m = 2000$, $\pi_0 = 0.5$, $\color{blue}{\rho = 2.5}$  & $m = 2000$, $\pi_0 = 0.5$, $\color{blue}{\rho = 3}$ & $m = 2000$, $\pi_0 = 0.5$, $\color{blue}{\rho = 3.5}$ \tabularnewline
\hspace{-10ex}
\includegraphics[width=2.5in]{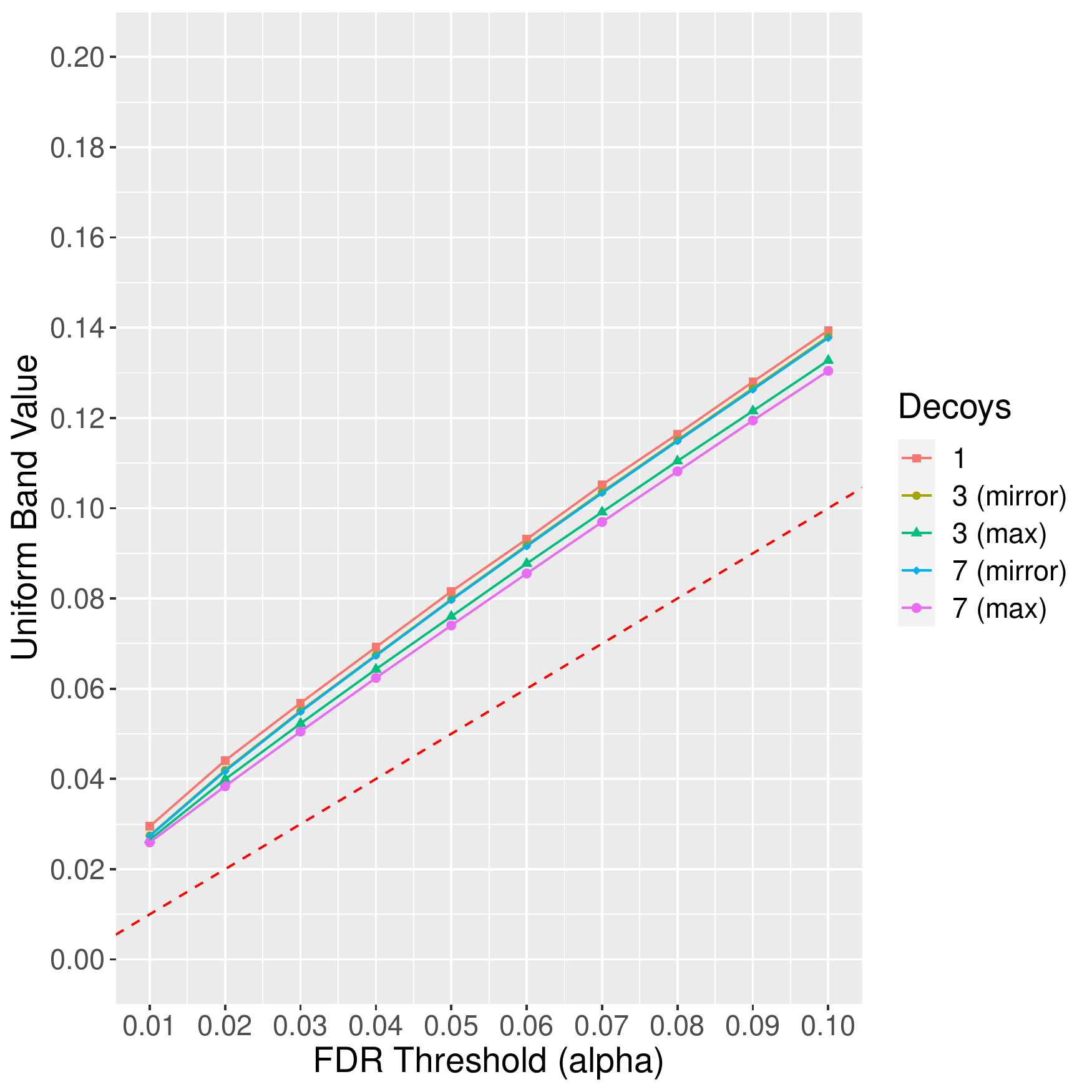}  & \includegraphics[width=2.5in]{{figures_Arya/multiple_decoys_method_figure_m2000_pi0.5_calTRUE_sep3_alpha0.05_conf0.05}.pdf}  & \includegraphics[width=2.5in]{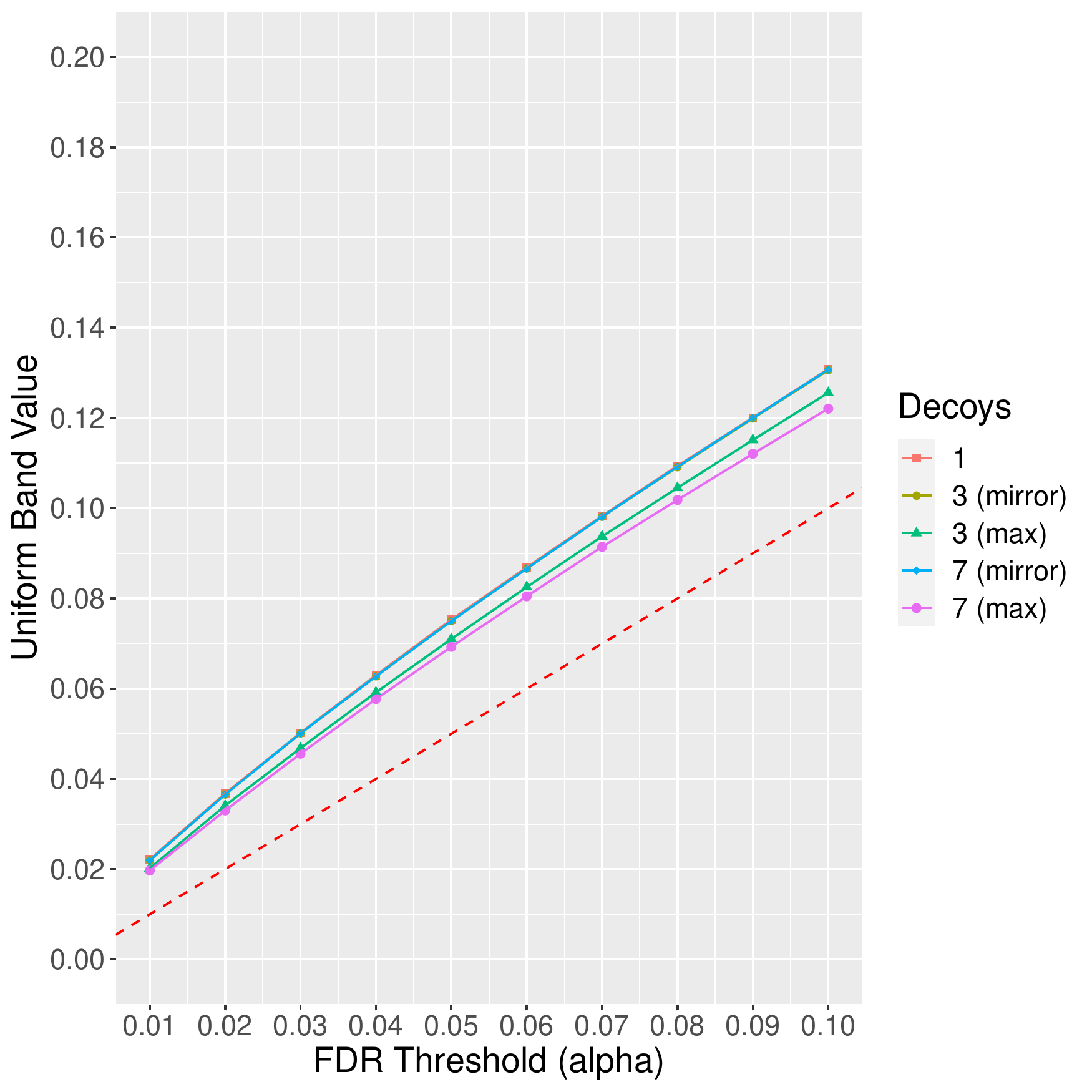} \tabularnewline
\end{tabular}
\caption{\textbf{Comparing the performance of TDC-UB using multiple decoys.} We plot the FDP bound returned by TDC-UB when using 1, 3 and 7 decoys. For 3 and 7 decoys we considered the max and mirror methods separately. Each point on the figure is the median on the same simulated set of 20K calibrated scores, using $\gamma = 0.05$ and the remaining parameters shown on top of each figure. The parameter in blue is that which varies across a given row. Included in the figure is the red-dashed line $x = y$.
\label{supfig:mult_method}}
\end{figure}

\clearpage

\begin{figure}
\centering %
\begin{tabular}{ccc}
\hspace{-10ex}
$\color{blue}m = 500$, $\pi_0 = 0.5$, $\rho = 3$  & $\color{blue}m = 2000$, $\pi_0 = 0.5$, $\rho = 3$ & $\color{blue}m = 10000$, $\pi_0 = 0.5$, $\rho = 3$ \tabularnewline
\hspace{-10ex}
\includegraphics[width=2.5in]{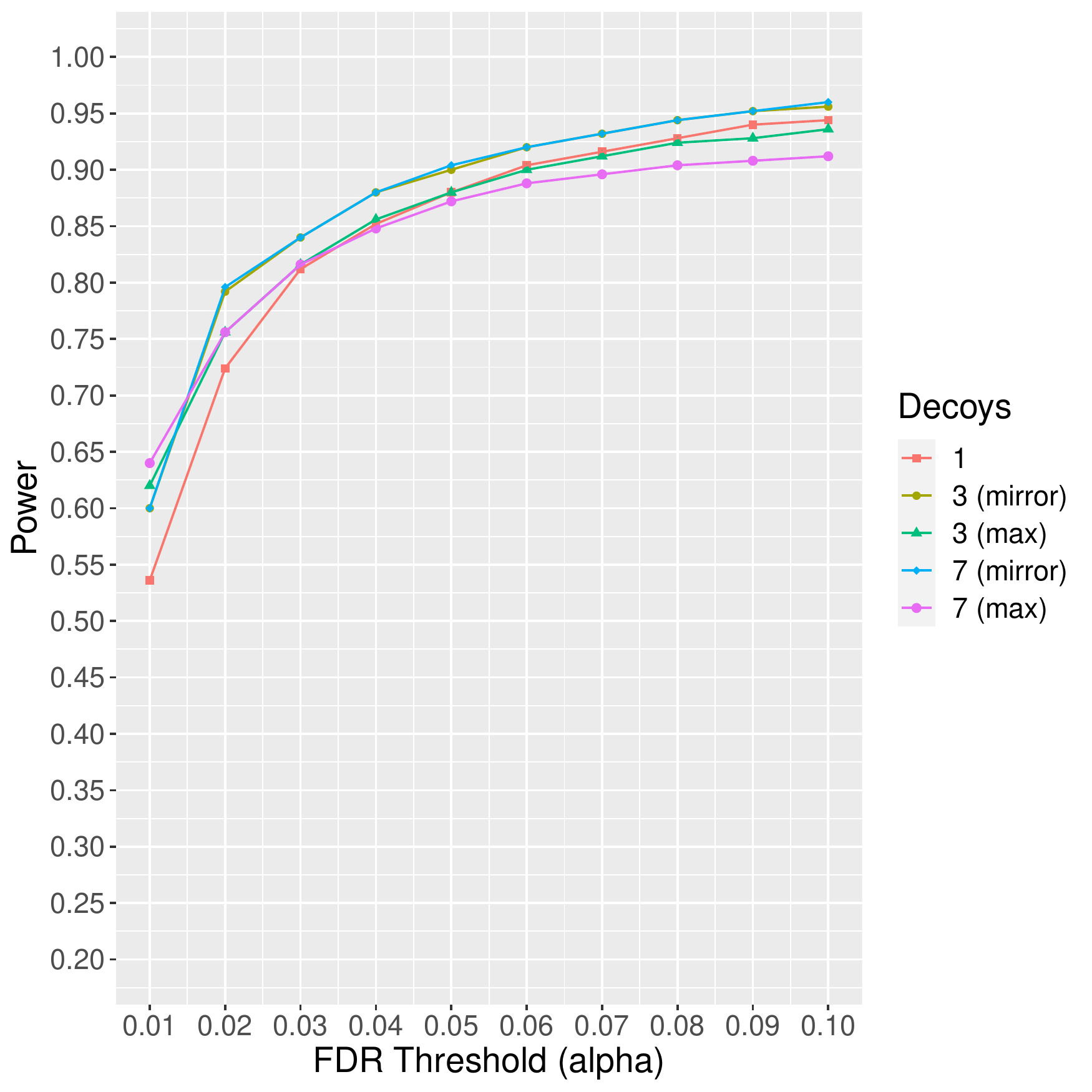}  & \includegraphics[width=2.5in]{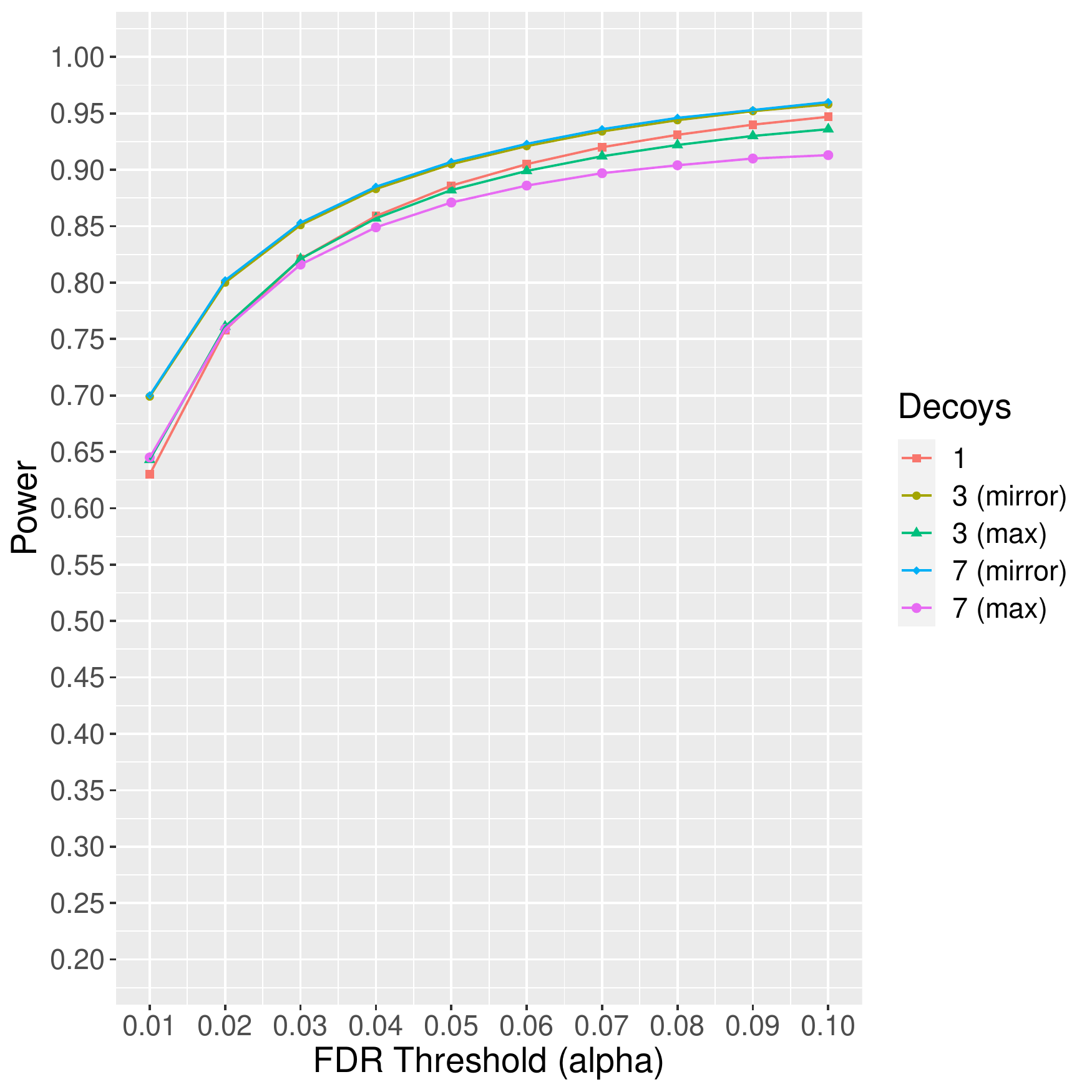}  & \includegraphics[width=2.5in]{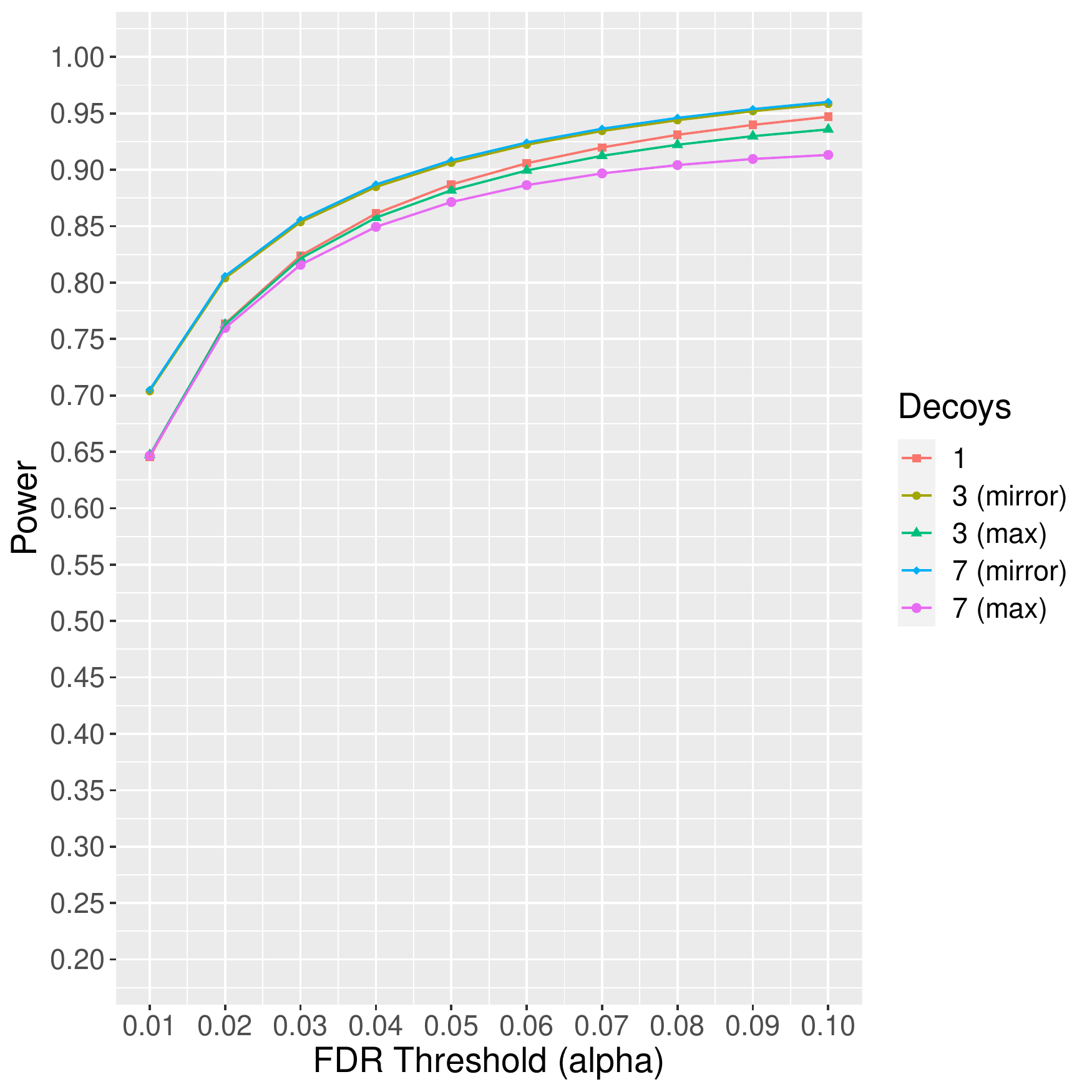} \tabularnewline
\hspace{-10ex}
$m = 2000$, $\color{blue}\pi_0 = 0.8$, $\rho = 3$  & $m = 2000$, $\color{blue}\pi_0 = 0.5$, $\rho = 3$ & $m = 2000$, $\color{blue}\pi_0 = 0.2$, $\rho = 3$ \tabularnewline
\hspace{-10ex}
\includegraphics[width=2.5in]{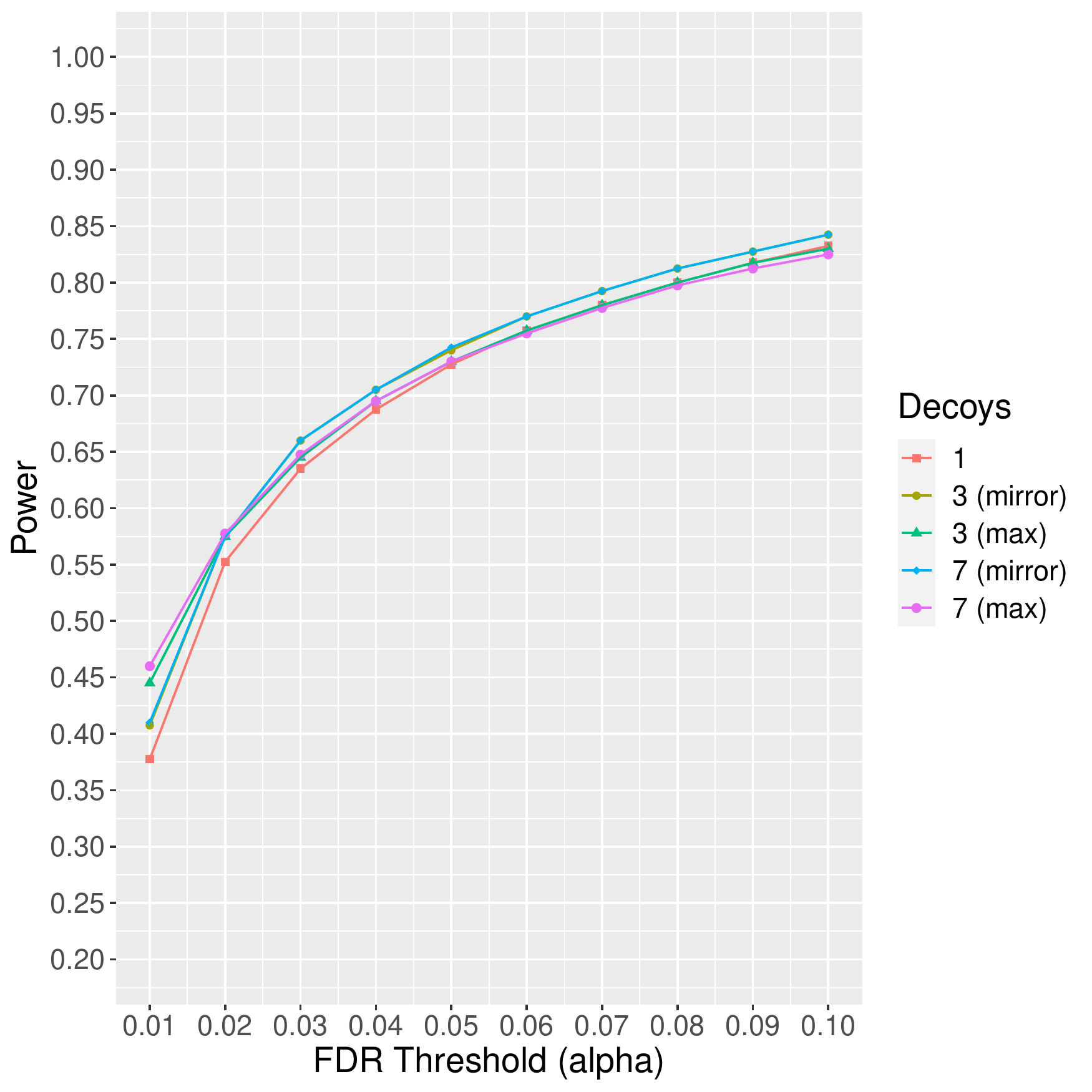}  & \includegraphics[width=2.5in]{{figures_Arya/multiple_decoys_power_figure_m2000_pi0.5_calTRUE_sep3_alpha0.05_conf0.05}.pdf}  & \includegraphics[width=2.5in]{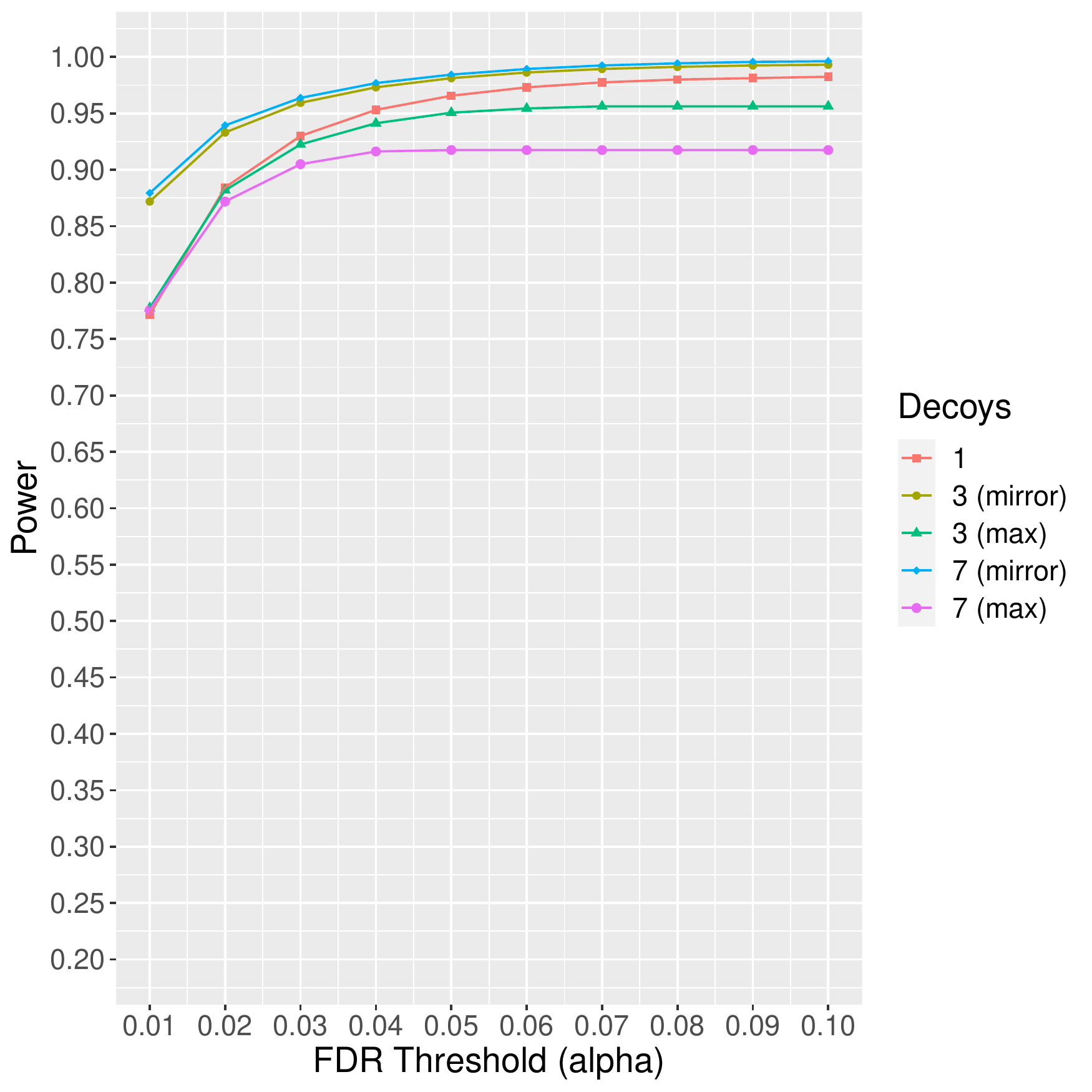} \tabularnewline
\hspace{-10ex}
$m = 2000$, $\pi_0 = 0.5$, $\color{blue}{\rho = 2.5}$  & $m = 2000$, $\pi_0 = 0.5$, $\color{blue}{\rho = 3}$ & $m = 2000$, $\pi_0 = 0.5$, $\color{blue}{\rho = 3.5}$ \tabularnewline
\hspace{-10ex}
\includegraphics[width=2.5in]{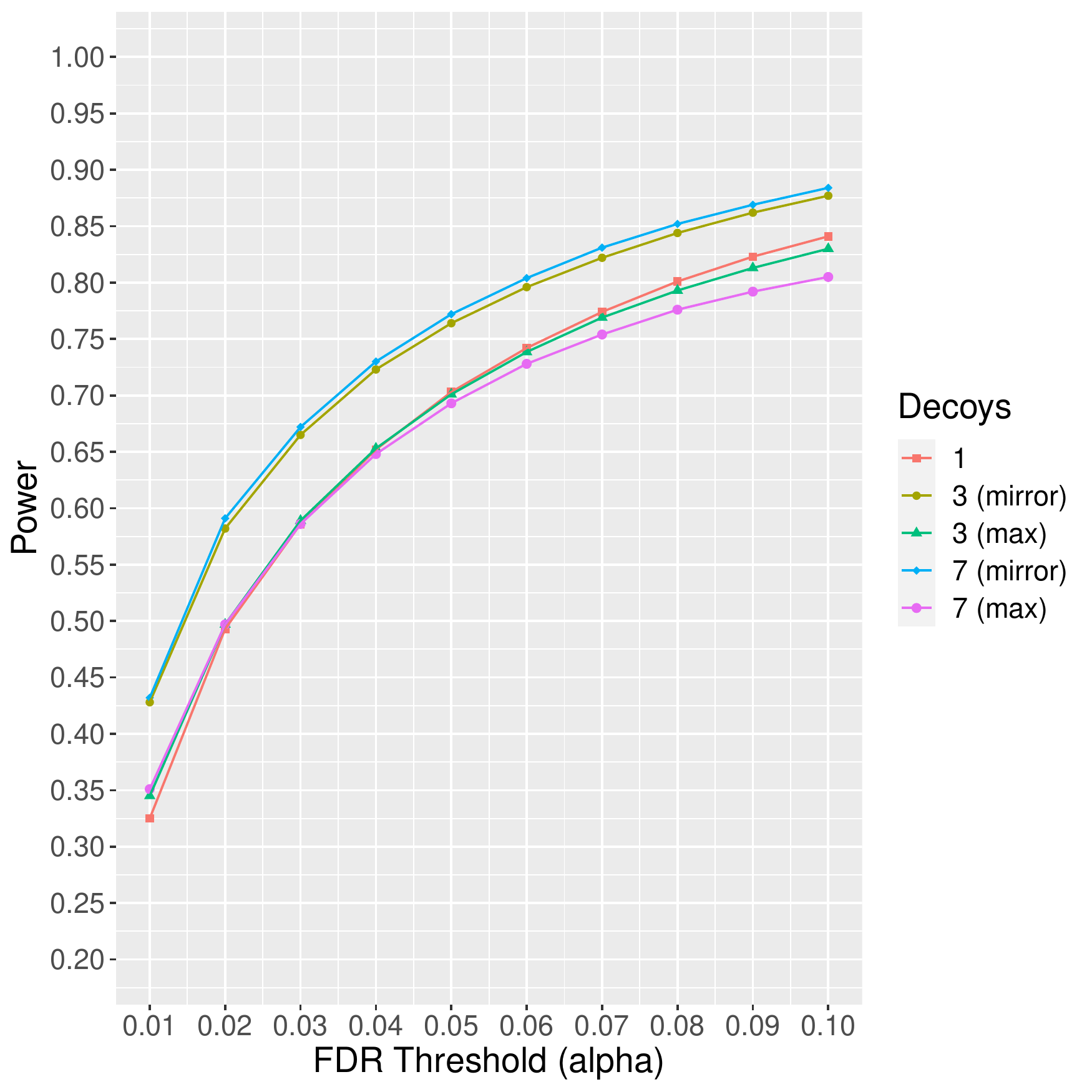}  & \includegraphics[width=2.5in]{{figures_Arya/multiple_decoys_power_figure_m2000_pi0.5_calTRUE_sep3_alpha0.05_conf0.05}.pdf}  & \includegraphics[width=2.5in]{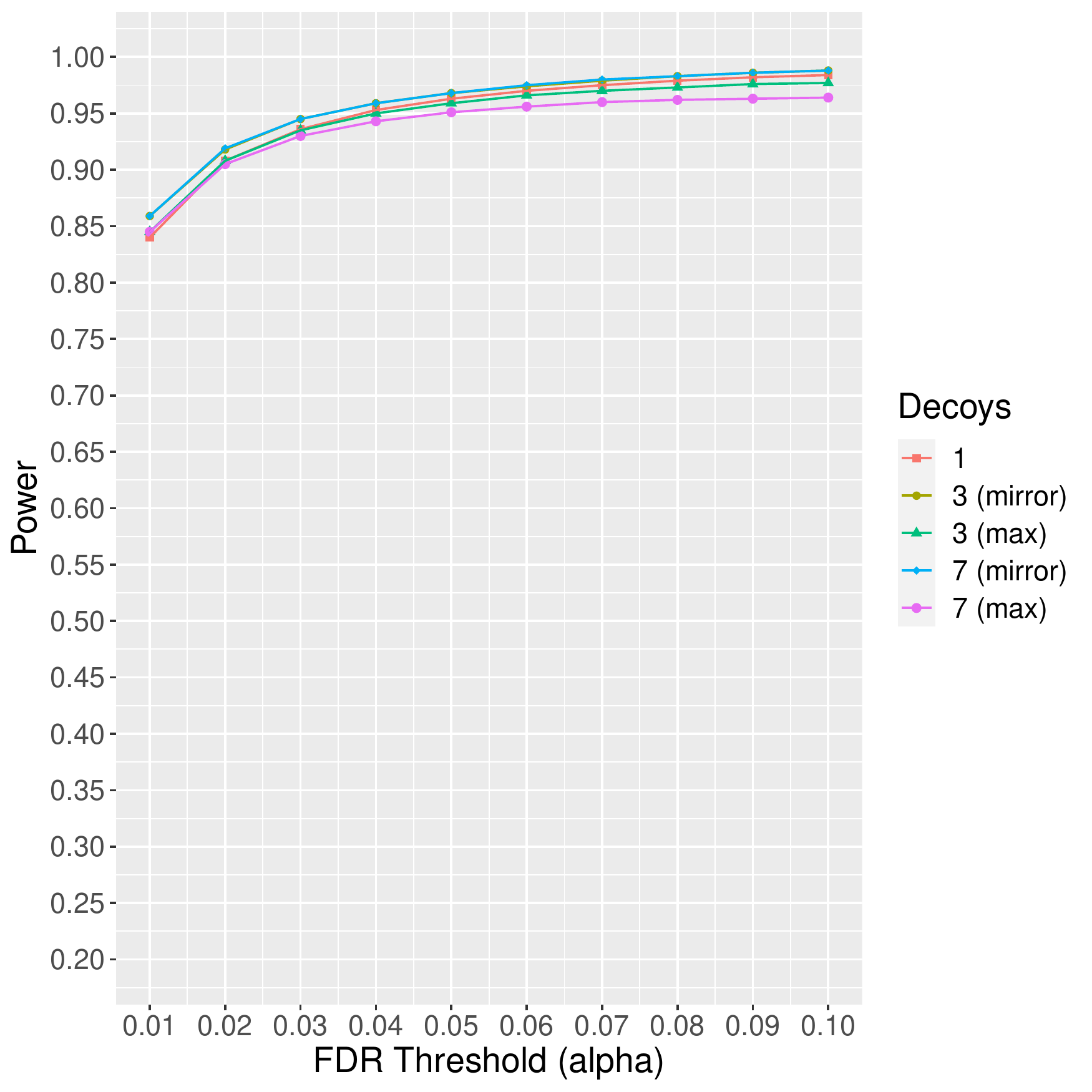} \tabularnewline
\end{tabular}
\caption{\textbf{Comparing the power of TDC using multiple decoys.} Similar to Figure \ref{supfig:mult_method}, except now we plot the power of TDC using 1, 3 and 7 decoys. Note, in particular, that even though Figure \ref{supfig:mult_method} shows that the max method generally offers tighter bounds, it also generally results in the least power across all methods. Moreover, there is a clear increase in power as we increase the number of decoys, using the mirror method. Compare this to Figure \ref{supfig:mult_method}, where the FDP bound is almost unchanged when increasing the number of decoys and using the mirror method.
\label{supfig:mult_power}}
\end{figure}

\clearpage

\end{document}